\documentclass[11pt, preprint]{imsart}
\usepackage{xr}
\usepackage{amsmath,amssymb}
\usepackage{parskip}
\usepackage{marvosym}
\usepackage[hang,small,bf]{caption}
\usepackage{mathtools,bm}
\usepackage[colorlinks]{hyperref}
\hypersetup{
	colorlinks,%
	citecolor=blue,%
	filecolor=black,%
	linkcolor=blue,%
	urlcolor=blue
}
\usepackage{enumerate}
\newcommand{\op}{\mathrm{op}}
\RequirePackage[OT1]{fontenc}

\usepackage[vlined, ruled]{algorithm2e}
\SetKwRepeat{Do}{do}{while}%

\usepackage{amsthm}
\usepackage{amsmath}
\usepackage{amssymb}
\usepackage{thmtools} 
\usepackage{subcaption}
\usepackage{textcomp}

\usepackage{parskip}

\usepackage{graphicx}
\usepackage{verbatim}
\usepackage{array, float}
\usepackage{fontenc}
\usepackage[toc,page]{appendix}
\usepackage[nottoc]{tocbibind}
\usepackage[english]{babel}
\usepackage{marvosym}

\usepackage[pagewise]{lineno}

\usepackage{mathtools}
\allowdisplaybreaks
\usepackage{bbm}
\usepackage[noabbrev,capitalize]{cleveref}

\newtheoremstyle{exampstyle}
{8pt} 
{8pt} 
{\it} 
{} 
{\bfseries} 
{.} 
{.5em} 
{} 

\theoremstyle{exampstyle}
\newtheorem{theorem}{Theorem}

\newtheorem{lemma}{Lemma}

\newtheorem{remark}{Remark}

\newtheorem{proposition}{Proposition}

\newtheorem{defn}{Definition}
\newtheorem{assump}{Assumption}

\numberwithin{equation}{section}
\numberwithin{example}{section}
\numberwithin{theorem}{section}
\numberwithin{lemma}{section}
\numberwithin{corollary}{section}
\numberwithin{proposition}{section}
\numberwithin{defn}{section}
\numberwithin{remark}{section}

\usepackage{tikz}
\tikzset{
	treenode/.style = {shape=rectangle, rounded corners,
		draw, align=center,
		top color=white, bottom color=blue!20},
	root/.style     = {treenode, font=\Large, bottom color=yellow},
	env/.style      = {treenode, font=\ttfamily\normalsize},
	con/.style      = {treenode, font=\ttfamily, bottom color=green!25},
	nocon/.style    = {treenode, font=\ttfamily, bottom color=red!30},
	dummy/.style    = {circle,draw}
}

\startlocaldefs

\newcommand{\eat}[1]{}
\newcommand{\indep}{\perp \!\!\! \perp}

\usepackage{enumitem}

\renewcommand{\hat}[1]{\widehat{#1}}
\renewcommand{\tilde}[1]{\widetilde{#1}}
\renewcommand{\P}{\mathbbm{P}}
\@addtoreset{proofpart}{theorem}
\renewcommand\thesection{\arabic{section}}
\makeatletter
\newcommand*{\rom}[1]{\expandafter\@slowromancap\romannumeral #1@}
\makeatother

\definecolor{LightCyan}{rgb}{0.88,1,1}
\definecolor{Gray}{gray}{0.9}

\newcommand{\HS}{\mathrm{HS}}

\newcommand{\F}{\mathcal{F}}
\newcommand{\V}{\mathcal{V}}
\newcommand{\G}{\mathcal{G}}
\newcommand{\E}{\mathbb{E}}
\newcommand{\p}{\mathbb{P}}
\newcommand{\R}{\mathbb{R}}
\newcommand{\h}{\mathcal{H}}
\newcommand{\X}{\mathcal{X}}
\newcommand{\W}{\mathcal{W}}
\newcommand{\Y}{\mathcal{Y}}
\newcommand{\Z}{\mathcal{Z}}
\newcommand{\Cov}{\mathrm{Cov}}
\newcommand{\Var}{\mathrm{Var}}
\newcommand{\ran}{\mathrm{ran}\, }
\newcommand{\mcg}{\mathcal{G}}
\newcommand{\emgn}{\mathcal{E}(\mathcal{G}_n)}
\newcommand{\tmk}{{\mathcal{M}}_{k}}

\endlocaldefs

\RequirePackage[numbers]{natbib}

\begin{document}

\begin{frontmatter}
\title{Kernel Partial Correlation Coefficient --- a Measure of Conditional Dependence}
\runtitle{Kernel Partial Correlation Coefficient}

    
    
    
    
    

\begin{aug}
\author{\fnms{Zhen} \snm{Huang},}
\author{\fnms{Nabarun} \snm{Deb},}
\and
\author{\fnms{Bodhisattva} \snm{Sen}\thanksref{t3}}
\affiliation{
University of California, Berkeley\thanksmark{a1} and
Columbia University\thanksmark{a2}
}

\thankstext{t3}{Supported by NSF grant DMS-2015376}

\runauthor{Huang, Deb, and Sen}

\address{1255 Amsterdam Avenue \\
New York, NY 10027\\
e-mail: \href{mailto:zh2395@columbia.edu}{\em \ttfamily zh2395@columbia.edu}\\
\phantom{E-mail:}}

\address{1255 Amsterdam Avenue \\
New York, NY 10027\\
e-mail: \href{mailto:nd2560@columbia.edu}{\em \ttfamily nd2560@columbia.edu}\\
\phantom{E-mail:}}

\address{1255 Amsterdam Avenue \\
New York, NY 10027\\
e-mail: \href{mailto:bodhi@stat.columbia.edu}{\em \ttfamily bodhi@stat.columbia.edu}
}
\end{aug}

\begin{abstract}
In this paper we propose and study a class of simple, nonparametric, yet interpretable measures of conditional dependence between two random variables $Y$ and $Z$ given a third variable $X$, all taking values in general topological spaces. The population version of any of these nonparametric measures --- defined using the theory of reproducing kernel Hilbert spaces (RKHSs) --- captures the strength of conditional dependence and it is 0 if and only if $Y$ and $Z$ are conditionally independent given $X$, and 1 if and only if $Y$ is a measurable function of $Z$ and $X$.
Thus, our measure --- which we call \emph{kernel partial correlation} (KPC) coefficient --- can be thought of as a nonparametric generalization of the classical partial correlation coefficient that possesses the above properties when $(X,Y,Z)$ is jointly normal. 
We describe two consistent methods of estimating KPC. Our first method of estimation is graph-based and utilizes the general framework of geometric graphs, including $K$-nearest neighbor graphs and minimum spanning trees. A sub-class of these  estimators can be computed in near linear time and converges at a rate that automatically adapts to the intrinsic dimensionality of the underlying distribution(s). 
Our second strategy involves direct estimation of conditional mean embeddings using cross-covariance operators in the RKHS framework. 
Using these empirical measures we develop forward stepwise (high-dimensional) nonlinear variable selection algorithms. We show that our algorithm, using the graph-based estimator, yields a provably consistent model-free variable selection procedure, even in the high-dimensional regime when the number of covariates grows exponentially with the sample size, under suitable sparsity assumptions. 
Extensive simulation and real-data examples illustrate the superior performance of our methods compared to existing procedures. The recent conditional dependence measure proposed by~\citet{azadkia2019simple} can also be viewed as a special case of our general framework.
\end{abstract}



\end{frontmatter}

\section{Introduction}\label{sec:introduction}

	
Conditional independence is an important concept in modeling causal relationships~\cite{Dawid-79,pearl2000models}, in graphical models~\cite{Lauritzen-96, KF-09}, in economics~\cite{Chiappori00}, and in the literature of program evaluations~\cite{Heckman97}, among other fields. Measuring conditional dependence has many important applications in statistics such as Bayesian network learning \cite{pearl2000models,spirtes1993causation}, variable selection \cite{george2000variable,azadkia2019simple}, dimension reduction \cite{cook2002dimension,fukumizu2004dimensionality,li2018sufficient}, and conditional independence testing \cite{linton1996conditional, bergsma2004testing, su2007consistent,  song2009testing, huang2010testing, zhang2012kernel, su2014testing,  doran2014permutation, wang2015conditional, PatraEtAl-16, runge2018conditional, Ke2019, shah2020hardness}.

Suppose that $(X,Y,Z) \sim P$ where $P$ is supported on a subset of some topological space $\X\times \Y \times \Z$ with marginal distributions $P_{X}, P_Y$ and $P_{Z}$ on $\X,\Y$ and $\Z$ respectively. In this paper, we propose and study a class of simple, nonparametric, yet interpretable measures $\rho^2 \equiv \rho^2(Y,Z|X)$ and their empirical counterparts, that capture the strength of {\it conditional dependence} between ${Y}$ and ${Z}$, given $X$. 

To explain our motivation, consider the case when $\Y=\Z=\R$ and $(X,Y,Z)$ is jointly Gaussian, and suppose that we want to measure the strength of association between $Y$ and $Z$, with the effect of $X$ removed. In this case a well-known measure of this conditional dependence is the \emph{partial correlation} coefficient $\rho_{YZ\cdot X}$. In particular, the partial correlation squared $\rho_{YZ\cdot X}^2$ is: (i) A deterministic number between $[0,1]$; (ii) $0$ if and only if $Y$ is conditionally independent of $Z$ given $X$ (i.e., $Y\indep Z|X$); and (iii) $1$ if and only if $Y$ is a (linear) function of $Z$ given $X$. Moreover, any value between 0 and 1 of $\rho_{YZ\cdot X}^2$ conveys an idea of the strength of the relationship between $Y$ and $Z$ given $X$.

In this paper we answer the following question in the affirmative:  {\it ``Is there a nonparametric generalization of $\rho_{YZ\cdot X}^2$ having the above properties, that is applicable to random variables $X,Y,Z$ taking values in general topological spaces and having any joint distribution?''.}


In particular, we define a generalization of $\rho_{YZ\cdot X}^2$ --- the {\it kernel partial correlation (KPC) coefficient} --- which
measures the strength of the conditional dependence between $Y$ and $Z$ given $X$, that can deal with any distribution $P$ of $(X,Y,Z)$ and is capable of detecting any nonlinear relationships between $Y$ and $Z$ (conditional on $X$). Moreover, given i.i.d.~data from $P$, we develop and study two different strategies to estimate this population quantity --- one that is based on geometric graph-based methods (see~\cref{sec:GraphEst} below) and the other based on kernel methods using cross-covariance operators (see~\cref{sec:KernelEst} below). We conduct a systematic study of the various statistical properties of these two classes of estimators, including their consistency and (automatic) adaptive properties.
A sub-class of the proposed graph-based estimators can even be computed in {\it near linear time} (i.e., in $O(n\log n)$ time). We use these measures to  develop a provably consistent model-free (high-dimensional) variable selection algorithm in regression. 


\subsection{Kernel partial correlation (KPC) coefficient}\label{sec:KPC}
Our measure of conditional dependence between $Y$ and $Z$ given $X$ is defined using the  framework of \emph{reproducing kernel Hilbert spaces} (RKHSs); see~\cref{sec:Prelim} for a brief introduction to this topic. A function $k(\cdot,\cdot):\Y\times\Y\to\R$ is said to be a {\it kernel} if it is symmetric and nonnegative definite. Let $P_{Y|xz}$ denote the regular conditional distribution of $Y$ given $X=x$ and $Z=z$, and $P_{Y|x}$ denote the regular conditional distribution of $Y$ given $X=x$.

We define the \emph{kernel partial correlation} (KPC) coefficient $\rho^2(Y,Z|X)$ (depending on the kernel $k(\cdot,\cdot)$) as:
\begin{equation}\label{eq:eta}
\rho^2 \equiv \rho^2(Y,Z|X) := \frac{\E[{\rm MMD}^2(P_{Y|XZ},P_{Y|X})]}{\E[{\rm MMD}^2({\delta_Y},P_{Y|X})]},
\end{equation}
where MMD is the \emph{maximum mean discrepancy} --- a distance metric between two probability distributions depending on the kernel $k(\cdot,\cdot)$ (see Definition~\ref{defn:MMD}) --- and $\delta_Y$ denotes the Dirac measure at $Y$. In~\cref{lem:Well-Def} we show that $\rho^2$ in~\eqref{eq:eta} is well-defined when $Y$ is not a deterministic function of $X$ and the kernel $k(\cdot,\cdot)$ satisfies some mild regularity conditions.

We show in Theorem~\ref{thm:Eta} that $\rho^2(Y,Z|X)$  satisfies the following three properties for any joint distribution $P$ of $(X,Y,Z)$:
\begin{enumerate}
    \item[(i)] $\rho^2 \in[0,1]$;
    \item[(ii)] $\rho^2=0$ if and only if $Y\indep Z|X$; 
    \item[(iii)] $\rho^2 = 1$ if and only if $Y$ is a measurable function of $Z$ and $X$.
\end{enumerate}
Further, we show that $\rho^2(Y,Z|X)$ monotonically increases as the `dependence' between $Y$ and $Z$ given $X$ becomes stronger: We illustrate this in Propositions~\ref{prop:class_parcor} and~\ref{prop:monotone_lambda} where we consider different kinds of dependence between $Y$ and $Z$ (given $X$). In~\cref{prop:class_parcor}  we show that $\rho^2$, for a large class of kernels $k(\cdot,\cdot)$ and $Y$ and $Z$ scalars, is a monotonic function of $\rho_{YZ\cdot X}^2$ (the squared partial correlation coefficient), for Gaussian data. Moreover, we show that when the linear kernel is used (i.e., $k(y,y') = y y'$ for $y,y' \in \R$), $\rho^2$ reduces exactly to $\rho_{YZ\cdot X}^2$. Thus, our proposed measure KPC is indeed a generalization of the classical partial correlation and captures the strength of conditional association. 

In~\citet{azadkia2019simple} a measure satisfying properties (i)-(iii) was proposed and studied, when $Y$ is a scalar (and $Z$ and $X$ are Euclidean). 
We show in Lemma~\ref{lem:AC19} that this measure is a special case of our general framework by taking a specific choice of the kernel $k(\cdot,\cdot)$. The advantage of our general framework is that we no longer require $Y$ to be a scalar. In fact, $(X,Y,Z)$ can even belong to a non-Euclidean space, as long as a kernel function can be defined on $\Y \times \Y$, and $\X$ and $\Z$ are, for example, metric spaces. Further, $\rho^2$, as defined in~\eqref{eq:eta}, provides a lot of flexibility to the practitioner as there are a number of kernels known in the literature for different spaces $\Y$ (\cite{lafferty2002information, fukumizu2009characteristic, danafar2010characteristic}), some of which may have better properties than others, depending on the application at hand. 


In spite of all the above nice properties of $\rho^2$, it is not immediately clear if we can estimate $\rho^2$ efficiently, given i.i.d.~data. In the following we introduce two methods of estimating $\rho^2$ that are easy to implement and enjoy many desirable statistical properties.




\subsection{Estimation of KPC using geometric graph-based methods}\label{sec:GraphEst}
Suppose that we are given an i.i.d.~sample $\{(X_i,Y_i,Z_i): i =1,\ldots, n\}$ from $P$. Our first estimation strategy is based on {\it geometric graphs}; e.g.,~$K$-nearest neighbor ($K$-NN) graphs and minimum spanning trees (MSTs). It relies on the following alternate expression of $\rho^2(Y,Z|X)$ (see Lemma~\ref{lem:eta-2}):
\begin{eqnarray}
\rho^2 \equiv \rho^2(Y,Z|X) & = & \frac{\E\left[k(Y_2,Y_2')\right] - \E\left[ k(Y_1,Y_1')\right]  }{\E[k(Y,Y)]-\E[k(Y_1,Y_1')]} \nonumber \\
& = & \frac{\E\left[\E[k(Y_2,Y_2')|X,Z]\right] - \E\left[\E[k(Y_1,Y_1')|X]\right]  }{\E[k(Y,Y)]-\E[\E[k(Y_1,Y_1')|X]]}, \label{eq:K-Exp}
\end{eqnarray}
where $Y \sim P_Y$ and $(X,Y_1,Y_1')$ has the following distribution: $X \sim P_X$, and then $Y_1,Y_1'$ are drawn i.i.d.~from the conditional distribution $P_{Y|X}$. Similarly, $(X, Z, Y_2,Y_2')$ has the following distribution: First draw $(X,Z)$ from its joint distribution $P_{XZ}$, and then $Y_2,Y_2'$ are drawn i.i.d.~from the conditional distribution $P_{Y|XZ}$. Thus, while $Y_1$ and $Y_1'$ (respectively $Y_2$ and $Y_2'$) have the same marginal distribution $P_Y$, they are coupled through $X$ (resp.~$X$ and $Z$) and hence dependent.

From the expression of $\rho^2$ in~\eqref{eq:K-Exp} it is clear that 
\begin{equation}\label{eq:2-Terms}
\E\left[\E[k(Y_1,Y_1')|X]\right]\qquad \mbox{ and }\qquad \E\left[\E[k(Y_2,Y_2')|X,Z]\right]
\end{equation} 
are the two non-trivial terms to estimate; note that $\E[k(Y,Y)]$ can be easily estimated by $n^{-1} \sum_{i=1}^n k(Y_i,Y_i)$. In our first estimation strategy, we use {\it geometric graphs} to approximate the above two terms.  

In order to motivate our geometric graph-based estimators let us first consider estimating $\E\left[\E[k(Y_1,Y_1')|X]\right]$. For simplicity, let us assume that $\X$ is a metric space, and we begin by constructing the $1$-\emph{nearest neighbor graph} ($1$-NN graph) of the data points $\{X_i\}_{i=1}^n$, i.e., a graph with vertices $X_1,\ldots ,X_n$ where every vertex $X_i$ shares an edge with its $1$-NN. The $1$-NN graph has the following property which makes it useful for our application --- the node pairs defining the edges represent points that tend to be `close' together (small distance or dissimilarity). For $i \in \{1,\ldots, n\}$, let $X_{N(i)}$ be the NN of $X_i$ and $Y_{N(i)}$ be the corresponding $Y$-value (for the ${N(i)}$-th observation). Then, $k(Y_i,Y_{N(i)})$ can be \emph{informally} viewed as a naive empirical analogue of $\E [k(Y_1,{Y}_1')|X'=X_i]$. Then, we can estimate the first term in~\eqref{eq:2-Terms}  by $$ \frac{1}{n}\sum_{i=1}^n k(Y_i,Y_{N(i)}).$$ 
We can estimate $\E\left[\E[k(Y_2,Y_2')|X,Z]\right]$ similarly, by considering the 1-NN graph on the data points $\{(X_i,Z_i)\}_{i=1}^n$. Combining all this we obtain an empirical estimator of $\rho^2(Y,Z|X)$:
\begin{equation}\label{eq:Rho-Hat-1NN}
\hat{\rho^2} := \frac{\frac{1}{n}\sum_{i=1}^n k(Y_i,Y_{\ddot{N}(i)}) - \frac{1}{n}\sum_{i=1}^n k(Y_i,Y_{N(i)})}{\frac{1}{n}\sum_{i=1}^n k (Y_i,Y_i)- \frac{1}{n}\sum_{i=1}^n k(Y_i,Y_{N(i)})}
\end{equation} 
where, for $i \in \{1,\ldots, n\}$,  $(X_{\ddot{N}(i)},Z_{\ddot{N}(i)})$ is the NN of $(X_i,Z_i)$ and $Y_{\ddot{N}(i)}$ is the corresponding $Y$-value. Note that~\eqref{eq:Rho-Hat-1NN} is just one example from the class of estimators we propose later; see Section~\ref{sec:Est_graph} for the details.

What is quite remarkable about this simple estimator $\hat{\rho^2}$ in~\eqref{eq:Rho-Hat-1NN} is that it is {\it consistent} in estimating $\rho^2$ under very weak conditions; see Theorem~\ref{thm:graph_consistency}. 
In the following we briefly summarize some of the key features of our geometric graph-based estimator $\hat{\rho^2}$:
\begin{enumerate}
	\item It can be computed for continuous and categorical/discrete data on Euclidean domains and also for random variables $X, Y, Z$ taking values in general topological spaces, e.g., $\X$ and $\X\times \Z$ being metric spaces suffices. This can be particularly useful in functional regression~\cite{morris2015functional}, real-life machine learning and human actions recognition \cite{danafar2010characteristic}, dynamical systems \cite{song2009hilbert}, etc. 
		
		\item It has a simple interpretable form, yet it is fully nonparametric. No estimation of conditional densities or distributions (and/or characteristic functions) is involved.
		
		\item It converges to a limit in $[0,1]$ which equals $0$ if and only if $Y$ and $Z$ are conditionally independent given $X$; and equals $1$ if and only if $Y$ is a (noiseless) measurable function of $X$ and $Z$ (see Theorems~\ref{thm:Eta}~and~\ref{thm:graph_consistency}). The consistency of $\hat{\rho^2}$ requires only mild moment assumptions on the kernel (see Theorem~\ref{thm:graph_consistency}); in particular, consistency always holds if a bounded kernel is used.
		
		\item It is $O_p(n^{-1/2})$ concentrated around a population quantity, under mild assumptions on the kernel (see~\cref{prop:concen_moment}). We further establish rates of convergence for $\hat{\rho^2}$ (constructed using $K$-NN graphs, for $K \ge 1$) to $\rho^2$ that adapt to the {\it intrinsic dimensions} of $X$ and $(X,Z)$; see~\cref{thm:conv_rate}.

    \item It can be calculated in almost {\it linear} time (i.e., in $O(n\log n)$ time) for a broad variety of metric spaces $\X$ and $\X\times\Z$ (including all Euclidean spaces). 

    \item It can be used for variable selection in regression where the response and predictors can take values in general topological spaces. In particular, it provides a stopping criterion for a {forward stepwise variable selection algorithm} which we call \emph{kernel feature ordering by conditional independence} (KFOCI) --- inspired by the variable selection algorithm FOCI in~\cite{azadkia2019simple} --- that automatically determines the number of predictor variables to choose; see Section~\ref{sec:application} for the details. We study the properties of KFOCI which is model-free and is provably consistent even in the high-dimensional regime where the number of covariates grows exponentially with the sample size, under suitable sparsity assumptions (see Theorem~\ref{thm:var_select}). By allowing general kernel functions and different geometric graphs, KFOCI can achieve superior performance when compared to FOCI as we illustrate in Section~\ref{sec:simulation}. 

\end{enumerate}

As far as we are aware, our methods are the only procedures that possess all the above mentioned desirable properties.

\subsection{Estimation of KPC using RKHS-based methods}\label{sec:KernelEst}
As the population version of the KPC coefficient $\rho^2$ (as in~\eqref{eq:K-Exp}) is expressed in terms of the kernel function $k(\cdot,\cdot)$, it is natural to ask if kernel methods can be directly used to estimate $\rho^2$. This is precisely what we do in our second estimation strategy. Observe that the MMD between two distributions is the distance between their \emph{kernel mean embeddings} (see~\cref{Defn:Mean-Emb}) in the RKHS; see e.g.,~\cite{muandet2017kernel, Berlinet2004,smola2007hilbert}. Further, the {kernel mean embedding} of a conditional distribution, which is usually called the \emph{conditional mean embedding} (CME; see~\cref{defn:CME}), can be expressed in terms of \emph{cross-covariance operators} (see~\cref{def:CC}) between the two RKHSs; see e.g.,~\cite{baker1973joint, fukumizu2004dimensionality, song2009hilbert, song2013kernel, klebanov2019rigorous}.
As cross-covariance operators can be easily estimated empirically, we can use a plug-in approach to estimate $\rho^2$, and denote it by $\tilde{\rho^2} \equiv \tilde{\rho^2}(Y,Z|X)$. 
We refer to $\tilde{\rho^2}$ as the RKHS-based estimator.

We study this estimation strategy in detail in Section~\ref{sec:Est}. In particular, $\tilde{\rho^2}$ can be computed using simple matrix operations of kernel matrices (see Proposition \ref{lem:eta}). The computation can also be accelerated using incomplete Cholesky decomposition (see Remark~\ref{rk:approx_inchol}). We derive the consistency of this estimator in~\cref{thm:kernel_consistency}. In the process of deriving this result we prove the consistency of the plug-in estimator of the CME (see~\cref{thm:CME_result}); this answers an open question stated in~\cite{klebanov2019rigorous} and may be of independent interest. Furthermore, $\tilde{\rho^2}$ reduces to the empirical (classical) partial correlation squared when the linear kernel is used; see Proposition~\ref{prop:reduce_classi}.
Through extensive simulation studies we show that $\tilde{\rho^2}$ has good finite sample performance in a variety of tasks
(see Section~\ref{sec:simulation}).

A forward stepwise variable selection algorithm, like KFOCI, can also be devised using the RKHS-based estimator $\tilde{\rho^2}$. However, unlike KFOCI, we need to prespecify the number of variables to be chosen beforehand. Both variable selection algorithms --- one based on $\tilde{\rho^2}$ and the other on $\hat{\rho^2}$ --- perform very well in simulations (see~\cref{sec:Simul}) and real data examples (see~\cref{sec:Real-Data}) as illustrated in the thorough finite sample studies in Section~\ref{sec:simulation}.

As a consequence of our general RKHS-based estimation strategy, we can also study the problem of measuring the strength of mutual dependence between $Y$ and $Z$ (when there is no $X$), and estimate it using $\tilde{\rho^2}(Y,Z|\emptyset)$; see~\cref{rem:Kernel-MAc}. This complements the graph-based approach of estimating ${\rho^2}(Y,Z|\emptyset)$ as developed in~\cite{deb2020kernel}.

Besides having superior performance on Euclidean spaces, the two proposed estimators are applicable in much more general spaces.
In Section~\ref{sec:simulation} we illustrate this by considering two such typical examples: One where the response $Y$ takes values in the \emph{special orthogonal group} ${\rm SO}(3)$,  and the other where we have \emph{compositional data} ($Y$ taking values in the simplex $\Y = \{y\in \R^d:y_1+\ldots + y_d = 1,y_i\geq 0\}$). In addition, $\hat{\rho^2}$ and $\tilde{\rho^2}$ can also be easily applied in the existing model-X framework \cite{candes2018panning} to yield valid tests for conditional independence (see~\cref{sec:Testing}) and variable selection algorithms with finite sample FDR control (see~\cref{sec:ModelXFDR}).

\subsection{Related works}

A plethora of procedures --- parametric and nonparametric, applicable to discrete and continuous data --- have been proposed in the literature, over the last 60 years, to detect conditional dependencies between $Y$ and $Z$ given $X$; see e.g.,~\cite{cochran1954some, mantel1959statistical, linton1996conditional, bergsma2004testing, su2008nonparametric, song2009testing, doran2014permutation, candes2018panning, neykov2020minimax} and the references therein. However none of these methods satisfy property (ii) (as mentioned in~\cref{sec:KPC}). While these methods are indeed useful in practice, they have one common problem: They are all designed primarily for testing conditional independence, and not for measuring the strength of conditional dependence between $Y$ and $Z$ given $X$.

In this paper we are interested in nonparametrically measuring the strength of conditional dependence. While there is a rich literature on measuring unconditional dependence between two random variables/vectors (see~\cite{josse2016measuring} for a review),
conditional dependence has been somewhat less explored, especially in the nonparametric setting. Measures of conditional dependence such as $I^{COND}$ (defined as the Hilbert-Schmidt norm of a normalized conditional cross-covariance operator) \cite{fukumizu2008kernel}, HSCIC (defined as the Hilbert-Schmidt independence criterion, HSIC \cite{MR2255909}, between $P_{Y|X}$ and $P_{Z|X}$) \cite{park2020measure}, conditional distance covariance CdCov (defined as the distance covariance \cite{SzekelyCorrDist07} between $P_{Y|X}$ and $P_{Z|X}$) \cite{wang2015conditional} and HS${\rm \ddot{C}}$IC (defined as the Hilbert-Schmidt norm of a conditional cross-covariance operator) \cite{sheng2019distance,fukumizu2004dimensionality} have the property of always being nonnegative, and they attain the value 0 if and only if $Y\indep Z|X$. However, they do not satisfy properties (i) and (iii) (mentioned in~\cref{sec:KPC}). The conditional distance correlation CdCor \cite{wang2015conditional} is normalized between $[0,1]$, but it is not guaranteed to satisfy property (iii). Moreover, HSCIC~\cite{park2020measure} and CdCor~\cite{wang2015conditional} are actually a family of measures indexed by the value of $X$, rather than a single number. The idea of finding a normalized measure of dependence, like ${\rho^2}(Y,Z|X)$, has been explored in the recent paper~\cite{Ke2019}, but their estimation strategy is very different from the two class of estimators proposed here.

The most relevant work to this paper, and the main motivation behind our work, is the recent paper~\cite{azadkia2019simple}, where a measure satisfying properties (i)-(iii) was proposed. However, their measure is only applicable to a scalar $Y$. Our measure KPC provides a general framework for measuring conditional dependence that is flexible enough to allow the user to choose any kernel of their liking and can handle variables taking values in general topological spaces. 

\subsection{Organization}
In Section~\ref{sec:Measure}, we introduce and study the population version of the KPC coefficient $\rho^2(Y,Z|X)$. In Section~\ref{sec:Est_graph}, we describe our first method of estimation, based on geometric graphs such as $K$-NNs and MSTs. Section~\ref{sec:Est} describes our second estimation strategy, using a RKHS-based method. Applications to variable selection, including a consistency theorem for variable selection, are provided in Section~\ref{sec:application}. A detailed simulation study and real data analyses are presented in Section~\ref{sec:simulation}.~\cref{sec:gendis} contains some general discussions that were deferred from the main text of the paper. In~\cref{sec:Proofs} we provide the proofs of our main results, while~\cref{sec:techlem} gives some other auxiliary results.

\section{Kernel partial correlation (KPC)}\label{sec:Measure}
Let $\X,\Y,\Z$ be topological spaces equipped with complete Borel probability measures and let $\X\times \Y\times\Z$ be the completion of the product space. Suppose that $(X,Y,Z) \sim P$ is a random element on $\X \times \Y \times \Z$ with marginal distributions $P_X$, $P_Y$ and $P_Z$, on $\X$, $\Y$ and $\Z$, respectively. Let $P_{XZ}$ denote the joint distribution of $(X,Z)$. Recall the notation $P_{Y|xz}$, $P_{Y|x}$ from the Introduction; we will assume the existence of these regular conditional distributions.

\subsection{Preliminaries}\label{sec:Prelim}
Let $\h_\Y$ be an RKHS with kernel $k(\cdot,\cdot)$ on the space $\Y$. By a kernel function $k:\Y\times\Y \to\mathbb{R}$ we mean a symmetric and nonnegative definite function such that $k(y,\cdot)$ is a (real-valued) measurable function on $\Y$, for all $y\in\Y$. Thus, $\h_\Y$ is a Hilbert space of real-valued functions on $\Y$ such that, for any $f \in \h_\Y$, we have $f(y) = \langle f, k(y, \cdot) \rangle_{\h_\Y}$, for all $y \in \Y$; this is usually referred to as the {\it reproducing property} of the kernel $k(\cdot,\cdot)$. Let us denote the inner product and norm on $\h_\Y$ by $\langle\cdot,\cdot\rangle_{\h_\Y}$ and $\|\cdot\|_{\h_\Y}$ respectively. For an introduction to the theory of RKHS and its applications in statistics we refer the reader to~\cite{BT-A04,SVM}.

In the following we define two concepts that will be crucial in defining the KPC coefficient satisfying properties (i)-(iii) mentioned in the Introduction. 

\begin{defn}[Kernel mean embedding]\label{Defn:Mean-Emb} Suppose that $W$ has a probability distribution  $Q$ on $\Y$ such that $\E_Q[\sqrt{k(W,W)}] < \infty$. There exists a unique $\mu_Q \in \h_\Y$ satisfying
\begin{equation}\label{eq:Embedding}
\langle \mu_Q, f\rangle_{\h_\Y} = \E_Q[f(W)],\qquad \mbox{for all } \,f \in \h_\Y,
\end{equation} 
which is called the (kernel) {\it mean element} of the distribution $Q$ (or the {\it mean embedding} of $Q$ into $\h_\Y$).  
\end{defn}

\begin{defn}[Maximum mean discrepancy]\label{defn:MMD}
The difference between two probability distributions $Q_1$ and $Q_2$ on $\Y$ can then be conveniently measured by $${\rm MMD}(Q_1,Q_2) := \|\mu_{Q_1} - \mu_{Q_2}\|_{\h_\Y}$$ (here $\mu_{Q_i}$ is the mean element of $Q_i$, for $i=1,2$) which is called the {\it maximum mean discrepancy} (MMD) between $Q_1$ and $Q_2$. The following alternative representation of the squared MMD is also known (see e.g.,~\cite[Lemma 6]{Gretton12}):
\begin{equation}\label{eq:MMD}
{\rm MMD}^2(Q_1,Q_2) = \E[k(S,S')] + \E[k(T,T')] - 2 \E[k(S,T)],
\end{equation}
where $S, S' \stackrel{iid}{\sim} Q_1$ and $T, T' \stackrel{iid}{\sim} Q_2$.
\end{defn}
Let $\mathcal{P}$ be the class of all Borel probability distributions on $\Y$.
The kernel mean embedding defines a map from $\mathcal{P}$ to $\h_\Y$ such that $Q \mapsto \mu_Q$, where $Q \in \mathcal{P}$.

\begin{defn}[Characteristic kernel]\label{defn:Char}
The kernel $k(\cdot,\cdot)$ is said to be {\it characteristic} if and only if the kernel mean embedding is injective, i.e., $$\mu_{Q_1} = \mu_{Q_2} \quad  \implies \quad Q_1=Q_2.$$
Note that the last condition is equivalent to $\E_{S \sim Q_1}[f(S)] = \E_{T \sim Q_2}[f(T)]$, for all $f \in \h_\Y$, implies $Q_1=Q_2$; this implicitly assumes that the associated RKHS is rich enough.
\end{defn}

\begin{remark}[Examples of characteristic kernels]\label{rem:excharker}
A number of popular characteristic kernels have been studied in the literature. Some popular ones in $\R^d$ include the Gaussian kernel (i.e., $k(u,v) := \exp(-\sigma\|u - v\|^2)$ with $\sigma>0$) \cite{charac-Bharath}, the Laplace kernel (i.e., $k(u,v) = \exp(-\sigma\|u - v\|_1)$ with $\sigma>0$ where $\|\cdot\|_1$ denotes the $L_1$-norm) \cite{charac-Bharath} and the distance kernel \cite{EquivRKHS13}
\begin{equation}\label{eq:Dist-kernel}
k(u,v):= {2}^{-1}\Big(\lVert u\rVert^{\alpha} +\lVert v\rVert^{\alpha} - \lVert u-v\rVert^{\alpha} \Big), \qquad \mbox{for all } \,u,v \in\R^d,
\end{equation}
for $\alpha \in (0,2)$. See~\cite{fukumizu2008kernel, fukumizu2009characteristic,danafar2010characteristic} for other examples of characteristic kernels on more general topological spaces. Sufficient conditions for a kernel to be characteristic are discussed in~\cite{fukumizu2008kernel,sriperumbudur2008injective,fukumizu2009characteristic,sriperumbudur2010,charac-Bharath,szabo2017characteristic}.
\end{remark}



\subsection{KPC: The population version}\label{sec:Pop}
We are now ready to formally define and study our measure \emph{kernel partial correlation} (KPC) coefficient $\rho^2 \equiv \rho^2(Y,Z|X)$. Let $\h_\Y$ be an RKHS on $\Y$ with kernel $k(\cdot,\cdot)$. Recall the definition of $\rho^2$ from~\eqref{eq:eta}:
$$\rho^2(Y,Z|X) := \frac{\E[{\rm MMD}^2(P_{Y|XZ},P_{Y|X})]}{\E[{\rm MMD}^2({\delta_Y},P_{Y|X})]}, $$ where $\delta_Y$ denotes the Dirac measure at $Y$. 
To study the various properties of $\rho^2$ defined above we will assume the following regularity conditions: 
\begin{assump}\label{assump:characteristic_kernel}
    $k(\cdot,\cdot)$ is characteristic and $\E[k(Y,Y)]<\infty$. 
\end{assump}
\begin{assump}\label{assump:separable_RKHS}
    $\mathcal{H}_\Y$ is separable. 
\end{assump}
\begin{assump}\label{assump:nondegenrate}
    $Y$ is not a measurable function of $X$; equivalently, $Y|X = x$ is not degenerate for almost every (a.e.) $x$. 
\end{assump}

\begin{remark}[{On Assumptions \ref{assump:characteristic_kernel}--\ref{assump:nondegenrate}}]\label{rem:Assumptions-1-3}
Note that $\E[k(Y,Y)]<\infty$ in Assumption \ref{assump:characteristic_kernel} is very common in the kernel literature (see e.g.,~\cite{baker1973joint,fukumizu2007statistical,fukumizu2008kernel,fukumizu2009kernel,park2020measure}). See~\cref{characteristic_kernel} below for a brief explanation as to why a characteristic kernel is necessary; see~\cref{rem:excharker} for some examples of characteristic kernels. Assumption \ref{assump:separable_RKHS} is needed for technical reasons and can be ensured under mild conditions\footnote{For example, if $\Y$ is a separable space and $k(\cdot,\cdot)$ is continuous (see e.g.,~\cite[Lemma 4.33]{SVM}).}; see~\cref{rem:Separable} for a detailed discussion. Assumption \ref{assump:nondegenrate} just ensures that $Y$ is not degenerate given $X$, so that the denominator of $\rho^2(Y,Z|X)$ is not $0$; see~\cref{footnote:degenerate_imply_function} for a proof of this equivalence.
\end{remark}

We will assume that Assumptions \ref{assump:characteristic_kernel}--\ref{assump:nondegenrate} hold throughout the paper, unless otherwise specified. The following lemma (proved in~\cref{sec:Well-Def}) shows that $\rho^2(Y,Z|X)$, as defined in~\eqref{eq:eta}, is well-defined.
\begin{lemma}\label{lem:Well-Def}
Under Assumptions \ref{assump:characteristic_kernel}--\ref{assump:nondegenrate}, $\rho^2(Y,Z|X)$ in~\eqref{eq:eta} is well-defined.
\end{lemma}
The following lemma (proved in~\cref{sec:eta-2}) gives another alternate form for $\rho^2$ which will be especially useful to us while constructing estimators of $\rho^2$.
\begin{lemma}\label{lem:eta-2}
Suppose that Assumptions \ref{assump:characteristic_kernel}--\ref{assump:nondegenrate} hold. Then, we have
\begin{equation}\label{eta_formula}
    \begin{aligned}
        \rho^2(Y,Z|X) &= \frac{\E\left[\E[k(Y_2,Y_2')|X,Z]\right] - \E\left[\E[k(Y_1,Y_1')|X]\right]  }{\E[k(Y,Y)]-\E[\E[k(Y_1,Y_1')|X]]},\\
    \end{aligned}
\end{equation}
where in the denominator $(X,Y_1, Y_1')$ has the following joint distribution: 
\begin{equation}\label{eq:Y_1}
X \sim P_X, \qquad Y_1|X\stackrel{}{\sim} P_{Y|X}, \qquad Y_1'|X\stackrel{}{\sim} P_{Y|X},  \quad \mbox{and} \quad Y_1 \indep Y_1' |X;
\end{equation}
and in the numerator $(X,Y_2, Y_2',Z)$  has the following joint distribution: 
\begin{equation}\label{eq:Y_2}
(X, Z) \sim P_{XZ}, \quad Y_2|X,Z\stackrel{}{\sim} P_{Y|XZ}, \quad Y_2'|X,Z\stackrel{}{\sim} P_{Y|XZ}, \quad \mbox{and} \quad Y_2 \indep Y_2' |X,Z.
\end{equation}
\end{lemma}

Our first main result,~\cref{thm:Eta} (proved in~\cref{sec:Eta}), shows that indeed, under the above assumptions (i.e., Assumptions \ref{assump:characteristic_kernel}--\ref{assump:nondegenrate}), our measure of conditional dependence satisfies the three desirable properties (i)-(iii) mentioned in the Introduction.
\begin{theorem}\label{thm:Eta}
Under Assumptions \ref{assump:characteristic_kernel}--\ref{assump:nondegenrate}, $\rho^2(Y,Z|X)$ in~\eqref{eq:eta} satisfies:
\begin{enumerate}
    \item[(i)] $\rho^2(Y,Z|X)  \in[0,1]$;
    \item[(ii)] $\rho^2(Y,Z|X)=0 \quad$ if and only if $\quad Y$ is {\it conditionally independent} of $Z$  given $X$ (i.e., $Y\indep Z|X$); equivalently, $\rho^2=0$ if and only if  $P_{Y|XZ}=P_{Y|X}$ almost surely (a.s.); 
    \item[(iii)] $\rho^2(Y,Z|X) = 1 \quad$ if and only if $\quad Y$ is a measurable function of $Z$ and $X$ a.s. (equivalently, $Y$ given $X$ and $Z$ is a degenerate random variable a.s.).
\end{enumerate}
\end{theorem}

\begin{remark}[Characteristic kernel]
    \label{characteristic_kernel}
    A close examination of the proof of Theorem~\ref{thm:Eta} (in Section~\ref{sec:Eta}) reveals that $k(\cdot,\cdot)$ being characteristic is used for: (a) Proving $\rho^2 = 0$ implies $Y\indep Z |X$; (b) proving that $\rho^2 = 1$ implies $Y$ is  a function of $X$ and $Z$ (here actually the weaker assumption that the feature map $y\mapsto k(y,\cdot)$ is injective would have sufficed); (c) the denominator of $\rho^2$ is non-zero.
\end{remark}


In the following we provide two results that go beyond~\cref{thm:Eta} and illustrate that KPC indeed measures conditional association --- any value of KPC between 0 and 1 conveys an idea of the strength of the association between $Y$ and $Z$, given $X$. In~\cref{prop:class_parcor} below (proved in~\cref{pf:1}) we show that when the underlying distribution is multivariate Gaussian, $\rho^2$ is a strictly monotonic function of the classical partial correlation coefficient squared $\rho_{YZ\cdot X}^2$ (for a large class of kernels), and equals $\rho_{YZ\cdot X}^2$ if the linear kernel (i.e., $k(u,v) = u^\top v$ for $u, v \in \R^d, d\ge 1$) is used. In the following we will restrict attention to kernels having the following form: 
\begin{equation}\label{eq:h-1-2-3}
k(u,v)=h_1(u) + h_2(v) + h_3(\|u-v\|),\qquad \mbox{ for}\;\;  u,v \in \R^d
\end{equation}
where $h_i:\R^{d}\to\R$ for $i=1,2,3$ are arbitrary functions, and $h_3$ is nonincreasing. Note that the Gaussian and distance kernels in~\cref{rem:excharker} and the linear kernel (and the Laplace kernel when $d=1$) are all of this form.

\begin{proposition}[Monotonicity with classical partial correlation]\label{prop:class_parcor}
    Suppose $(Y,Z,X)$ are jointly normal with $Y$ and $Z$ (given $X$) having (classical) partial correlation $\rho_{YZ\cdot X}$.
    Suppose the kernel $k(\cdot,\cdot)$ has the form in~\eqref{eq:h-1-2-3} where $h_3$ is assumed to be strictly decreasing. Then:
\begin{itemize}
	\item[(a)] $\rho^2(Y,Z|X)$ (in~\eqref{eq:eta}) is a strictly increasing function of $\rho_{YZ\cdot X}^2$, provided ${\rm Var}(Y|X)>0$ is held fixed (as we change the joint distribution of $(Y,Z,X)$). 
	
	\item[(b)] $\rho^2(Y,Z|X)=0$ (resp. 1) if and only if $\rho_{YZ\cdot X}^2=0$ (resp. 1).
	
	\item[(c)] If the distance kernel is used (see~\eqref{eq:Dist-kernel}), $\rho^2(Y,Z|X)$ is a strictly increasing function of $\rho_{YZ\cdot X}^2$, irrespective of the value of ${\rm Var}(Y|X)$.
    
    	\item[(d)] If the linear kernel is used, then $\rho^2(Y,Z|X)=\rho_{YZ\cdot X}^2$.
\end{itemize}
\end{proposition}

Next, we consider the two extreme cases: (i) If $Y= g(X)$, then $Y\indep Z|X$; (ii) if $Y = f(X,Z)$, then $Y$ is a function of $Z$ given $X$ and consequently $\rho^2=1$. The following proposition (proved in~\cref{pf:2}) states that if $Y$ follows a regression model with noise, and we are somewhere in between the above two extreme cases (i) and (ii), expressed as a convex combination, then $\rho^2(Y,Z|X)$ is monotonic in the weight used in the convex combination, regardless of how complicated the dependencies $g(X)$ and $f(X,Z)$ are.


\begin{proposition}[Monotonicity]\label{prop:monotone_lambda}
Consider the following regression model with $\Y = \R$, and arbitrary measurable functions $g:\X \to \Y$ and $f:\X \times \Z \to \Y$, and $\lambda\in[0,1]$:
    $$Y = (1-\lambda)g(X) + \lambda f(X,Z) + \epsilon,$$
where $\epsilon$ is the noise variable (independent of $X$ and $Z$) such that for another independent copy $\epsilon'$, $\epsilon - \epsilon'$ is unimodal\footnote{A random variable with distribution function $F$ is unimodal about $v$ if $F(x)=p\delta_v(x) + (1-p)F_1(x)$ for some $p\in[0,1]$, where $\delta_v$ is the distribution function of the Dirac measure at $v$, and $F_1$ is an absolutely continuous distribution function with density increasing on $(-\infty,v]$ and decreasing on $[v,\infty)$ \cite{purkayastha1998simple}. If $\epsilon$ is unimodal, then $\epsilon - \epsilon'$ is symmetric and unimodal about 0 \cite[Theorem 2.2]{purkayastha1998simple}.}. Then $\rho^2(Y,Z|X)$ is monotone nondecreasing in $\lambda$, when a kernel of the form~\eqref{eq:h-1-2-3} is used to define $\rho^2$.
\end{proposition}


In addition to the properties already described above, $\rho^2$ also possesses important \emph{invariance} and \emph{continuity} properties; see~\cref{sec:Inv-Cont} for a detailed discussion on this.


When $Y$ is a scalar, in~\citet[Equation (2.1)]{azadkia2019simple} a measure $T(Y, Z| X)$, satisfying properties (i)-(iii) of~\cref{thm:Eta} was proposed and studied. More specifically,
$$T(Y,Z|X):=\frac{\int \E({\rm Var}(\p(Y\geq t|Z,X)|X))dP_Y(t)}{\int\E ({\rm Var}(1_{Y\geq t}|X))dP_Y(t)}.$$
Indeed, as we see in the following result (see~\cref{sec:Cha-Stat} for a proof), $T(Y, Z| X)$ can actually be seen as a special case of $\rho^2$, for a suitable choice of the kernel $k(\cdot,\cdot)$.

\begin{lemma}\label{lem:AC19}
$T(Y, Z| X) = \rho^2(Y,Z|X)$ when we use the kernel\footnote{Note that this kernel is not characteristic, but the distance kernel mentioned later in the lemma is characteristic.} $k(y_1,y_2) :=\int 1_{y_1\geq t} 1_{y_2\geq t} d P_{Y}(t)$, for $y_1,y_2 \in \R$. Further, if $Y$ has a continuous cumulative distribution function $F_Y$, then $T(Y, Z| X) = \rho^2(F_Y(Y),Z|X)$ where we consider the distance kernel, as in~\eqref{eq:Dist-kernel} with $\alpha =1$. 
\end{lemma} 
Thus, our measure $\rho^2$ can be thought of as a generalization of $T(Y, Z| X)$ in~\cite{azadkia2019simple}, allowing $Y$ (and $(X,Z)$) to take values in more general spaces. In fact, our framework is more general, as we allow for `any' choice of the kernel $k(\cdot,\cdot)$.

\begin{remark}[Measuring association between $Y$ and $Z$]\label{rem:Ind} Consider the special case when there is no $X$, i.e., 
    $$\begin{aligned}
        \rho^2(Y,Z|X) = \rho^2(Y,Z|\emptyset) = \frac{\mathbb{E}[{\rm MMD}^2(P_{Y|Z},P_{Y})]}{\mathbb{E}[{\rm MMD}^2(P_{\delta_Y},P_{Y})]}.
    \end{aligned}$$
Now, $\rho^2(Y,Z|\emptyset)$ can be used to measure the unconditional dependence between $Y$ and $Z$. This measure $\rho^2(Y,Z|\emptyset)$ has been proposed and studied in detail in the recent paper~\cite{deb2020kernel}; also see \cite[Section 2.4]{Ke2019}. In particular, as illustrated in~\cite{deb2020kernel}, $\rho^2(Y,Z|\emptyset)$ can be effectively estimated using graph-based methods (in near linear time) and can also be readily used to test the hypothesis of mutual independence between $Y$ and $Z$. In Section~\ref{sec:Est} we also develop an RKHS-based estimator of $\rho^2(Y,Z|\emptyset)$.
\end{remark}

Let us now discuss some properties of $\rho^2$ when we use the linear kernel (i.e., $k(u,v)=u^\top v$, for $u,v \in \Y = \R^d$, with $d \ge 1$). Suppose $Y=(Y^{(1)},\ldots,Y^{(d)})^\top \in \Y = \R^d$. Then, from~\eqref{eta_formula}, we have the following expression for $\rho^2(Y,Z|X)$:
\begin{equation}\label{eq:LinKernel}
        \rho^2(Y,Z|X) = \frac{\sum_{i=1}^d \E\left(\Var\left[\E(Y^{(i)}|X,Z)|X \right] \right)}{\sum_{i=1}^d \E [\Var(Y^{(i)}|X)]}.
        \end{equation}
\begin{remark}[Connection to~\citet{jansonfloodgate}]
When $d=1$, the numerator of $\rho^2(Y,Z|X)$ is equal to the minimum mean squared error gap (mMSE gap): $\E \left[(Y-\E[Y|X])^2 \right] - \E \left[(Y-\E[Y|X,Z])^2 \right]$, which has been used to quantify the conditional relationship between $Y$ and $Z$ given $X$ in the recent paper~\cite{jansonfloodgate}. Note that mMSE gap is not invariant under arbitrary scalings of $Y$, but $\rho^2(Y,Z|X)$ (which is equal to the squared partial correlation $\rho^2_{YZ\cdot X}$; see Proposition \ref{prop:class_parcor}) is. Thus,  $\rho^2(Y,Z|X)$ can be viewed as a normalized version of the mMSE gap.
\end{remark}
        
\begin{remark}[{Linear kernel and~\cref{thm:Eta}}] \label{rk:lin_ker_gaussian}
As the linear kernel is not characteristic, $\rho^2$, defined with the linear kernel, does not satisfy all the three properties in~\cref{thm:Eta}. It satisfies (i) and (iii): If $Y$ is a function of $X$ and $Z$, then $\rho^2 = 1$.
Conversely, if $\rho^2 = 1$, then $\E [\Var(Y^{(i)}|X,Z)] = 0$ for all $i$ in~\eqref{eq:LinKernel}, which implies that $Y^{(i)}|X,Z$ is degenerate for all $i$, i.e., $Y$ is a function of $X$ and $Z$. However, it is not guaranteed to satisfy (ii): If $Y\indep Z | X$, then indeed $\rho^2 = 0$; but $\rho^2 = 0$ does not necessarily imply that $Y\indep Z |X$\footnote{ A counter example is: Let $X,Z,U$ be i.i.d. having a continuous distribution with mean 0. Let $Y:= X+ZU$. Then $\rho^2 = 0$ but $Y$ is not independent of $Z$ given $X$.}. This is because the linear kernel is not characteristic. However, if $(Y,X,Z)$ is jointly normal, then $\rho^2 = 0$ does imply $Y\indep Z |X$ (see Lemma \ref{multi_normal}) and  $\rho^2$ satisfies all the three properties in~\cref{thm:Eta}.
\end{remark}

\section{Estimating KPC with geometric graph-based methods}\label{sec:Est_graph}
Suppose that $(X_1,Y_1, Z_1),\ldots, (X_n,Y_n,Z_n)$ are i.i.d.~observations from $P$ on $\X\times \Y\times \Z$. Here we assume that $\X$ and $\X\times \Z$ are metric spaces and we have a kernel function $k(\cdot,\cdot)$ defined on $\Y$. In fact, $\X$ and $\X\times \Z$ can be even more general spaces --- with metrics being replaced by certain semimetrics, ``similarity'' functions or ``divergence'' measures; see e.g.,~\cite{boytsov2013learning,athitsos2004boostmap,miranda2013very,jacobs2000classification,Gottlieb2017}. In this section we propose and study a general framework --- using {\it geometric graphs} --- to estimate $\rho^2(Y,Z|X)$ (as in~\eqref{eq:K-Exp}). We will estimate each of the terms in~\eqref{eq:K-Exp} separately to obtain our final estimator $\hat{\rho^2}$ of  $\rho^2(Y,Z|X)$.  

Note that $\mathbb{E}[k (Y,Y)]$ in~\eqref{eq:K-Exp} can be easily estimated by $\frac{1}{n} \sum_{i=1}^n k(Y_i,Y_i)$. So we will focus on estimating the two other terms --- $\mathbb{E}[\mathbb{E}[k(Y_1,Y_1')|X]]$ and $\mathbb{E}[\mathbb{E}[k(Y_2,Y_2')|X,Z]]$; recall the joint distributions of $(X,Y_1,Y_1')$ and $(X,Z,Y_2,Y_2')$ as mentioned in~\eqref{eq:K-Exp}. As our estimation strategy for both the above terms is similar, let us first focus on estimating $$T := \mathbb{E}[\mathbb{E}[k(Y_1,Y_1')|X]].$$ 
To motivate our estimator of $T$, let us consider a simple case, where $X_i$'s are categorical, i.e., take values in a finite set. A natural estimator for $T$ in that case would be
\begin{equation}\label{eq:targeterm}
	\frac{1}{n}\sum_{i=1}^n \frac{1}{\#\{j:X_j=X_i\}} \sum_{j:X_j=X_i} k(Y_i,Y_j),
\end{equation}
as in \cite[Section 3.1]{Ke2019}. In this section, instead of assuming $X_i$'s are categorical, we will focus on general distributions for $X_i$'s, typically continuous.

 We will use the notion of \emph{geometric graph functionals} on $\mathcal{X}$; see~\cite{deb2020kernel, bhattacharya2019general}. $\mcg$ is said to be a  geometric graph functional on $\mathcal{X}$ if, given any finite subset $S$ of $\mathcal{X}$, $\mcg(S)$ defines a graph with vertex set $S$ and corresponding edge set, say $\mathcal{E}(\mcg(S))$, which is invariant under any permutation of the elements in $S$. The graph can be both directed or undirected, and we will restrict ourselves to simple graphs, i.e., graphs without multiple edges and self loops. Examples of such functionals include minimum spanning trees (MSTs) and $K$-nearest neighbor ($K$-NN) graphs, as described below. Define $\mathcal{G}_n\coloneqq \mathcal{G}(X_1,\ldots ,X_n)$ where $\mathcal{G}$ is some graph functional on $\mathcal{X}$.
\begin{enumerate}
    \item \textbf{$K$-NN graph}: The directed $K$-NN graph $\mathcal{G}_n$ puts an edge from each node $X_i$ to its $K$-NNs among $X_1,\ldots,X_{i-1},X_{i+1},\ldots,X_n$ (so $X_i$ is excluded from the set of its $K$-NNs). Ties will be broken at random if they occur, to ensure the out-degree is always $K$. The undirected $K$-NN graph is obtained by ignoring the edge direction in the directed $K$-NN graph and removing multiple edges if they exist.
    
    \item \textbf{MST}: An MST is a subset of edges of an edge-weighted undirected graph which connects all the vertices with the least possible sum of edge weights and contains no cycles. For instance, in a metric (say $\rho_\X(\cdot,\cdot)$) space, given a set of points $X_1,\ldots,X_n$ one can construct an MST for the complete graph with vertices as $X_i$'s and edge weights $\rho_\X(X_i,X_j)$. 
\end{enumerate}

As is explained in the Introduction, in order to estimate $\mathbb{E}[\mathbb{E}[k(Y_1,Y_1')|X]]$, ideally we would like to have multiple $Y_i$'s (say $Y_i'$'s) from the conditional distribution $P_{Y|X_i}$, so as to average over all such $k(Y_i,Y_i')$. However, this is rarely possible in real data (if $X_1$ is continuous, for example). As a result, our strategy is to find $X_j$'s that are ``close" to $X_i$ and average over all such $k(Y_i,Y_j)$. The notion of geometric graph functionals comes in rather handy in formalizing this notion. The key intuition is to define a graph functional $ \mathcal{G}$ where $X_i$ and $X_j$ are connected (via an edge) in $\mathcal{G}_n \coloneqq \mathcal{G}(X_1,\ldots ,X_n)$ provided they are ``close". Towards this direction, let us define the following statistic (as in~\cite{deb2020kernel}):
\begin{align}\label{eq:statest}
T_n(Y,X) \equiv T_n\coloneqq \frac{1}{n}\sum_{i=1}^n\frac{1}{d_i}\sum_{j:(i,j)\in\emgn} k(Y_i,Y_j),
\end{align}
where $\mathcal{G}_n = \mathcal{G}(X_1,\ldots ,X_n)$ for some graph functional $\mathcal{G}$ on $\mathcal{X}$, and $\emgn$ denotes the set of (directed/undirected) edges of $\mathcal{G}_n$, i.e.,~$(i,j)\in\emgn$ if and only if there is an edge from $i\to j$ if $\mathcal{G}_n$ is a directed graph, or an edge between $i$ and $j$ if $\mathcal{G}_n$ is an undirected graph. Here $d_i:=|\{j:(i,j)\in \emgn \}|$ denotes the degree (or out-degree in a directed graph) of $X_i$ in $\mathcal{G}_n$. Note that when $\mathcal{G}_n$ is undirected, $(i,j) \in \emgn$ if and only if $(j,i) \in \emgn$.

The next natural question is: ``Does $T_n$, as defined in~\eqref{eq:statest}, consistently estimate $T$?". The following result in~\cite[Theorem 3.1]{deb2020kernel} answers this question in the affirmative, under appropriate assumptions on the graph functional. 
\begin{theorem}[{\citet[Theorem 3.1]{deb2020kernel}}]\label{theo:consis}
Suppose $\mathcal{G}_n$ satisfies Assumptions~\ref{assump:conv_nn}--\ref{assump:degupbd} (detailed in~\cref{sec:Assump-Graph}) and $\mathcal{H}_{\Y}$ is separable. For $\theta >0$, let $\tmk^{\theta}(\Y)$ be the collection of all Borel probability measures $Q$ over $\Y$ such that $\E_{Q}\left[k^\theta (Y,Y) \right]<\infty$. If $P_Y\in\tmk^{2+\epsilon}(\Y)$ for some fixed $\epsilon>0$, then $T_n\overset{p}{\to} T$. Further, if $P_Y\in\tmk^{4+\epsilon}(\Y)$ for some fixed $\epsilon>0$, then $T_n\overset{a.s.}{\longrightarrow} T$.
\end{theorem}
Note that Assumptions~\ref{assump:conv_nn}--\ref{assump:degupbd} required on the graph functional $\mathcal{G}_n$ for the above result were made in~\cite[Theorem 3.1]{deb2020kernel}; see~\cite[Section 3]{deb2020kernel} for a detailed discussion on these assumptions.
It can be further shown that for the $K$-NN graph ($K=o(n/\log n)$) and the MST these assumptions are satisfied under mild conditions; see~\cite[Proposition 3.2]{deb2020kernel}.

%

\subsection{Estimation of $\hat{\rho^2}$}\label{sec:Est_graph2}
In the previous subsection we constructed a very general geometric graph-based consistent estimator for $T$, the first term in~\eqref{eq:2-Terms}. We will use a similar strategy to estimate the second term in~\eqref{eq:2-Terms}, namely, $\mathbb{E}[\mathbb{E}[k(Y_2,Y_2')|X,Z]]$, i.e., we define a geometric graph functional on the space $\X \times \Z$ and construct an estimator like $T_n$ (in~\eqref{eq:statest}) but now the geometric graph is defined on $\X \times \Z$ with the data points $\{(X_i,Z_i)\}_{i=1}^n$. 

For simplicity of notation, we let ${\ddot{X}} := (X,Z)$ and $\ddot{\X}:=\X\times \Z$. Let $\mathcal{G}_n^{X}$ (resp.~$\mathcal{G}_n^{\ddot{X}}$) be the graph constructed based on $\{X_i\}_{i=1}^n$ (resp.~$\{\ddot{X}_i \equiv (X_i,Z_i)\}_{i=1}^n$).
Let $d_i^{{X}}$ (resp. $d_i^{\ddot{X}}$) be the degree of $X_i$ (resp. $\ddot{X}_i$) in $\mathcal{G}_n^{{X}}$ (resp. $\mathcal{G}_n^{\ddot{X}}$), for $i=1,\ldots, n$.
We are now ready to define our graph-based estimator of $\rho^2$:
\begin{equation}\label{eq:Rho-Hat}
\hat{\rho^2} \equiv \hat{\rho^2}(Y,Z|X) := \frac{\frac{1}{n}\sum \limits_{i=1}^n \frac{1}{d_i^{\ddot{X}}}\sum \limits_{j:(i,j)\in\mathcal{E}(\mathcal{G}_n^{\ddot{X}})} {k (Y_i,Y_j)} - \frac{1}{n}\sum \limits_{i=1}^n \frac{1}{d_i^{{X}}} \sum \limits_{j:(i,j)\in\mathcal{E}(\mathcal{G}_n^{{X}})} {k (Y_i,Y_j)} }{\frac{1}{n}\sum \limits_{i=1}^n k (Y_i,Y_i)- \frac{1}{n}\sum \limits_{i=1}^n \frac{1}{d_i^{{X}}} \sum \limits_{j:(i,j)\in\mathcal{E}(\mathcal{G}_n^{{X}})} k (Y_i,Y_j) }.
\end{equation}


The estimator $\hat{\rho^2}$ is consistent for $\rho^2$; this follows easily\footnote{By the strong law of large numbers, $\frac{1}{n}\sum_{i=1}^n k (Y_i,Y_i)\overset{a.s.}{\longrightarrow}\E [k (Y,Y)]$. The result now follows from Theorem~\ref{theo:consis} and the continuous mapping theorem.} from Theorem~\ref{theo:consis}. We formalize this in the next result.
\begin{theorem}[Consistency]\label{thm:graph_consistency}
Suppose that Assumptions~\ref{assump:conv_nn}--\ref{assump:degupbd} (see~\cref{sec:Assump-Graph}) hold for both $\mathcal{G}_n^{X}$ and $\mathcal{G}_n^{\ddot{X}}$. If $P_Y\in {\mathcal{M}}_k^{2+\varepsilon}(\mathcal{Y})$ (as in Theorem~\ref{theo:consis}) for some $\varepsilon>0$, then
    $\hat{\rho^2}(Y,Z|X)\overset{p}{\to}{\rho}^2(Y,Z|X).$
    If $P_Y\in {\mathcal{M}}_k^{4+\varepsilon}(\mathcal{Y})$ for some $\varepsilon>0$, then $\hat{\rho^2}(Y,Z|X)\overset{a.s.}{\longrightarrow}{\rho}^2(Y,Z|X).$
\end{theorem}
A salient aspect of~\cref{thm:graph_consistency} is that consistency of $\hat{\rho^2}(Y,Z|X)$ does not need any continuity assumptions of the conditional distributions $P_{Y|x}$ and $P_{Y|xz}$, as $x \in \X$ and $z \in \Z$ vary. Our approach leverages Lusin's Theorem \cite{lusin1912proprietes}, which states that
any measurable function agrees with a continuous function on a ``large" set.
This is a generalization of the technique used in \cite{azadkia2019simple} and has also been used in \cite{chatterjee2020new}.

The following result (see~\cref{sec:pf_concen_moment} for a proof) provides a concentration bound for $T_n$ and states that $\hat{\rho^2}$ is $O_p(n^{-1/2})$-concentrated around a population quantity, if the underlying kernel is bounded.


\begin{proposition}[Concentration]\label{prop:concen_moment}
    Under the same assumptions as in Theorem \ref{thm:graph_consistency} (except Assumption \ref{assump:conv_nn}) on the two graphs $\mathcal{G}_n^{X}$  and $\mathcal{G}_n^{\ddot{X}}$, and provided $\sup_{y\in\Y} k(y,y)\leq M$ for some $M>0$,  there exists a fixed positive constant $C^*$ (free of $n$ and $t$), such that for any $t>0$, the following holds:
		\begin{equation}\label{eq:bdiffcon1}
		\P\left[\bigg|\frac{1}{n}\sum \limits_{i=1}^n \frac{1}{d_i^{{X}}} \sum \limits_{j:(i,j)\in\mathcal{E}(\mathcal{G}_n^{{X}})} {k (Y_i,Y_j)} -\E[k(Y_1,Y_{N(1)})]\bigg|\geq t\right]\leq 2\exp(-C^* nt^2),
		\end{equation}
where ${N}(1)$ is a uniformly sampled index from the neighbors (or out-neighbors in a directed graph) of $X_1$ in $\mathcal{G}_n^{X}$.
A similar sub-Gaussian concentration bound also holds for the term $\frac{1}{n}\sum \limits_{i=1}^n \frac{1}{d_i^{\ddot{X}}}\sum \limits_{j:(i,j)\in\mathcal{E}(\mathcal{G}_n^{\ddot{X}})} {k (Y_i,Y_j)}$, when centered around $\E[k(Y_1,Y_{\ddot{N}(1)})]$; here $\ddot{N}(1)$ is a uniformly sampled neighbor of $(X_1,Z_1)$ in $\mathcal{G}_n^{\ddot{X}}$. Consequently, 
$$\sqrt{n}\left(\hat{\rho^2} -\frac{\E [k(Y_1,Y_{\ddot{N}(1)})] - \E [k(Y_1,Y_{{N}(1)}) ]}{\E [k(Y_1,Y_1)] - \E [k(Y_1,Y_{{N}(1)})]} \right) = O_p(1).$$
\end{proposition}

The above result shows that $\hat{\rho^2}$ has a rate of convergence $n^{-1/2}$ around a limit which is not necessarily $\rho^2$. Further, by~\cref{prop:concen_moment}, it is clear that the rate of convergence of $\hat{\rho^2}$ to ${\rho^2}$ will be chiefly governed by the rates
 at which $\E[k(Y_1,Y_{N(1)})]$ and $\E [k(Y_1,Y_{\ddot{N}(1)})]$ converge to $\mathbb{E}[k(Y_1,Y_1')]$ and $\mathbb{E}[k(Y_2,Y_2')]$ respectively. As it turns out this rate of convergence is heavily dependent on the underlying graph functional $\mathcal{G}$. In~\cref{sec:rate_of_conv} we will focus on $K$-NN graphs and provide an upper bound on this rate of convergence. 

\subsubsection{Computational complexity of $\hat{\rho^2}$} 
Observe that, when using the Euclidean $K$-NN graph, the computation of $\hat{\rho^2}(Y,Z|X)$ takes $O(Kn\log n)$ time. This is because the $K$-NN graph can be found in $O(Kn\log n)$ time (e.g., using the k-d tree~\cite{bentley1975multidimensional}) and the $K$-NN graph has $O(Kn)$
edges\footnote{For directed graph, there are exactly $Kn$ edges; for undirected graph, there are no more than $Kn$ edges, but each edge will be used twice in the summation.}
(thus, we just have to sum over $O(Kn)$ terms in computing each of the two main quantities in~\eqref{eq:Rho-Hat}). The computational complexity of computing Euclidean MSTs is $O(n\log n)$ in $\R^d$ when $d=1$ or 2 (see~\cite{buchin2011delaunay}), and $O(n^{2-2^{-(d+1)}}(\log n)^{1-2^{-(d+1)}})$ when $d\geq 3$ (see e.g.,~\cite{yao1982constructing}). As an MST has just $n-1$ edges, the computational complexity for computing $\hat{\rho^2}$, using the MST, is of the same order as that of finding the MST.

Several authors have proposed tree-based data structures to speed up $K$-NN graph construction. Examples include ball-trees~\cite{omohundro1989five} and cover-trees~\cite{beygelzimer2006cover}. In~\cite{beygelzimer2006cover} the authors study NN graphs in general metric spaces and show that if
the dataset has a bounded expansion constant (which is a measure of its intrinsic dimensionality) the cover-tree data structure can be constructed in $O (n \log n)$ time, for bounded $K$. Thus, in a broad variety of settings, $\hat{\rho^2}$ can be computed in {\it near linear time}. 

\subsection{Rate of convergence of $\hat{\rho^2}$}\label{sec:rate_of_conv}
In this subsection, we will assume that the geometric graph functionals $\mathcal{G}_n^{{X}}$ and $\mathcal{G}_n^{\ddot{X}}$ belong to the family of $K$-NN graphs (directed or undirected) on the spaces $\X$ and $\ddot{\X}$ (assumed to be general metric spaces equipped with metrics $\rho_\X$ and $\rho_{\ddot{\X}}$) respectively. 
Our main result in this subsection,~\cref{thm:conv_rate}, shows that $\hat{\rho^2}$ converges to $\rho^2$, at a rate that depends on the {\it intrinsic dimensions} of $X$ and $\ddot{X}$, as opposed to the ambient dimensions of $\X$ and $\ddot{\X}$. This highlights the {\it adaptive} nature of the estimator $\hat{\rho^2}$.
Let us first define the intrinsic dimensionality of a random variable, which is a relaxation of the \emph{Assouad dimension} \cite[Section 9]{robinson2010dimensions}.
Recall that by a {\it cover} of a subset $A \subset (\X,\rho_\X)$, we mean a collection of subsets of $\X$ whose union contains $A$. Denote by $B(x^*,r) \subset \X$ the closed ball centered at $x^* \in \X$ with radius $r>0$, and by ${\rm supp}(X)$ the support of the random variable $X$.
\begin{defn}[Intrinsic dimension of $X$]\label{defn:intrinsic_dim}
   Let $X$ be a random variable taking values in a metric space $\X$. $X$ is said to have intrinsic dimension at most $d$, with constant $C >0$ at $x^* \in \X$, if for any $t>0$, $B(x^*,t)\cap {\rm supp} (X)$ can be covered with at most $C(t/\varepsilon)^d$ closed balls of radius $\varepsilon$ in $\X$, for any $\varepsilon  \in (0, t]$.
\end{defn}
This notion of the intrinsic dimensionality of $X$ extends the usual notion of dimension of a Euclidean set; see e.g.,~\cite[Definition 3.3.1]{chen2018explaining} and~\cite[Definition 5.1]{deb2020kernel}. For example, any probability measure on $\R^d$ has intrinsic dimension at most $d$. Moreover, if $X$ is supported on a $d_0$-dimensional hyperplane (where $d_0 \le d$), then $X$ has intrinsic dimension at most $d_0$.
Further, if $X$ is supported on a $d_0$-dimensional manifold which is {bi-Lipschitz}\footnote{By \emph{bi-Lipschitz} we mean that there exists $L>0$ such that the bijection $\varphi:\Omega \subset \R^{d_0} \to (\mathcal{M}, \rho_{\mathcal{M}})$, where $\mathcal{M}$ is a manifold with metric $\rho_{\mathcal{M}}$,  satisfies $L^{-1}\|x_1-x_2\|\leq \rho_{\mathcal{M}}(\varphi (x_1),\varphi(x_2))\leq L \|x_1-x_2\|$ for all $x_1,x_2\in\Omega$.} to some $\Omega\subset\R^{d_0}$, then
$X$ has intrinsic dimension at most $d_0$ \cite[Lemma 9.3]{robinson2010dimensions}.
Note that the  intrinsic dimension of $X$ need not be an integer, and can be defined on any metric space.

Let $X_1, \ldots, X_n \stackrel{iid}{\sim} P_X$, and let $\mathcal{G}_n^{{X}}$ be the $K_n$-NN  (here $\{K_n\}_{n \ge 1}$ is assumed to be a fixed sequence). We will assume the following conditions: 
\begin{assump}\label{assump:cont_dist}
   $\rho_\X (X_1,X_2)$ has a continous distribution.
\end{assump}
\begin{assump}\label{assump:intrin_dim}
    $X$ has intrinsic dimension at most $d$ with constant $C_1$ at $x^* \in \X$.
\end{assump}
\begin{assump}\label{assump:knndegbound}
    Let $s_n$ be the number of points having $X_1$ as a $K_n$-NN. Suppose that $\frac{s_n}{K_n} \le C_2$ a.s., for some constant $C_2 >0$ and all $n \ge 1$. 
    \end{assump}
\begin{assump}\label{assump:exp_decay}
    There exist $\alpha,C_3,C_4>0$ such that $\p \left(\rho_\X (X_1,x^*)\geq t \right)\leq C_3\exp(-C_4 t^\alpha)$ for all $t>0$, where $x^* \in \X$ is defined in Assumption \ref{assump:intrin_dim}.
 \end{assump}
\begin{assump}\label{assump:smooth}
   Set $$g(x):=\E[k (Y,\cdot)|X=x] \qquad \quad \mathrm{for} \; x \in \X.$$ There exist $\beta_1\geq 0$, $\beta_2\in(0,1]$, $C_5>0$ such that for any $x_1,x_2,\tilde{x}_1,\tilde{x}_2 \in \X$,
   $$\begin{aligned}
       \left| \langle g(x_1),g(x_2)\rangle_{\h_\Y} - \langle g(\tilde{x}_1),g(\tilde{x}_2)\rangle_{\h_\Y}\right|\leq C_5 \big( 1 + \rho_\X(x^*,x_1)^{\beta_1} + \rho_\X (x^*,x_2)^{\beta_1} \\
       + \rho_\X (x^*,\tilde{x}_1)^{\beta_1} + \rho_\X (x^*,\tilde{x}_2)^{\beta_1} \big)\left(\rho_\X (x_1,\tilde{x}_1)^{\beta_2}+\rho_\X ({x}_2 , \tilde{x}_2)^{\beta_2}\right),
   \end{aligned}$$
 where $x^* \in \X$ is defined in Assumption \ref{assump:intrin_dim}.
\end{assump}

\begin{remark}[On the assumptions]\label{rk:compare_assump}
Assumption \ref{assump:cont_dist} guarantees that the $K_n$-NN graph is uniquely defined. If $\X$ is a Euclidean space, Assumption \ref{assump:knndegbound} is satisfied because the number of points having $X_1$ as a $K$-NN is bounded by $KC(d)$, where $C(d)$ is a constant depending only on the dimension $d$ of $\X$ \cite[Lemma 8.4]{yukich1998Euclidean}.
Assumption~\ref{assump:exp_decay} just says that $X_1$ satisfies a tail decay condition that can be even slower than sub-exponential. Assumption~\ref{assump:smooth} is a technical condition on the smoothness of conditional expectation $g(\cdot)$. Without such an assumption, the rate of convergence of $\hat{\rho^2}$ to $\rho^2$ may be arbitrarily slow~\cite{azadkia2019simple}.
    See \cite[Proposition 5.1]{deb2020kernel} for sufficient conditions under which Assumption \ref{assump:smooth} holds. Note that similar assumptions were also made in \cite{azadkia2019simple}; in fact our assumptions are less stringent in the sense that they allow for: (a) Any general metric space $\X$, (b) tail decay rates of $X$ slower than sub-exponential, and (c) $\beta_2$ to vary in $(0,1]$.
\end{remark}
The following result (see~\cref{sec:pf_conv_rate} for a proof) gives an upper bound on the rate of convergence of $\hat{\rho^2}$. In particular, it shows that, if $d >2$ is an upper bound on the intrinsic dimensions of $X$ and $(X,Z)$ then $\hat{\rho^2}$ converges to $\rho^2$ at the rate $O_p\left(n^{- {\beta_2}/{d}} \right)$, up to a logarithmic factor (provided $K_n$ grows no faster than a power of $\log n$). Note that in certain situations, while the actual dimension of $\X$ (resp. $\ddot{\X}$) may be large, the intrinsic dimensionality of $X$ (resp. $\ddot{X}$) may be much smaller --- the rate of convergence of $\hat{\rho^2}$ automatically adapts to the unknown intrinsic dimensions of $X$ and $\ddot{X}$.

\begin{theorem}[Adaptive rate of convergence]\label{thm:conv_rate}
   Suppose $k (Y_1,Y_1)$ has sub-exponential tail\footnote{i.e., $\p (k(Y_1,Y_1)\geq t)\leq L_1 e^{-L_2t}$ (for all $t >0$), for some $L_1,L_2>0$.}. Let $K_n=o( n\,(\log n)^{-1})$. Suppose that Assumptions~\ref{assump:cont_dist}--\ref{assump:smooth} hold for $(Y,X)$ and also for $(Y,\ddot{X})$ (i.e., by replacing $X$ with $\ddot{X}$ in each of the Assumptions~\ref{assump:cont_dist}--\ref{assump:smooth}) with the same constants $\alpha$ and $\beta_2$ (in Assumptions~\ref{assump:exp_decay} and \ref{assump:smooth}). Define
   $$\nu_{n}:=\begin{cases}\frac{(\log{n})^2}{n}+\frac{K_n}{n}(\log{n})^{2\beta_2/d+2\beta_2/\alpha} & \mbox{if } d<2\beta_2, \\ \frac{(\log{n})^2}{n}+\frac{K_n}{n}(\log{n})^{2+d/\alpha} & \mbox{if } d=2\beta_2, \\ \frac{(\log{n})^2}{n}+\left(\frac{K_n}{n}\right)^{2\beta_2/d}(\log{n})^{2\beta_2/d+2\beta_2/\alpha} & \mbox{if } d>2\beta_2,\end{cases}$$
   where $d$ is the maximum of the intrinsic dimensions of $X$ and $\ddot{X}$ (in Assumption \ref{assump:intrin_dim}). Then $$\hat{\rho^2}(Y,Z|X) = {\rho}^2(Y,Z|X) + O_p\left(\sqrt{\nu_n}\right).$$
\end{theorem}
Although~\cref{thm:conv_rate} has many similarities with~\cite[Theorem 4.1]{azadkia2019simple}, our result is more general on various fronts: (a) We can handle $\X$, $\Y$, $\Z$ taking values in general metric spaces; (b) our upper bound $\nu_n$ depends on the intrinsic dimensions of $X$ and $\ddot{X}$ as opposed to their ambient dimensions; (c) our assumptions are less stringent (as discussed in~\cref{rk:compare_assump}); (d) our upper bound $\nu_n$ is also sharper in the logarithmic factor. A similar result, as~\cref{thm:conv_rate}, can be found in~\cite{deb2020kernel} for estimating $\rho^2(Y,Z|\emptyset)$.


\section{Estimating KPC using RKHS methods}\label{sec:Est}
As the population version of KPC $\rho^2$ (in~\eqref{eq:K-Exp}) is expressed in terms of the kernel function, it is natural to ask if the RKHS framework can be directly used to estimate $\rho^2$. This is precisely what we do in our second estimation strategy. 
We start with some notation. Recall that $(X,Y,Z) \sim P$ on $\X \times \Y \times \Z$ and $k(\cdot,\cdot)$ is a kernel on $\Y \times \Y$. We further assume in this section that $\mathcal{X}$ and $\ddot{\X} = \X\times\Z$ are equipped with separable RKHSs $\mathcal{H}_\mathcal{X}$ and $\mathcal{H}_\mathcal{\ddot{\X}}$ respectively, with kernels $k_\mathcal{X}$ and $k_\mathcal{\ddot{\X}}$.
Let $\ddot{X} = (X,Z) \sim P_{XZ}$. We also assume that $\E[ k_\X (X,X)]<\infty$ and $\E [k_{\ddot{\mathcal{X}}} (\ddot{X},\ddot{X}) ]<\infty$.

In the following we define two concepts --- the {\it (cross)-covariance operator} and {\it conditional mean embedding} --- that will be fundamental in the developments of this section. We direct the reader to~\cref{sec:Func-Ana} for a review of some basic concepts from functional analysis which will be used throughout this section.


\begin{defn}[Cross-covariance operator]\label{def:CC}
The {\it cross-covariance operator} $C_{XY}: \h_\Y \to \h_\X$ is the unique bounded operator that satisfies 
\begin{eqnarray}
\langle g, C_{XY} f\rangle_{\h_\X} & = & \Cov(f(Y), g(X)) \label{eq:Cross-Cov} \\
& = & \E [f(Y) g(X)] - \E [f(Y)] \, \E [g(X)], \quad \mathrm{for \; all }\; f \in \h_\Y, \,g \in \h_\X. \nonumber
\end{eqnarray}
\end{defn}
The existence of the cross-covariance operator follows from the Riesz representation theorem (see e.g.,~\cite[Theorem 1]{fukumizu2004dimensionality}). The {\it covariance operator} of $X$, denoted by $C_{X}$, is obtained when the two RKHSs in Definition~\ref{def:CC} are the same, namely $\h_\X$. Note that the covariance operator $C_{X}$ is bounded, nonnegative, self-adjoint, and trace-class if $\h_\X$ is separable and $\E[k_\X(X,X)] <\infty$ (see~\cref{sec:Func-Ana} and Lemma \ref{trace_class}).

The following explicit representation of the cross-covariance operator $C_{XY}$ will be useful:
\begin{equation}\label{eq:Cov}
C_{XY} = \E[k_\X(\cdot, X) \otimes k(\cdot, Y)] - \E[k_\X(\cdot, X)] \otimes \E[k(\cdot, Y)]
\end{equation}
where we have identified the tensor product space $\mathcal{H}_\X\otimes \mathcal{H}_\Y$ with the space of Hilbert-Schmidt operators from $\mathcal{H}_\Y$ to $\mathcal{H}_\X$, such that
$k_\X(\cdot,u)\otimes k (\cdot,v)(h) := h(v) k_\X(\cdot, u)$, for all $u\in \X,v\in\Y$ and $h \in \h_\Y$. See~\cref{rem:C-Cov} for more details on how \eqref{eq:Cov} is derived.


\begin{defn}[Conditional mean embedding (CME)]\label{defn:CME}
The CME $\mu_{Y|x}\in \h_\Y$, for $x \in \X$, is defined as the kernel mean embedding of the conditional distribution of $Y$ given $X = x$, i.e., $\mu_{Y|x}=\E_{Y\sim P_{Y|x}} [k(Y,\cdot)]$.
\end{defn}

CMEs have proven to be a powerful tool in various machine learning applications, such as dimension reduction~\cite{fukumizu2004dimensionality}, dynamic systems~\cite{song2009hilbert}, hidden Markov models~\cite{song2010hilbert}, and Bayesian inference~\cite{fukumizu2013kernel}; see the recent paper~\cite{klebanov2019rigorous} for a rigorous treatment.

Under certain assumptions, cross-covariance operators can be used to provide simpler expressions of CMEs; see e.g.,~\cite{klebanov2019rigorous}.
The following assumption is crucial for this purpose.
\begin{assump}\label{assump:CME}
    For any $g\in\mathcal{H}_\Y$, there exists $h\in \mathcal{H}_\X$ such that $\E[g(Y)|X=\cdot]-h(\cdot)$ is constant $P_X$-a.e.
\end{assump}

\begin{lemma}[{\cite[Theorem 4.3]{klebanov2019rigorous}}]\label{center_CME} 
Suppose $\E [k_\X(X,X)]<\infty$, $\E [k (Y,Y)]<\infty$, and both $\h_\X,\h_\Y$ are separable.
    Let $C_{XY}$ and $C_X$ be the usual cross-covariance and covariance operators respectively. Further let $\ran C_X$ denote the range of $C_X$ and let $C_X^{\dagger}$ denote the {\it Moore-Penrose inverse} (see e.g.,~\cite[Definition 2.2]{engl1996regularization}) of $C_X$. Suppose further that Assumption \ref{assump:CME} holds.
    Then $\ran C_{XY} \subset \ran C_X$, $C_X^{\dagger}C_{XY}$ is a bounded operator, and for $P_X$-a.e. $x\in\X$,
\begin{equation}\label{eq:CME} 
\mu_{Y|x}=\mu_Y + (C_X^\dagger C_{XY})^* (k (x,\cdot) - \mu_X)
\end{equation}
where for a bounded operator $A$, $A^{*}$ is the {\it adjoint} of $A$, and $\mu_X$ and $\mu_Y$ are the kernel mean embeddings of $P_X$ and $P_Y$ respectively (i.e., $\mu_X = \E[k_\X(X,\cdot)] \in \h_\X$).
\end{lemma}


\begin{remark}
Note that we do not require $k_\X(\cdot,\cdot)$ and $k(\cdot,\cdot)$ to be characteristic for~\eqref{eq:CME} to be valid. See \cite{klebanov2019rigorous} for other sufficient conditions, different from Assumption \ref{assump:CME}, that guarantee~\eqref{eq:CME}.
In particular, if the support of $X$ is finite and $k_\X$ is characteristic, then Assumption \ref{assump:CME} holds (see Lemma \ref{finite_support} for a proof). The CME formula given in~\eqref{eq:CME} is the centered version which uses the centered (cross)-covariance operators $C_X$ and $C_{XY}$. Uncentered covariance operators have also been used to define CMEs in the existing literature (see e.g.,~\cite{song2010hilbert, song2010nonparametric, fukumizu2013kernel}); see~\cref{sec:Un-CME} for a discussion. But it is known that the centered CME formula~\eqref{eq:CME}  requires less restrictive assumptions and hence is preferable; see~\cref{rk:uncenter}.
Our simulation results also validated this observation, and hence in this paper we advocate the use of the centered CME formula (as in~\eqref{eq:CME}).
\end{remark}

The following result (which follows from~\cref{center_CME}) shows that $\rho^2(Y,Z|X)$ can be expressed in terms of CMEs, which in turn, can be explicitly simplified in terms of (cross)-covariance operators under appropriate assumptions (as in~\eqref{eq:CME}). This will form our basis for estimation of $\rho^2(Y,Z|X)$ using the RKHS framework --- we will replace each of the terms in~\eqref{eq:eta-Cen} below with their sample counterparts to obtain the estimator $\tilde{\rho^2}$ (see~\eqref{eq:eta-Cen-Est} below).
\begin{proposition}\label{prop:Kernel-Version}
Suppose that the assumptions in Lemma \ref{center_CME} hold for $(Y,X)$ and $(Y, \ddot{X})$ (i.e. replacing $X$ with $\ddot{X}$). Then $\rho^2(Y,Z|X)$ in~\eqref{eq:eta} can be simplified as
\begin{eqnarray}\label{eq:eta-Cen}
    \rho^2 & = &\frac{ \E \big[\|\mu_{Y|XZ} - \mu_{Y|X}\|_{\h_\Y}^2\big]}{\E [\|k (Y,\cdot)  - \mu_{Y|X}\|_{\h_\Y}^2]} \nonumber \\
    & = & \frac{\E\Big[\big\|(C_{\ddot{X}}^\dagger C_{\ddot{X}Y})^* \big(k_{\ddot{\X}}(\ddot{X},\cdot ) - \mu_{\ddot{X}}\big)- (C_X^\dagger C_{XY})^* \left(k_{{\X}}({X},\cdot)-\mu_{X} \right) \big\|^2_{\mathcal{H}_\Y} \Big] }{\E\Big[\big\|k (Y,\cdot) -\mu_Y - (C_X^\dagger C_{XY})^* \left(k_{{\X}}({X},\cdot)-\mu_{X} \right)\big\|_{\mathcal{H}_\Y}^2 \Big]}.
\end{eqnarray}
Here, for $x \in \X$ and $z \in \Z$, $\mu_{Y|xz}$ is the CME of $Y$ given $X=x$ and $Z=z$.
\end{proposition}

\subsection{Estimation of $\rho^2$ by $\tilde{\rho^2}$}\label{sec:Est-kernel}
Suppose that we have i.i.d.~data $\{(X_i,Y_i,Z_i): i =1,\ldots, n\}$ from $P$ on $\X \times \Y \times \Z$. Let us first consider the estimation of the covariance operator $C_X$. 
The empirical covariance operator $\hat{C}_{X}$ is easily estimated by the sample analogue of~\eqref{eq:Cov}, i.e., by replacing the expectations in~\eqref{eq:Cov} by their empirical counterparts:
\begin{equation}\label{eq:Est-Cov}
    \hat{C}_{X} :=\frac{1}{n}\sum_{i=1}^n k_{\X}(X_i,\cdot)\otimes k_{\X}(X_i,\cdot) - \hat{\mu}_{X} \otimes \hat{\mu}_{X}
\end{equation}
where $$\hat{\mu}_{X}  := \frac{1}{n}\sum_{i=1}^n k_{\X}(X_i,\cdot)$$ is the estimator of the kernel mean embedding $\mu_{X} = \E[k_\X(X,\cdot)] \in \h_\X$. 
Similarly, the cross-covariance operator $C_{YX}$ can be estimated by 
$$\hat{C}_{YX} :=\frac{1}{n}\sum_{i=1}^n k (Y_i,\cdot) \otimes k_\X(X_i,\cdot) - \hat \mu_Y \otimes \hat \mu_X ,$$
where $\hat \mu_Y := \frac{1}{n}\sum_{i=1}^n k (Y_i,\cdot)$ is the estimator of the kernel mean embedding $\mu_{Y} = \E[k (Y,\cdot)] \in \h_\Y$. 
Further note that $\ran {\hat C_X}$ is spanned by the set $\{k_{\X}(X_i,\cdot): i=1,\ldots, n\}$, which implies that $\hat C_X$ is not invertible in general, since $\h_\X$ is typically infinite-dimensional. In fact, estimating the inverse of the compact operator $C_X$ is in general an ill-posed inverse problem \cite[Section 8.6]{manton2014primer}.
Hence the Tikhonov approach is often used for regularization which estimates $C_X^\dagger$ by $(\hat{C}_X + \varepsilon I)^{-1}$, for a tuning parameter $\varepsilon >0$ (e.g.,~\cite{song2010hilbert, song2010nonparametric, fukumizu2013kernel}).
Thus, $(C_X^\dagger C_{XY})^*$ can be estimated by $\big( (\hat{C}_X + \varepsilon I)^{-1}\hat{C}_{XY}\big)^* = \hat{C}_{YX}(\hat{C}_X + \varepsilon I)^{-1}$. $\rho^2$ in~\eqref{eq:eta-Cen} can therefore be naturally estimated empirically by
\begin{equation}\label{eq:eta-Cen-Est}
\tilde{\rho^2} := \frac{\frac{1}{n}\sum_{i=1}^n \|\hat{\mu}_{Y|\ddot{X}_i}  - \hat{\mu}_{Y|{X_i}}\|^2_{\mathcal{H}_\Y} }{\frac{1}{n} \sum_{i=1}^n \|k (Y_i,\cdot)- \hat{\mu}_{Y|{X_i}}\|^2_{\mathcal{H}_\Y}},
\end{equation}
where, for $x \in \X$, and $\ddot{x} \in \ddot{\X}$,
\begin{eqnarray}
\hat{\mu}_{Y|\ddot{x}} & := & \hat \mu_Y  + \hat{C}_{Y\ddot{X}}(\hat{C}_{\ddot{X}} + \varepsilon I)^{-1} \big( k_{\ddot{\X}}(\ddot{x},\cdot ) - \hat{\mu}_{\ddot{X}}\big), \nonumber \\
\hat{\mu}_{Y|{x}} & := & \hat \mu_Y  + \hat{C}_{Y{X}}(\hat{C}_{{X}} + \varepsilon I)^{-1} \big( k_{{\X}}({x},\cdot ) - \hat{\mu}_{{X}}\big) \label{eq:mu_Y|x}
\end{eqnarray}
and $\hat{\mu}_{\ddot{X}} := \frac{1}{n}\sum_{i=1}^n k_{\ddot{\X}}(\ddot{X}_i,\cdot)$ (here ${\ddot{X}}_i \equiv (X_i,Z_i)$, for all $i$). Note that $\tilde{\rho^2}$ is always nonnegative, but it is not guaranteed to be less than or equal to $1$. In practice, since we know that $\rho^2\in[0,1]$, we can always truncate $\tilde{\rho^2}$ at $1$ when it exceeds $1$.
Note that as opposed to the graph-based estimator $\hat{\rho^2}$, $\tilde{\rho^2}$ is always nonrandom. It does not involve tie-breaking for the $K$-NN graph in real data set.

Although the expression for $\tilde{\rho^2}$ in~\eqref{eq:eta-Cen-Est} looks complicated, it can be simplified considerably; see~\cref{lem:eta} below (and~\cref{sec:eta} for a proof). Before we describe the result, let us introduce some notation. We denote by $K_X$, $K_Y$ and $K_{\ddot{X}}$ the $n \times n$ kernel matrices, where for $i, j \in \{1,\ldots, n\}$,  $$(K_X)_{ij}=k_\mathcal{X}(X_i,X_j), \qquad (K_Y)_{ij}=k (Y_i,Y_j), \qquad (K_{\ddot{X}})_{ij}=k_\mathcal{{\ddot{X}}}({\ddot{X}}_i,{\ddot{X}}_j).$$ Let $H :=I-\frac{1}{n}\mathbf{1}\mathbf{1}^\top$ be the centering matrix (here $\mathbf{1} = (1,\ldots, 1) \in \R^n$ and $I \equiv I_n$ denotes the $n \times n$ identity matrix). Then, $$\tilde{K}_X := HK_X H, \qquad \tilde{K}_Y := HK_Y H, \qquad \tilde{K}_{\ddot{X}} := HK_{\ddot{X}} H$$ are the corresponding centered kernel matrices.
\begin{proposition}\label{lem:eta}
Fix $\varepsilon >0$. Let
$$\begin{aligned}
    M &:= \tilde{K}_X ( \tilde{K}_X + n\varepsilon I)^{-1} - \tilde{K}_{\ddot{X}}  ( \tilde{K}_{\ddot{X}} + n\varepsilon I )^{-1}= n\varepsilon \left(( \tilde{K}_{\ddot{X}} + n\varepsilon I )^{-1} - ( \tilde{K}_{X} + n\varepsilon I )^{-1}\right), \\
     N &:=I-\tilde{K}_X ( \tilde{K}_X + n\varepsilon I )^{-1} =  n\varepsilon ( \tilde{K}_X + n\varepsilon I )^{-1}.
\end{aligned}$$ Then, $\tilde{\rho^2}$ in~\eqref{eq:eta-Cen-Est} can be expressed as
\begin{equation}\label{eq:Compute-Rho}
\tilde{\rho^2}=\frac{{\rm Tr}(M^\top \tilde{K}_Y M)}{{\rm Tr}(N^\top  \tilde{K}_Y N)},
\end{equation}
where ${\rm Tr}(\cdot)$ denotes the trace of a matrix. 
\end{proposition}
A few remarks are now in order.

\begin{remark}[Regularization parameter(s)]
In the above result we used the same $\varepsilon >0$ to estimate $C_X^\dagger$ and $C_{\ddot{X}}^\dagger$ by $(\hat{C}_X + \varepsilon I)^{-1}$ and $(\hat{C}_{\ddot{X}} + \varepsilon I)^{-1}$ respectively. We can also use different $\varepsilon$'s and can consider the estimators $(\hat{C}_X+\varepsilon_1 I)^{-1}$ and $(\hat{C}_{\ddot{X}}+\varepsilon_2 I)^{-1}$, for some $\varepsilon_1,\varepsilon_2 >0$.~\cref{lem:eta} still holds in this case where we replace $M$ and $N$ by $\tilde{K}_X \big( \tilde{K}_X + n\varepsilon_1 I\big)^{-1} - \tilde{K}_{\ddot{X}} \big( \tilde{K}_{\ddot{X}} + n\varepsilon_2 I\big)^{-1}$ and $I-\tilde{K}_X \left( \tilde{K}_X + n\varepsilon_1 I\right)^{-1}$ respectively. 
\end{remark}


\begin{remark}[Kernel measure of association]\label{rem:Kernel-MAc}
    If $X$ is not present, then ${\rho}^2(Y,Z|\emptyset)$ yields a measure of association between the two variables $Y$ and $Z$. Our estimation strategy readily yields an empirical estimator of $\rho^2(Y,Z|\emptyset)$, namely,
    \begin{equation}\label{eq:CME_emptyX}
        \tilde{\rho^2}(Y,Z|\emptyset) = \frac{{\rm Tr}(M^\top \tilde{K}_Y M)}{{\rm Tr}(N^\top \tilde{K}_Y N)},
    \end{equation}
    where $M = \tilde{K}_Z \left( \tilde{K}_Z + n\varepsilon I\right)^{-1}$ and $N=I$. This estimator can be viewed as the kernel analogue to the graph-based estimation strategy employed in~\cite{deb2020kernel} to approximate ${\rho}^2(Y,Z|\emptyset)$. 
\end{remark}

\begin{remark}[Uncentered estimates of CMEs]\label{rem:Uncentered-CME}
Instead of using the centered estimates of CMEs, as in~\eqref{eq:CME}, to approximate $\tilde{\rho^2}$ (as in~\eqref{eq:eta-Cen-Est}), one could use their uncentered analogues; see~\cref{sec:Un-CME} for a discussion on this where an explicit expression for the corresponding `uncentered' estimator $\tilde{\rho^2_u}$ of $\rho^2$ is derived. In~\cref{prop:kernel_ridge} (in~\cref{sec:Un-CME}) we further show that $\tilde{\rho^2_u}$ has an interesting connection to kernel ridge regression.
\end{remark}

\begin{remark}[Approximate computation of $\tilde{\rho^2}$]\label{rk:approx_inchol}
The exact computation of $\tilde{\rho^2}$ costs $O(n^3)$ time as we will have to invert $n \times n$ matrices (see \eqref{eq:Compute-Rho}). A fast approximation of $\tilde{\rho^2}$ can be done using the method of \emph{incomplete Cholesky decomposition} \cite{bach2002kernel}. In particular, if we use incomplete Cholesky decomposition of all the three kernel matrices --- $K_X, K_Y$ and $K_{\ddot{X}}$ --- with ranks less than (or equal to) $r$, then the desired approximation of $\tilde{\rho^2}$ can be computed in time $O(nr^2)$; see~\cref{sec:approx_compute} for the details.
\end{remark}

An interesting property of $\tilde{\rho^2}$ is that it reduces to the empirical classical partial correlation squared when linear kernels are used and $\varepsilon\to 0$; this is stated in~\cref{prop:reduce_classi} (see \cref{sec:pf_reduce_classi} for a proof). Thus, $\tilde{\rho^2}$ can indeed be seen as a natural generalization of squared partial correlation.
\begin{proposition}\label{prop:reduce_classi}
    Suppose $\Y =\Z = \R$, $\X = \R^d$, with linear kernels used for all the three spaces $\Y,\X$, and $\ddot{\X}$. If $X$ and $\ddot{X}$ have nonsingular sample covariance matrices, then $\tilde{\rho^2}$ reduces to the classical empirical partial correlation squared as $\varepsilon\to 0$, i.e., 
    $$\lim_{\varepsilon\to 0} \tilde{\rho^2}(Y,Z|X) = \hat{\rho^2}_{YZ\cdot X}.$$
\end{proposition}

\subsection{Consistency results}\label{sec:CME_consistency}
We first state a result that shows the consistency of the CME estimator $\hat{\mu}_{Y|{\cdot}}$ in~\eqref{eq:mu_Y|x}. In particular, we show in~\cref{thm:CME_result} (see~\cref{pf:kernel_consistency} for a proof) that $\hat{\mu}_{Y|{\cdot}}$ is consistent in estimating ${\mu}_{Y|{\cdot}}$ in the averaged $\|\cdot\|_{\h_\Y}^2$-loss. This answers an open question mentioned in~\cite[Section 8]{klebanov2019rigorous} and may be of independent interest. 
\begin{theorem}\label{thm:CME_result}
Suppose the CME formula \eqref{eq:CME} holds, and the regularization parameter $\varepsilon \equiv \varepsilon_n\to 0^+$ (as $n \to \infty$) at a rate slower than $n^{-1/2}$ (i.e., $\varepsilon_n n^{1/2}\to\infty$). Then $$\frac{1}{n} \sum_{i=1}^n \|\hat \mu_{Y|X_i} - \mu_{Y|X_i}\|_{\h_\Y}^2 \stackrel{p}{\to} 0.$$
\end{theorem}

As a consequence of~\cref{thm:CME_result}, our RKHS-based estimator $\tilde{\rho^2}$ (see~\eqref{eq:eta-Cen-Est}) is consistent for estimating $\rho^2$ (as in~\eqref{eq:eta}). This result is formally stated below in~\cref{thm:kernel_consistency} (and proved in~\cref{pf:kernel_consistency}).

\begin{theorem}\label{thm:kernel_consistency}
    Suppose the CME formula \eqref{eq:CME} holds for both $(Y,X)$ and $(Y,\ddot{X})$. Let the regularization parameter $\varepsilon \equiv \varepsilon_n\to 0^+$ (as $n \to \infty$) at a rate slower than $n^{-1/2}$. Then $$\tilde{\rho^2}(Y,Z|X) \overset{p}{\to}\rho^2(Y,Z|X).$$
\end{theorem}


\begin{remark}
From the forms of $\mu_{Y|x}$ (see~\eqref{eq:CME})  and $\hat{\mu}_{Y|{x}}$ (see~\eqref {eq:mu_Y|x}),  one might conjecture whether $\hat{C}_{YX}(\hat{C}_X+\varepsilon_n I)^{-1}$ could converge to $(C_{X}^\dagger C_{XY})^*$ in some sense (e.g., in the Hilbert-Schmidt norm or operator norm), as was explored in \cite{song2010nonparametric}. However, such a convergence is rarely possible. In particular, it does not hold when $\h_{\X}$ is infinite-dimensional\footnote{Consider $Y=X$. Then
$\| (C_{X}^\dagger C_{XY})^* - \hat{C}_{YX}(\hat{C}_X+\varepsilon_n I)^{-1} \|_{\rm op} = \| I - \hat{C}_{X}(\hat{C}_X+\varepsilon_n I)^{-1}\|_{\rm op} = \|\varepsilon_n (\hat{C}_X + \varepsilon_n I)^{-1}\|_{\rm op}\equiv 1$. The last equality follows as: (i) $\hat{C}_X$ is  finite-rank (recall that a finite-rank operator is a bounded linear operator between Banach spaces whose range is finite-dimensional) having at least one zero eigenvalue; (ii) any eigenvalue of $\varepsilon_n (\hat{C}_X + \varepsilon_n I)^{-1}$ has the form $\frac{\varepsilon_n}{\lambda + \varepsilon_n}$, where $\lambda$ is an eigenvalue of $\hat{C}_X$; and (iii) taking $\lambda=0$ yields the desired conclusion. Consequently the convergence of Hilbert-Schmidt norm is also impossible as $\|\cdot\|_{\op}\leq \|\cdot\|_{\HS}$.}.
\end{remark}

\section{Variable selection using KPC}\label{sec:application}

Suppose that we have a regression problem with $p$ predictor variables $X_1,\ldots,X_p$ and a response variable $Y$. Here the response $Y \in \Y$ is allowed to be continuous/categorical, multivariate and even non-Euclidean~\cite{fukumizu2008kernel, fukumizu2009characteristic,danafar2010characteristic,hron2016simplicial,petersen2016functional,petersen2019wasserstein,ramosay2002applied,tsagris2015regression}, as long as a kernel function can be defined on $\Y \times \Y$. Similarly, the predictors $X_1,\ldots,X_p$ could also be non-Euclidean; we just want each $X_i$ to take values in some metric space $\X_i$. In regression the goal is to study the effect of the predictors on the response $Y$. We can postulate the following general model:
\begin{equation}\label{eq:Reg-Mdl}
Y = f(X_1,\ldots, X_p, \epsilon)
\end{equation}
where $\epsilon$ (the unobserved error) is independent of $(X_1,\ldots, X_p)$ and $f$ is an unknown function. 

The problem of {\it variable selection} is to select a subset of predictive variables from $X_1,\ldots,X_p$ to explain the response $Y$ in the simplest possible way,
as suggested by the principle of Occam's razor: ``Never posit pluralities without necessity" \cite{schaffer2015not}.
For $S\subset \{1,\ldots,p\}$, let us write $X_S := (X_j)_{j\in S}$. We assume that $X_S$ takes values in the metric space $\X_S$ (for which a natural choice would be the product metric space induced by $\{\X_i:i\in S\}$). Our goal is to find an $S\subset \{1,\ldots,p\}$ such that \begin{equation}\label{eq:Suff}
Y\indep X_{S^c}|X_S.
\end{equation}
Such an $S$ (satisfying~\eqref{eq:Suff}) is called a \emph{sufficient subset} \cite{vergara2014review,azadkia2019simple}. Ideally, we would want to select a sufficient subset $S$ that has the smallest cardinality,
so that we can write $Y= \E [f(X_1,\ldots,X_p, \epsilon)|X_S,\epsilon] =:  g(X_S, \epsilon)$ in~\eqref{eq:Reg-Mdl}. 

Common variable selection methods in statistics often posit a parametric model, e.g., by assuming a linear model~\cite{Breiman95,barber2015controlling,CD94, LAR04, Friedman91,ELS-2, Miller02, Lasso,Candes-Tao-07, Fan-Li-01, Ravi-09, Yuan-Lin-06, Zou-05, Zou-06}. These methods are powerful
when the underlying  parametric assumption holds true, but could have poor performance when the data generating process is more complicated.
In general nonlinear settings, although there are popular algorithms for feature selection based on machine learning methods, such as random forests and neural nets \cite{speiser2019comparison,ELS-2,amit1997shape,battiti1994using,breiman2001random,breimancj,ho1998random,vergara2014review},
the performance of these algorithms could depend heavily on how well the machine learning techniques fit the data, and often their theoretical guarantees are weaker than those of the model-based methods. The recent paper~\cite{azadkia2019simple} attempted to balance both these aspects by proposing a fully model-free forward stepwise variable selection algorithm, and formally proving the consistency of the procedure, under suitable assumptions. 

The main idea is to use our proposed KPC $\rho^2$ to detect conditional independence in~\eqref{eq:Suff} (note that $\rho^2$ is 0 if and only if conditional independence holds). 
In the following two subsections we propose two model-free variable selection algorithms --- one based on our graph-based estimator $\hat{\rho^2}$ and the other on the RKHS-based framework $\tilde{\rho^2}$. Our procedures do not make any parametric model assumptions, are easily implementable and have strong theoretical guarantees. They provide a more general framework for feature selection (when compared to~\cite{azadkia2019simple}) that can handle any kernel function $k(\cdot, \cdot)$. This flexibility indeed yields more powerful variable selection algorithms, having better finite sample performance; see~\cref{subsec:var_select} for the details. 




\subsection{Variable selection with graph-based estimator (KFOCI)}\label{subsec:KFOCI}
We introduce below the algorithm \emph{Kernel Feature Ordering by Conditional Independence} (KFOCI) which is a model-free forward stepwise variable selection algorithm (for  regression). The proposed algorithm has an automatic stopping criterion and yields a provably consistently variable selection method (i.e., it selects a sufficient subset of predictors with high probability) even in the high-dimensional regime under suitable assumptions; see~\cref{thm:var_select} below. KFOCI can be viewed as a generalization of the FOCI algorithm proposed in \cite{azadkia2019simple}; it provides a more general framework for feature selection that can handle any kernel function $k(\cdot, \cdot)$ and any geometric graph (including $K$-NN graphs for any $K\ge 1$).  Further, as mentioned before, the response $Y$ and the predictor variables $X_i$'s can take values in any topological space (e.g., metric space).

Let us describe the algorithm KFOCI. Suppose that predictors $X_{j_1},\cdots,X_{j_k}$, for $k \ge 0$, have already been selected by KFOCI. Quite naturally, we would like to find $X_{j_{k+1}}$ that maximizes $\rho^2(Y,X_{j_{k+1}}|X_{j_1},\ldots,X_{j_k})$.
To this end define 
\begin{equation}\label{eq:T_KFOCI}
T(S):=\mathbb{E}[\mathbb{E}[k(Y,Y')|X_S]],
\end{equation}
where we first draw $X_S$, and then draw $Y,Y'$ i.i.d.~from the conditional distribution of $Y$ given $X_S$. A closer look of the expression of $\rho^2$ in \eqref{eq:K-Exp} reveals that finding $X_{j_{k+1}}$ that maximizes $\rho^2(Y,X_{j_{k+1}}|X_{j_1},\cdots,X_{j_k})$ is equivalent to finding $X_{j_{k+1}}$ that maximizes $T(\{X_1,\ldots,X_{j_{k+1}}\})$ in~\eqref{eq:T_KFOCI}, for $j_{k+1} \in \{1,\ldots, p\} \setminus \{j_1,\ldots, j_k\}$.
Note that $T(S)$ satisfies $$T(S')\geq T(S) \qquad \mbox{whenever} \qquad S'\supset S,$$ since the numerator of $\rho^2 (Y,X_{S'\backslash S}|X_S)$ (as in~\eqref{eq:K-Exp}) is always greater than (or equal to) 0.
If $S_0$ is a \emph{sufficient subset}, then $T(S_0)=T(\{1,\ldots,p\})\geq T(S)$, for all $S\subset \{1,\ldots,p\}$. Therefore, $T(S)$ can be viewed as measuring the importance of $S$ in predicting $Y$. 

For our implementation, we propose the use of the estimator $T_n$ (as in~\eqref{eq:statest}) instead of the unknown $T$ (in~\eqref{eq:T_KFOCI}). Note that Theorem~\ref{theo:consis} shows that $T_n(S) \equiv T_n(Y,X_S)$ is a consistent estimator of $T(S)$, for every $S \subset \{1,\ldots, p\}$. What is even more interesting is that the use of $T_n(\cdot)$ automatically yields a stopping rule --- we stop our algorithm when adding any variable does not increase our objective function, i.e., $T_n(\hat S \cup \{\ell\}) < T_n(\hat S)$, for any $\ell \in \{1,\ldots, p\} \setminus \hat S$ where $\hat S$ is the current sufficient subset. Algorithm~\ref{algo:graph} gives the pseudocode.


\begin{algorithm}[H]
    \label{algo:graph}
\SetKwInOut{Input}{Input}
\SetKwInOut{Output}{Output}
\SetAlgoLined
\KwData{$(Y_i,X_{1i},\ldots, X_{pi})$, \hspace{0.3in} for $i=1,\ldots, n$}
 {\bf Initialization}: $k=-1,\, \hat S \leftarrow \emptyset,\, \{j_0\} =\emptyset, \, T_n(\emptyset) = -\infty$\;
 \Do{$T_n(\hat S \cup \{j_{k+1}\}) \ge T_n(\hat S)\;\;\;  \mathrm{and}  \;\;\;k < p$}{
 \begin{enumerate}
\item[\bf 1.]  $k  \leftarrow k + 1$\;

\item[\bf 2.]  $\hat S \leftarrow \hat S \cup \{j_k\}$\;

\item[\bf 3.]  Choose 
$j_{k+1}\in \{1,\ldots, p\} \setminus \hat S$ such that $T_n(\hat S \cup \{j_{k+1}\})$ is {\it maximized}, i.e.,
\begin{equation*}\label{eq:Algo}
{j_{k+1}} := \arg \max_{\ell \in \{1,\ldots, p\} \setminus \hat S} T_n(\{j_1,\ldots,j_k, \ell\});
\end{equation*}

\end{enumerate}
}
\Output{$\hat S$} 
 \caption{KFOCI --- a forward stepwise variable selection algorithm}
\end{algorithm}

%
%

In Algorithm~\ref{algo:graph}, at each step, we are actually  selecting $X_{j_{k+1}}$ to maximize $\hat{\rho^2}(Y,X_{j_{k+1}}|X_{j_1},\ldots,X_{j_k})$. Note that the stopping criterion in Algorithm~\ref{algo:graph} corresponds to the case when $\hat{\rho^2}(Y,X_{j_{k+1}}|X_{j_1},\cdots,X_{j_k}) < 0$
    for all $j_{k+1}\in \{1,\ldots,p\}\backslash \{j_1,\ldots,j_k\}$. 
The following result (see Section~\ref{pf:var_select} for a proof) shows the variable selection consistency of KFOCI. 
\begin{theorem}\label{thm:var_select}
    Suppose the following assumptions hold:
    \begin{enumerate}
            \item[(a)] There exists $\delta>0$ such that for any insufficient subset $S \subset \{1,\ldots, p\}$, there is some $j$ such that 
         $T (S\cup \{j\})\geq T(S) + \delta$.
        
        \item[(b)] Suppose that kernel satisfies $\sup_{y\in\Y} k (y,y)\leq M <\infty$. Let $\kappa :=  \lfloor \frac{M}{\delta} + 1 \rfloor$. 
        
        \item[(c)] Suppose that the $K_n$-NN graph is used as the geometric graph in~\eqref{eq:statest} (with $K_n\leq C_6(\log n)^\gamma$ for some $C_6 >0$, $\gamma\geq 0$). For every $S \subset \{1,\ldots, p\}$ of size less than or equal to $\kappa$, we suppose that Assumptions~\ref{assump:cont_dist}-\ref{assump:smooth} hold with $X$ replaced by $X_S$, with a uniform upper bound of intrinsic dimension $d$ in Assumption~\ref{assump:intrin_dim} and the same constants $\{C_i\}_{i=1}^5,\alpha,\beta_1,\beta_2$ in Assumptions~\ref{assump:intrin_dim}-\ref{assump:smooth}.
    \end{enumerate}
    Then there 
    exist $L_1,L_2>0$ depending only on $\alpha,\beta_1,\beta_2,\gamma,\{C_i\}_{i=1}^6,d,M,\delta$ such that
    $$\p(\hat{S}\ {\rm is\ sufficient}) \geq 1-L_1p^{\kappa}e^{-L_2n}.$$
\end{theorem}

Suppose the above algorithm selects $\hat{S}$. Then~\cref{thm:var_select} shows that if $n\gg \log p$, then the algorithm selects a sufficient subset with high probability. In particular, in the low dimensional setting (where $p$ is fixed), the algorithm selects a sufficient subset with probability that goes to 1 exponentially fast.~\cref{thm:var_select} is in the same spirit as \cite[Theorem 6.1]{azadkia2019simple}, but allows for the predictors and response variable to be metric-space valued, and offers the flexibility of using any kernel and a general $K$-NN graph functional (note that the FOCI algorithm in \cite{azadkia2019simple} used the 1-NN graph). 
This flexibility leads to better finite sample performance for KFOCI, when compared with FOCI~\cite{azadkia2019simple}. Even in parametric settings, our performance is comparable to (and sometimes even better than) classical methods such as the Lasso~\cite{Lasso} and the Dantzig selector~\cite{Candes-Tao-07}; see~\cref{subsec:var_select} for the detailed simulation studies.

\begin{remark}[On our assumptions]
Condition (a) is essentially a sparsity assumption, which is also assumed in \cite[Theorem 6.1]{azadkia2019simple}.
Note that as the kernel $k(\cdot,\cdot)$ is bounded by $M$ (by Assumption (b)), $T(S)=\mathbb{E}[\mathbb{E}[k(Y,Y')|X_S]]$ is also bounded by $M$.
    Thus (a) implies that there exists a sufficient subset of size less than $\lfloor \frac{M}{\delta} + 1\rfloor$.
Another implication of (a) is a lower bound assumption on the signal strength of important predictors, indicating that we can expect an improvement of $\delta >0$ in terms of $T(\cdot)$ for each iteration of KFOCI. Condition (c) can be viewed as a uniform version of Assumptions \ref{assump:cont_dist}-\ref{assump:smooth}, in the sense that those assumptions need to hold with the same constants uniformly over all subsets of $\{1,\ldots, p\}$ with cardinality no larger than $\kappa$.
Similar to Remark~\ref{rk:compare_assump} Assumption (c) is less stringent when compared to the assumptions made in \cite[Theorem 6.1]{azadkia2019simple} in the sense that it allows for: (i) any general metric space $\X_S$ (recall $X_S$ takes values in $\X_S$), (ii) tail decay rates of $X_S$ slower than sub-exponential, and (iii) $\beta_2$ to vary in $(0,1]$.
\end{remark}

\subsection{Variable selection using the RKHS-based estimator}
As in~\cref{subsec:KFOCI}, we can also develop a similar model-free forward stepwise variable selection algorithm using the RKHS-based estimator $\tilde{\rho^2}$ (instead of $\hat{\rho^2}$). Algorithm~\ref{algo:CME} gives the pseudocode of the proposed procedure. As $\tilde{\rho^2}$ is always nonnegative (see~\eqref{eq:eta-Cen-Est} and~\eqref{eq:Compute-Rho}) one can no longer specify an automatic stopping criterion as in Algorithm~\ref{algo:graph}. Thus, in Algorithm~\ref{algo:CME} we have to prespecify the number of variables $p_0 \le p$ to be chosen a priori. Note that to use Algorithm~\ref{algo:CME} one must also specify a kernel function $k_S$, for every $S\subset \{1,\ldots, p\}$ with cardinality $|S| \le p_0$.
For the most common case where all the $X_i$'s are real-valued and are suitably normalized (e.g., all $X_i$'s have mean 0 and variance 1), an automatic choice of $k_S$ can be the Gaussian kernel with empirically chosen bandwidth\footnote{Example: For $x,x' \in \R^{|S|}$, $k_S(x,x') :=\exp\left(\|x-x'\|^2_{\R^{|S|}}/(2s^2)\right)$, where $s$ is the median of pairwise distances $\{\|(X_S)_i-(X_S)_j\|_{\R^{|S|}}\}_{i<j}$.}.
In our numerical studies we see that Algorithm~\ref{algo:CME} can also detect complex nonlinear conditional dependencies  and has good finite sample performance; see~\cref{subsec:var_select} for the details.
\begin{algorithm}[H]
    \label{algo:CME}
\SetKwInput{KwInput}{Input}
\SetKwInOut{Output}{Output}
\SetAlgoLined
\KwInput{The number of variables $p_0\leq p$ to select; a kernel function $k(\cdot,\cdot)$ on $\Y \times \Y$;
a kernel $k_S$ for each $X_S$, $S \subset \{1,\ldots, p\}$ with $|S| \le p_0$,
and regularization parameter $\varepsilon>0$ for computing $\tilde{\rho^2}$.}
\KwData{$(Y_i,X_{1i},\ldots, X_{pi})$, \hspace{0.3in} for $i=1,\ldots, n$}
 {\bf Initialization}: $k=-1,\, \tilde S \leftarrow \emptyset,\, \{j_0\} =\emptyset$\;
 \Do{$k < p_0$}{
 \begin{enumerate}
\item[\bf 1.]  $k  \leftarrow k + 1$\;

\item[\bf 2.]  $\tilde S \leftarrow \tilde S \cup \{j_k\}$\;

\item[\bf 3.]  Choose the next
$j_{k+1}\in \{1,\ldots, p\} \setminus \tilde S$ such that $\tilde{\rho^2}(Y,X_{j_{k+1}}|\tilde{S})$ is {\it maximized}, i.e.,
\begin{equation*}\label{eq:Algo-2}
{j_{k+1}} := \arg \max_{\ell \in \{1,\ldots, p\} \setminus \tilde S} \tilde{\rho^2}(Y,X_{j_{k+1}}|X_{j_1},\ldots,X_{j_k});
\end{equation*}

\end{enumerate}
}
\Output{$\tilde S$} 
 \caption{Forward stepwise variable selection algorithm using $\tilde \rho^2$}
\end{algorithm}

\section{Finite sample performance of our methods}\label{sec:simulation}
In this section, we report the finite sample performance of $\hat{\rho^2}$, $\tilde{\rho^2}$ and the related variable selection algorithms, on both simulated and real data examples. We consider both Euclidean and non-Euclidean responses $Y$ in our examples. Even when restricted to Euclidean settings, our algorithms achieve superior performance compared to existing methods. All our finite sample experiments are reproducible using our \texttt{R} package \texttt{KPC}\footnote{\url{https://github.com/zh2395/KPC}}.
\subsection{Examples with simulated data}\label{sec:Simul}
\subsubsection{Consistency of $\hat{\rho^2}$ and $\tilde{\rho^2}$}\label{sec:Simul-Consis}
Here we examine the consistency of our two empirical estimators $\hat{\rho^2}$ and $\tilde{\rho^2}$.
As has been mentioned earlier, the consistency of $\hat{\rho^2}$ requires only very weak moment assumptions on the kernel (see~\cref{thm:graph_consistency}), whereas the consistency of $\tilde{\rho^2}$ depends on the validity of the CME formula (in~\eqref{eq:CME}) which in turn depends on the hard to verify Assumption~\ref{assump:CME}.
We first restrict ourselves to the Euclidean setting and consider the following models: 
\begin{itemize}
    \item {\bf Model I}: $X,Z\ i.i.d.~N(0,1)$, $\qquad \qquad\, \;Y=X+Z+N(1,1)$.
    \item {\bf Model II}: $X,Z\ i.i.d.~N(0,1)$, $\qquad \qquad Y\sim {\rm Bernoulli}(e^{-Z^2/2})$.
    \item {\bf Model III}: $X,Z\ i.i.d.~{\rm Uniform}[0,1]$, $\quad Y=X+Z\ ({\rm mod}\ 1)$.
\end{itemize}
\noindent {\bf Model I}: Here we let $k,k_\X$ and $k_{\ddot{\X}}$ be linear kernels. As we are in a Gaussian setting, both our estimators $\hat{\rho^2}$ and $\tilde{\rho^2}$ are consistent\footnote{Note that the assumptions in~\cref{thm:graph_consistency} hold in this setting and thus the graph-based estimator $\hat{\rho^2}$ is consistent. Further, Assumption~\ref{assump:CME} holds which implies that $\tilde{\rho^2}$ is also consistent (by~\cref{thm:kernel_consistency}). To check that Assumption~\ref{assump:CME} holds, we notice that the RKHS associated with the linear kernel $k_\X$ is $\h_\X=\{f(x)=a^\top x|a\in\X\}$ --- the space of all linear functions on $\X = \R$. So $\E[a^\top Y|X] = a^\top (X+1) \in \h_\X +\R$. Note that Assumption~\ref{assump:CME}
also holds for $(Y,\ddot{X})$ by the same argument. For the RKHS-based estimator using the uncentered CME (see~\cref{rem:Uncentered-CME}), the analogous sufficient condition (like Assumption~\ref{assump:CME}; see Remark \ref{rk:uncenter}) does not hold, since for $a\neq 0$, $a^\top x + a\notin\h_\X$.}. One can check that $\rho^2(Y,Z|X) = 0.5$. For the graph-based estimator $\hat{\rho^2}$, we use directed 1-NN and 2-NN graphs. For $\tilde{\rho^2}$, we set $\varepsilon_n = 10^{-3}\cdot n^{-0.4}$ for all the three models considered here (which satisfies the condition for consistency in Theorem~\ref{thm:kernel_consistency}). Although the linear kernel is not characteristic, as $(Y,Z,X)$ is jointly normal, all the desired properties of $\rho^2(Y,Z|X)$ in~\cref{thm:Eta} hold (see~\cref{rk:lin_ker_gaussian}); and $\rho^2$ is equal to the squared partial correlation coefficient (see~\cref{prop:class_parcor}-(d)). The left panel of~\cref{EX1} shows, for different sample sizes $n$, the mean and 2.5\%, 97.5\% quantiles of the various estimators of $\rho^2$ (obtained from 1000 replications). It can be seen that all the estimators except the RKHS-based estimator using the uncentered CME formula (see~\cref{rem:Uncentered-CME}; also see~\cref{sec:Un-CME}) are consistent, converging to the true value $\rho^2=0.5$. As expected, $\hat{\rho^2}$ constructed from the 1-NN graph has less bias but higher variance compared to $\hat{\rho^2}$ constructed using the 2-NN graph. Note that $\tilde{\rho^2}$ achieves almost the same statistical performance as the squared partial correlation coefficient; a consequence of~\cref{prop:reduce_classi}. As model I describes a Gaussian setting, the classical partial correlation coefficient (and $\tilde{\rho^2}$) has the best performance. Notice that $\hat{\rho^2}$, in spite of being fully nonparametric, also achieves good performance. \newline

\noindent {\bf Model II}: We let $k(y_1,y_2) := \mathbf{1}{\{y_1=y_2\}}$ be the discrete kernel. In this case, $\rho^2=\frac{2\sqrt{6}+2\sqrt{3}-3\sqrt{2}-3}{3}\approx 0.37$. As the kernel $k(\cdot,\cdot)$ is bounded, $\hat{\rho^2}$ is automatically consistent (see~\cref{thm:graph_consistency}). To compute the RKHS-based estimator $\tilde{\rho^2}$ we take $k_\X(x_1,x_2) = e^{-|x_1-x_2|^2/2}$ (a Gaussian kernel) and $k_{\ddot{\X}}\left((x_1,z_1),(x_2,z_2)\right)=\left(e^{-|x_1-x_2|^2/2} + 1 \right)e^{-|z_1-z_2|^2/2}$ (a product kernel). 
One can also check that Assumption~\ref{assump:CME} holds in this case, and thus $\tilde{\rho^2}$ is consistent (by Theorem~\ref{thm:kernel_consistency}).
The behavior of all the estimators, as the sample size increases, is shown in the middle panel of~\cref{EX1}.
Note that the classical partial correlation coefficient fails to capture the conditional dependence and is almost 0. \newline

\noindent {\bf Model III}: We let $k$, $k_\X$, and $k_{\ddot{\X}}$ be Gaussian kernels (see~\eqref{rem:excharker}) with different bandwidths\footnote{Here $k(\cdot,\cdot) = k_\X(\cdot,\cdot) = \exp(-5|\cdot - \cdot|^2)$ and $k_{\ddot{\X}} (\cdot,\cdot)  =\exp(-2\|\cdot - \cdot\|^2)$. These bandwidths are just arbitrary choices that approximately fit to the scale of the data.}.
This model has also been considered in \cite{azadkia2019simple}. Note that here $\rho^2(Y,Z|X)=1$ since $Y$ is a measurable function of $X$ and $Z$. Here $\hat{\rho^2}$ (constructed using 1-NN and 2-NN graphs) is consistent as the Gaussian kernel is bounded (which is a sufficient condition for~\cref{thm:graph_consistency} to hold).
From the right panel of~\cref{EX1} we see that both of the graph-based estimators are very close to 1. 
Assumption~\ref{assump:CME} for the CME formula holds with $(Y,X)$ (since $Y$ and $X$ are independent) but it does not hold for $(Y,\ddot{X})$\footnote{Note that as $k_{\ddot{\X}}(\cdot,\cdot)$ is continuous on $\ddot{\X}\times \ddot{\X}$ with $\ddot{\X}$ compact, all the functions in $\h_{\ddot{\X}}$ are continuous~\cite[Chapter 3, Theorem 2]{cucker2002mathematical}. But for any $g \in \h_\Y$, and $h \in \h_{\ddot{\X}}$ (which is continuous), we have $\E[g(Y)|\ddot{X}=\ddot{x}]-h(\ddot{x}) = g(\ddot{x}) - h(\ddot{x})$ which is discontinuous and cannot be a constant $P_{\ddot{X}}$-a.e.}.
Therefore,~\cref{thm:kernel_consistency} does not hold and $\tilde{\rho^2}$ cannot be guaranteed to be consistent. As can be seen from the right panel of~\cref{EX1}, $\tilde{\rho^2}$ 
does not seem to converge to 1. 
However, compared to the classical partial correlation, which is almost 0, both $\hat{\rho^2}$ and $\tilde{\rho^2}$ provide evidence that $Y$ is highly dependent on $Z$, conditional on $X$. \newline


\begin{figure}
    \centering
    \includegraphics[width = 1\textwidth]{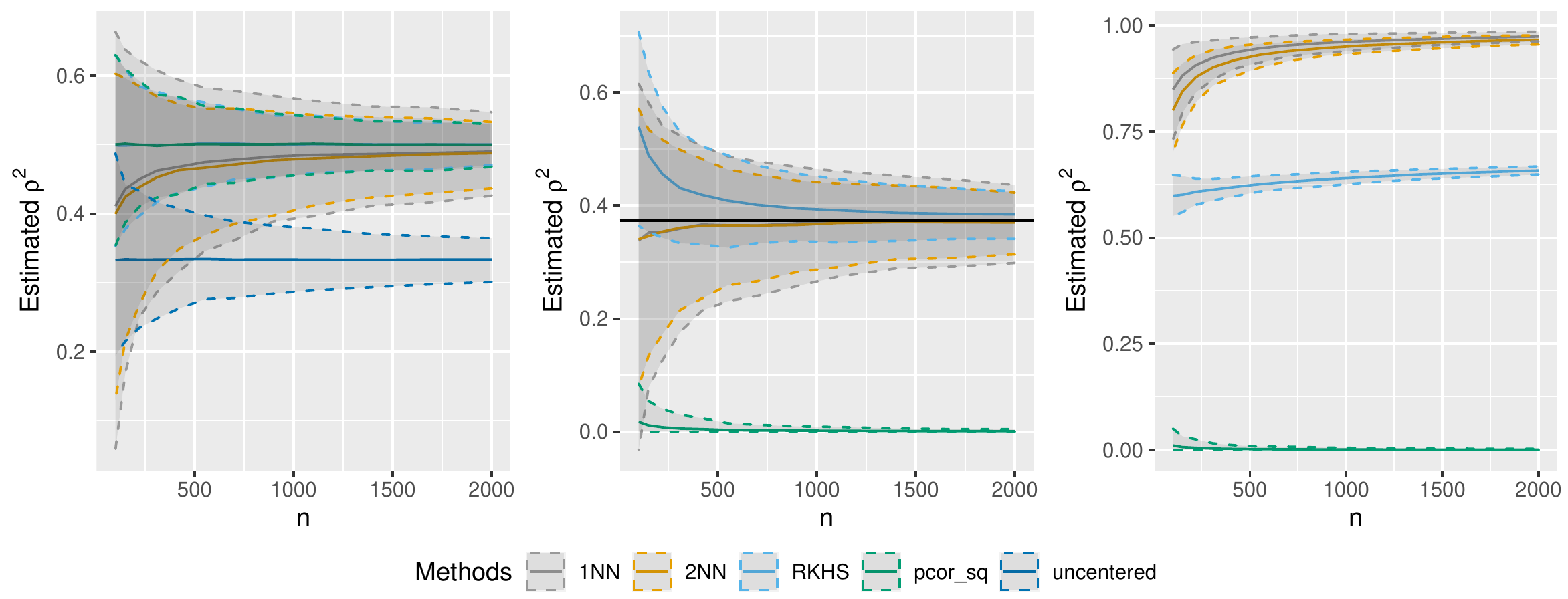}
    \caption{The performance of different estimators of $\rho^2$ as the sample increases. The solid lines show the mean for each estimator computed from 1000 replications; the dashed lines are the corresponding 2.5\% and 97.5\% quantiles. Here `1NN' and `2NN' denote the two graph-based estimators constructed using 1-NN and 2-NN graphs, `RKHS' denotes $\tilde{\rho^2}$, `pcor\_sq' denotes the squared partial correlation, and `uncentered' denotes the RKHS-based estimator constructed using the uncentered CME  (see~\cref{rem:Uncentered-CME}; also see~\cref{sec:Un-CME}).  
    The left, middle and right panels correspond to models I, II, and III respectively.}
    \label{EX1}
\end{figure}


\noindent{\bf A non-Euclidean example}: Next, we consider the case where $\Y$ is the \emph{special orthogonal group} ${\rm SO}(3)$, the space consisting of $3\times 3$ orthogonal matrices with determinant 1.
${\rm SO(3)}$ has been used to characterize the rotation of tectonic plates in geophysics \cite{hanna2000fitting} as well as in the studies of human kinematics and robotics \cite{stavdahl2005optimal}.
We use the following characteristic kernel on ${\rm SO}(3)$ (given by \cite{fukumizu2009characteristic}):
\begin{equation}\label{eq:SO3ker}
    k(A,B) := \frac{\pi \theta(\pi-\theta)}{8\sin (\theta)},
\end{equation}
where $e^{\pm \sqrt{-1}\theta}\ (0\leq \theta\leq \pi)$ are the eigenvalues of $B^{-1}A$, i.e., $\cos \theta = \frac{{\rm Tr}(B^{-1}A)-1}{2}$.
Define the rotation around $x$- and $z$-axis as $R_1,R_3:\R\to {\rm SO}(3)$, defined by (for $x,z \in \R$)
$$R_1(x) :=\begin{pmatrix}
    1 &0 &0\\
   0 &\cos (x) & -\sin (x)\\
    0&\sin (x) & \cos(x)\\
\end{pmatrix},\qquad  R_3(z) := \begin{pmatrix}
    \cos (z) &-\sin (z) &0\\
   \sin (z) &\cos (z) & 0\\
    0&0& 1\\
\end{pmatrix}.$$
Let $X,Z\overset{i.i.d.}{\sim} N(0,1)$. Consider the two models:
\begin{itemize}
    \item {\bf Model IV}: $Y_1 = R_1(X)R_3(Z)$. Here $Y_1$ is a function of $X$ and $Z$, and $\rho^2(Y_1,Z|X)=1$.
    \item {\bf Model V}: $Y_2 = R_1(X)R_3(\varepsilon)$, for an independent $\varepsilon\sim N(0,1)$. In this model, $Y_2\indep Z|X$, and thus $\rho^2(Y_2,Z|X)=0$.
\end{itemize}
Figure \ref{fig:SO3} shows the means and 95\% confidence bands for $\hat{\rho^2}(Y_1,Z|X)$ and $\hat{\rho^2}(Y_2,Z|X)$ constructed using 1-NN and 2-NN graphs from 5000 replications.
It can be seen that as $n$ increases, $\hat{\rho^2}(Y_1,Z|X)$ gets very close to 1 (which provides evidence that $Y_1$ is a function of $Z$ given $X$), and $\hat{\rho^2}(Y_2,Z|X)$ comes very close to 0 (suggesting that $Y_2$ is conditionally independent of $Z$ given $X$).

\begin{figure}
    \centering
    \includegraphics[width = 0.8\textwidth]{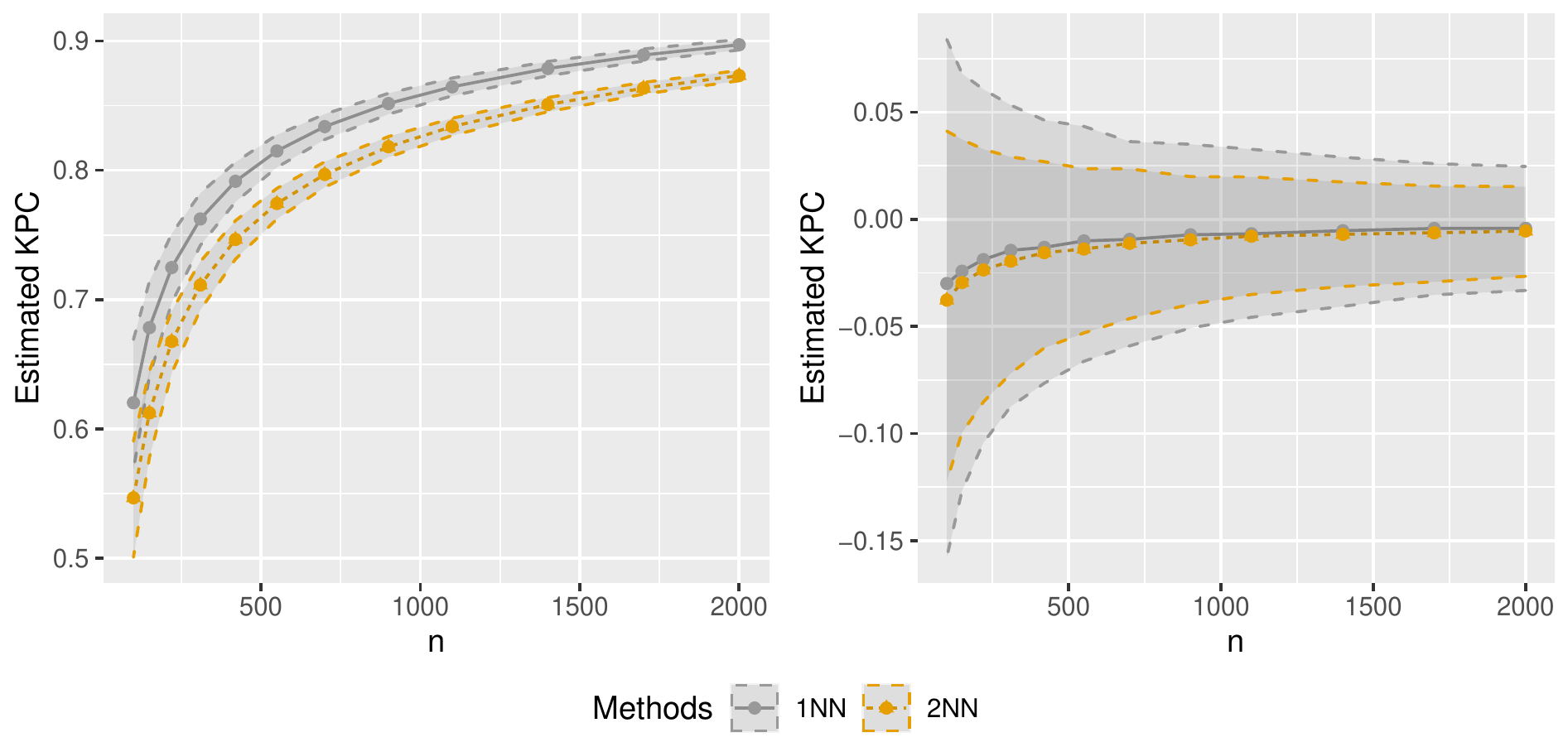}
    \caption{The performance of $\hat{\rho^2}$ (constructed using 1-NN and 2-NN graphs) as  sample increases. The solid lines show the mean of $\hat{\rho^2}$ from 5000 replications and the confidence bands show the corresponding 2.5\% and 97.5\% quantiles as a function of the sample size $n$. The left panel corresponds to model IV, while the right panel corresponds to model V.}
    \label{fig:SO3}
\end{figure}

We did not plot $\tilde{\rho^2}$ as it is not consistent and also quite sensitive to the choice of the regularization parameter $\varepsilon$; see~\cref{table:tilde} where we report $\tilde{\rho^2}$ for models IV and V when $n=1000$, $k_\X (x,x') = \exp(-|x-x'|^2)$, and $k_{\ddot{\X}} (x,x') = \exp(-\|x-x'\|^2/2)$.
However, $\tilde{\rho^2}(Y_1,Z|X)$ is always much larger than $\tilde{\rho^2}(Y_2,Z|X)$, indicating the conditional association between $Y_1$ and $Z$ is much stronger than that between $Y_2$ and $Z$ when controlling for $X$. \newline

\begin{table}
    \centering
    \caption{$\tilde{\rho^2}$ for $n=1000$ observations with different $\varepsilon$'s.}
    \label{table:tilde}
    \begin{tabular}{|c|cccccccc|} 
     \hline
     Estimands  & Estimators  & $\varepsilon = 10^{-3}$& $10^{-4}$& $10^{-5}$  & $10^{-6}$ & $10^{-7}$&$10^{-8}$&$10^{-9}$\\
     \hline
     $\rho^2(Y_1,Z|X)=1$&$\tilde{\rho^2}(Y_1,Z|X)$ &0.427 &0.548&0.620& 0.665 &0.697&0.724&0.747\\
     \hline
     $\rho^2(Y_2,Z|X)=0$ & $\tilde{\rho^2}(Y_2,Z|X)$ & 0.027&0.038&0.052&0.067&0.080& 0.093&0.106\\
     \hline
    \end{tabular}
\end{table}

 \noindent {\bf Summary}: In all the simulation examples we see that the graph-based estimator $\hat{\rho^2}$ has very good performance. It is able to capture different kinds of nonlinear conditional dependencies, under minimal assumptions. The RKHS-based estimator $\tilde{\rho^2}$ also performs quite well, although it need not be consistent always (as Assumption~\ref{assump:CME} may not hold in certain applications).
 Further, as both $\hat{\rho^2}$ and $\tilde{\rho^2}$ can handle non-Euclidean responses, we believe that they are useful tools for detecting conditional dependencies in any application. 

\subsubsection{Variable selection}\label{subsec:var_select}
In this subsection we examine the performance of our proposed variable selection procedures --- KFOCI (Algorithm \ref{algo:graph}) and Algorithm \ref{algo:CME} ---  in a variety of settings.
Our examples include both low-dimensional and high-dimensional models.
Note that KFOCI can automatically determine the number of variables to select, while Algorithm \ref{algo:CME} requires prespecifying the number of variables to be chosen.

We consider the following models with $n=200$, $p=10$, and $X =(X_1,\ldots, X_p)\sim N(0,I_p)\in\R^p$:
\begin{itemize}
    \item LM (linear model): $Y = 3X_1 + 2X_2 - X_3 + N(0,1)$.
    \item GAM (generalized additive model): $Y = \sin(X_1) + 2\cos(X_2) + e^{X_3} + N(0,1)$.
    \item Nonlin1 (nonlinear model studied in \cite{azadkia2019simple}): $Y=X_1X_2 + \sin(X_1X_3)$.
    \item Nonlin2 (heavy-tailed): $Y=\frac{2\log (X_1^2 + X_2^4)}{\cos (X_1) + \sin(X_3)} + \epsilon\;\;$ where $\epsilon \sim t_1$, the $t$-distribution with 1 degree of freedom.
    \item Nonlin3 (non-additive noise): $Y=|X_1 + U|^{\sin(X_2-X_3)}$, where $U \sim {\rm Uniform}[0,1]$.
    \item SO(3) (non-Euclidean response): $Y=R_1(X_1)R_3(X_2X_3)\in {\rm SO}(3)$.
\end{itemize}
In all the examples, the noise variable is assumed to be independent of $X$. The above models cover different kinds of linear and nonlinear relationships between $Y$ and $X$.

We first compare the performance of KFOCI with other competing methods.  
We implement the KFOCI algorithm with the directed 1-NN graph and the 10-NN graph using Algorithm~\ref{algo:graph}.
For all the models except when the response takes values in SO(3), we use the Gaussian kernel with empirically chosen bandwidth, i.e., $k(y,y')=\exp\left(\frac{|y-y'|^2}{2s^2}\right)$, where $s$ is the median of pairwise distances $\{|y_i-y_j|\}_{i<j}$.
When $\Y =$ SO(3), we use the kernel in~\eqref{eq:SO3ker}. A natural competitor of KFOCI is FOCI~\cite{azadkia2019simple} which is implemented in the \texttt{R} package \texttt{FOCI} \cite{RpackFOCI}. We also compare KFOCI with `ols (penter=0.01)' which is the forward stepwise variable selection algorithm in linear regression (where a variable with the smallest $p$-value less than 0.01 enters the model at every stage), implemented using the function \texttt{ols\_step\_forward\_p} in the \texttt{R} package \texttt{olsrr}~\cite{hebbali2018olsrr}.
We also consider `VSURF', variable selection using random forests, implemented using the \texttt{R} package  \texttt{VSURF}\footnote{For \texttt{VSURF}, we take the variables obtained at the ``interpretation step", which aims to select all variables related to the response for interpretation purpose.}~\cite{genuer2015vsurf}.
Note that for all the considered models, $X_1,X_2,X_3$ are the ``correct" variables.
In~\cref{table:varselect} we report: (i) The proportion of times $\{X_1,X_2,X_3\}$ is exactly selected, (ii) the proportion of times $\{X_1,X_2,X_3\}$ is selected with possibly other variables, and (iii) the average number of variables selected, by the different methods in 100 replications. 
It can be seen from~\cref{table:varselect} that KFOCI achieves the best performance in all the nonlinear settings considered; in particular, KFOCI with the 10-NN graph selects exactly $\{X_1,X_2,X_3\}$ more than 90\% of the times in all the nonlinear examples and also has good performance in the linear setting. Although FOCI uses the 1-NN graph in its algorithm, it has inferior performance compared to KFOCI with 1-NN; this indicates that the Gaussian kernel may be better at detecting various conditional associations than the special kernel used in FOCI (see~\cref{lem:AC19}).

\begin{table}
    \centering
    \caption{Performance of the various variable selection algorithms. The reported numbers are: The proportion of times $\{X_1,X_2,X_3\}$ is exactly selected / the proportion of times $\{X_1,X_2,X_3\}$ is selected possibly with other variables / the average number of variables selected. The method with the best performance is highlighted in bold.  Here ``---" means that the method cannot deal with responses in the space ${\rm SO}(3)$.}
    \label{table:varselect}
    \begin{tabular}{|c|ccccc|} 
     \hline
     Models  & KFOCI (1-NN)  & KFOCI (10-NN) & FOCI \cite{azadkia2019simple} &ols (penter=0.01) & VSURF \cite{genuer2015vsurf}\\
     \hline
     LM & 0.87/0.98/3.09  &0.81/0.81/2.81 &0.57/0.87/3.25 &\textbf{0.95/1.00}/3.05 &0.82/0.82/2.82\\
     GAM &0.39/0.74/3.28  &\textbf{0.92/0.93}/2.94 &0.15/0.64/3.62 &0.03/0.04/2.03 &0.71/0.71/2.68\\
     Nonlin1 &0.88/0.95/2.98  &\textbf{1.00/1.00}/3.00 &0.56/0.72/2.87 &0.00/0.00/1.06 &0.21/0.21/2.21\\
     Nonlin2 &0.41/0.79/3.36  &\textbf{0.93/0.97}/3.01 &0.22/0.53/2.93 &0.00/0.00/1.04 &0.00/0.03/2.67\\
     Nonlin3 &0.53/0.82/3.16  &\textbf{1.00/1.00}/3.00 &0.28/0.52/2.74 &0.00/0.00/1.07 &0.06/0.23/2.90\\
     SO(3)&\textbf{1.00/1.00}/3.00 &0.97/0.97/2.94 & --- & ---&---\\
    \hline
    \end{tabular}
\end{table}
Next, we consider the case where the number of variables to select is set by the oracle as 3.
For KFOCI, we still use 1-NN and 10-NN graphs as before, but without imposing the automatic stopping criterion (denoted by `1-NN', `10-NN' in~\cref{table:varselect2}).
For Algorithm \ref{algo:CME} (denoted by `KPC (RKHS)' in~\cref{table:varselect2}), we set the kernel on $\Y$ as the same kernel for the methods `1-NN'/`10-NN'; the kernel on $\X_S$ is taken as $k_{\X_S}(x,x')=\exp\left(\|x-x'\|_{\R^{|S|}}^2/|S| \right)$, and $\varepsilon = 10^{-3}$.
We compare our methods with `FWDselect (GAM)', the forward stepwise variable selection algorithm for general additive models, using the function \texttt{selection} in the \texttt{R} package \texttt{FWDselect} \cite{sestelo2016fwdselect}. We also compare with \texttt{varimp}, which selects three variables with the highest importance scores in the random forest model, implemented by the function \texttt{varimp} in the \texttt{R} package \texttt{party}~\cite{hothorn2010party} (using default settings).
In~\cref{table:varselect2} we report: (i) The proportion of times $\{X_1,X_2,X_3\}$ is exactly selected, and (ii) the average number of correct variables among the 3 selected variables $\hat{S}$ (i.e.,~$|\hat{S}\cap \{X_1,X_2,X_3\}|$),  by the different methods (in 100 replications).
It can be seen that our methods achieve superior performance compared to the other algorithms.
In particular, `10-NN' and `KPC (RKHS)' select exactly $X_1,X_2,X_3$ more than 99\% of the times, in all the models. \newline
\begin{table}
    \centering
    \caption{Performance of the various variable selection algorithms. The number of variables to be selected is set by the oracle as 3. The reported numbers are: The proportion of times $\{X_1,X_2,X_3\}$ is exactly selected / the average number of correct variables among the 3 selected variables (i.e. $|\hat{S}\cap \{X_1,X_2,X_3\}|$). The method with the best performance is highlighted in bold.
    }
    \label{table:varselect2}
    \begin{tabular}{|c|cccccc|} 
     \hline
     Models  & 1-NN  & 10-NN & KPC (RKHS) & FOCI & FWDselect (GAM) & varimp (RF)\\
     \hline
     LM &\textbf{1.00/3.00}  &\textbf{1.00/3.00} &\textbf{1.00/3.00} &0.93/2.93 &\textbf{1.00/3.00} &\textbf{1.00/3.00}\\
     GAM &0.78/2.78 &0.99/2.99 &\textbf{1.00/3.00} &0.54/2.52 &\textbf{1.00/3.00} &0.97/2.97\\
     Nonlin1 &0.88/2.86 &\textbf{1.00/3.00} &\textbf{1.00/3.00}  &0.60/2.24 &0.05/1.67 &0.66/2.66\\
     Nonlin2 &0.64/2.54 &\textbf{1.00/3.00} &0.99/2.99  &0.40/2.18 &0.02/1.15 &0.00/1.11\\
     Nonlin3 &0.68/2.57 &\textbf{1.00/3.00} &\textbf{1.00/3.00}  &0.46/1.93 &0.10/1.69 &0.47/2.33\\
     SO(3)&\textbf{1.00/3.00} &\textbf{1.00/3.00} &\textbf{1.00/3.00} &--- & ---&---\\
    \hline
    \end{tabular}
\end{table}

\noindent\textbf{High-dimensional setting}: To study the performance of the methods in the high-dimensional setting we increase $p$ from $10$ to $1000$ keeping $n=200$ and consider the same models as in the beginning of~\cref{subsec:var_select}. We compare our method KFOCI with popular high-dimensional variable selection algorithms such as Lasso \cite{Lasso} and Dantzig selector~\cite{Candes-Tao-07}.
Lasso was implemented using the \texttt{R} package \texttt{glmnet}~\cite{glmnet} and the Dantzig selector was implemented using the package \texttt{hdme}~\cite{hdme}; in both cases the tuning parameters were chosen as ``lambda.1se"\footnote{``lambda.1se", as proposed in \cite{ELS-2}, is the largest value of lambda such that the mean cross-validated error is within 1 standard error of the minimum.} obtained from 5-fold cross-validation.
As in the low-dimensional setting, in~\cref{table:varselect_highdim} we report: (i) The proportion of times $\{X_1,X_2,X_3\}$ is exactly selected,  (ii) the proportion of times $\{X_1,X_2,X_3\}$ is selected with possibly other variables, and (iii) the average number of variables selected, by the different methods (in 100 replications). 
It can be seen that our method KFOCI (`10-NN') still selects the correct variables most of the times in the high-dimensional setting and outperforms all other competitors in all models except in the linear model (LM). Even in the linear model, if the goal is to exactly select $\{X_1,X_2,X_3\}$,
then KFOCI (`10-NN') also outperforms Lasso and the Dantzig selector that are designed specifically for this setting.
In the nonlinear models, all the other methods are essentially never able to select the correct set of predictors.
\begin{table}
    \centering
    \caption{(High-dimensional setting: $n=200$, $p=1000$)
    Performance of the various variable selection algorithms. The reported numbers are: The proportion of times $\{X_1,X_2,X_3\}$ is exactly selected / the proportion of times $\{X_1,X_2,X_3\}$ is selected possibly with other variables / the average number of variables selected. The method with the best performance is highlighted in bold.  Here ``---" means that the method cannot deal with responses in the space ${\rm SO}(3)$.
    }
    \label{table:varselect_highdim}
    \begin{tabular}{|c|ccccc|} 
     \hline
     Models  & KFOCI (1-NN)  & KFOCI (10-NN) & FOCI \cite{azadkia2019simple} &Lasso \cite{Lasso} & Dantzig \cite{Candes-Tao-07}\\
     \hline
     LM & 0.23/0.90/3.72 &\textbf{0.82}/0.82/2.82 &0.02/0.53/3.99 &0.66/\textbf{1.00}/4.19 &0.01/\textbf{1.00}/29.63\\
     GAM & 0.01/0.21/3.83  &\textbf{0.76}/\textbf{0.92}/3.14 &0.00/0.05/4.14 &0.00/0.00/1.17 &0.00/0.00/2.39\\
     Nonlin1 & 0.21/0.49/3.77  &\textbf{0.99}/\textbf{1.00}/3.01 &0.01/0.03/3.43 &0.00/0.00/0.02 &0.00/0.00/0.00\\
     Nonlin2 & 0.00/0.03/3.48  &\textbf{0.38}/\textbf{0.78}/8.82 &0.00/0.00/3.31 &0.00/0.00/0.00 &0.00/0.00/0.00\\
     Nonlin3 &0.00/0.03/3.21  &\textbf{0.87}/\textbf{0.94}/3.50  &0.00/0.01/3.30 &0.00/0.00/0.00 &0.00/0.00/0.00\\
     SO(3)& 0.82/0.83/2.84 & \textbf{0.95}/\textbf{0.95}/2.90 & --- & ---&---\\
    \hline
    \end{tabular}
\end{table}

We now fix the number of covariates to be selected to $3$ to 
examine how well Algorithm \ref{algo:CME} performs in the high-dimensional setting.
We compare our methods with \emph{least angle regression} (LARS) \cite{LAR04} implemented in the \texttt{R} package \texttt{lars}~\cite{Rpacklars}. We use the first 3 variables selected by LARS; and in our examples it almost always selected the first 3 variables entering the Lasso path before any of the variables left the active set. With the same choices of the kernel and $\varepsilon$ as before, in~\cref{table:varselect4} we report: (i) The proportion of times $\{X_1,X_2,X_3\}$ is exactly selected, and (ii) the average number of correct variables among the 3 selected variables  (i.e.,~$|\hat{S}\cap \{X_1,X_2,X_3\}|$),  by the different methods (in 100 replications). Besides `10-NN', KPC (RKHS) (i.e., Algorithm \ref{algo:CME}) performs very well in the high-dimensional regime, exactly selecting $\{X_1,X_2,X_3\}$ more than 90\% of the times and achieving the best performance among all the methods in all the considered models. It can also be seen that the performance of both `1-NN' and `10-NN' improves when compared to the case when the number of predictors was not prespecified (cf.~Tables~\ref{table:varselect_highdim} and~\ref{table:varselect4}).

\begin{table}
    \centering
    \caption{(High-dimensional setting: $n=200$, $p=1000$)
    Performance of the various variable selection algorithms. The number of variables to be selected is set by the oracle as 3. The reported numbers are: The proportion of times $\{X_1,X_2,X_3\}$ is exactly selected / the average number of correct variables among the 3 selected variables (i.e. $|\hat{S}\cap \{X_1,X_2,X_3\}|$). The method with the best performance is highlighted in bold.
    }
    \label{table:varselect4}
    \begin{tabular}{|c|ccccc|} 
     \hline
     Models  & 1-NN  & 10-NN & KPC (RKHS) & FOCI \cite{azadkia2019simple} & LARS \cite{LAR04}\\
     \hline
     LM & 0.88/2.88& \textbf{1.00/3.00}&\textbf{1.00/3.00}&0.42/2.42&\textbf{1.00/3.00}\\
     GAM & 0.14/1.81 & 0.96/2.96 & \textbf{0.98/2.98} & 0.02/1.09 & 0.01/1.82\\
     Nonlin1 & 0.24/1.18 & 0.99/2.99 &\textbf{1.00/3.00}  & 0.02/0.11 &0.00/0.11\\
     Nonlin2 & 0.02/0.26 & 0.72/2.42 & \textbf{0.92/2.88} &0.00/0.10 &0.00/0.02\\
     Nonlin3 & 0.02/0.28 & 0.89/2.79& \textbf{1.00/3.00} & 0.01/0.07 &0.00/0.24\\
     SO(3)& 0.83/2.75 & \textbf{1.00/3.00}& \textbf{1.00/3.00}&--- &---\\
    \hline
    \end{tabular}
\end{table}

\noindent {\bf Summary}: The comparison with the various variable section methods reveal that KFOCI
and Algorithm \ref{algo:CME} (with the Gaussian kernel) have excellent performance, both in the low- and high-dimensional regimes, in both linear and nonlinear models.
Further, KFOCI has a stopping criterion that can automatically determine the number of variables to select. 
As mentioned above, even in the high-dimensional linear regression setting, our model-free approach KFOCI yields comparable, and sometimes better results than the Lasso and the Dantzig selector. 


\subsubsection{Testing for conditional independence under the model-X framework}\label{sec:Testing}
In this section we briefly consider the problem of testing the hypothesis of conditional independence between $Y$ and $Z$ given $X$. Although there is a huge literature for this problem (see e.g.,~\cite{linton1996conditional,bergsma2004testing, su2007consistent,su2008nonparametric, fukumizu2008kernel, song2009testing, huang2010testing, su2014testing, doran2014permutation, wang2015conditional, PatraEtAl-16, runge2018conditional, Ke2019, shah2020hardness} and the references therein) most of the proposed methods only provide approximate level-$\alpha$ tests that involve delicate choice of tuning parameter(s). Our estimators of $\rho^2$, $\hat{\rho^2}$ and $\tilde{\rho^2}$, can also be used as test statistics for this testing problem; however their exact distributions, under the null hypothesis, is currently unknown. Nevertheless, under the model-X framework~\cite{candes2018panning}, which assumes a known conditional distribution of $Z$ given $X$, we can use both $\hat{\rho^2}$ and $\tilde{\rho^2}$ to develop {\it exact level-$\alpha$} tests for testing the hypothesis of conditional independence.
In Section~\ref{sec:ModelXFDR} we further show that our estimators of $\rho^2$ can be easily applied in the model-X framework~\cite{candes2018panning} to perform variable selection while controlling the \emph{false discovery rate} (FDR).

Consider testing the null hypothesis $H_0:Y\indep Z|X$, given i.i.d.~data $\{(X_i,Y_i,Z_i)\}_{i=1}^n$ from $P$ using the model-X conditional randomization test~\cite{candes2018panning, berrett2019conditional}. Let us denote the observed data matrices by $\mathbf{X},\mathbf{Y},\mathbf{Z}$. Suppose that $T(\mathbf{X},\mathbf{Y},\mathbf{Z})$ is the test statistic used which rejects $H_0$ for large values of $T$; e.g., we can take $T = \hat{\rho^2}$ or $T = \tilde{\rho^2}$. To simulate the null distribution of $T$, we generate $Z_i^{(j)}$ (for $j=1,\ldots,B$) from the conditional distribution $Z|X =X_i$ (which is assumed to be known). Let us denote $\{Z_i^{(j)}\}_{i=1}^n$ by $\mathbf{Z}^{(j)}$. Thus, $\{T(\mathbf{X},\mathbf{Y},\mathbf{Z}^{(j)})\}_{j=1}^{B}$ represents $B$ realizations of the test statistic under the null hypothesis. Then, 
$$p := \frac{1+\sum_{j=1}^B 1_{T(\mathbf{X},\mathbf{Y},\mathbf{Z}^{(j)})\geq T(\mathbf{X},\mathbf{Y},\mathbf{Z})} }{1+B}$$ yields a valid $p$-value for testing the hypothesis of conditional independence; a consequence of~\cite{candes2018panning,berrett2019conditional}. We may reject $H_0$ if $p\leq 0.05$.

We illustrate the power of our proposed testing procedure using a small simulation study. Consider $n=200$, $X\sim N(0,1)$ and the following models:
\begin{itemize}
    \item (Additive) $Z = X+U$ where $U \sim {\rm Uniform}[-1,1]$ ($X \indep U$), and $$Y=\gamma \sin (ZX) + (1-\gamma)\left(e^X X^{-2} + \epsilon\right), \quad \mbox{ where } \epsilon \sim N(0,1).$$
    \item (Multiplicative) $Z = XU$ where $U \sim N(0,1)$ ($X \indep U$), and $$Y=|\tanh(X) + \epsilon|^{1-\gamma}\cdot [\cosh(ZX)]^\gamma, \quad \mbox{ where } \epsilon \sim N(0,1).$$
\end{itemize}
The above two models are just arbitrary choices that cover different types of nonlinear relationships. Here $\gamma$ can be viewed as a parameter that captures the strength of association between $Y$ and $Z$ given $X$: When $\gamma=0$, $Y\indep Z|X$ and when $\gamma=1$, $Y$ is a function of $Z$ given $X$.

\begin{figure}
    \centering
    \includegraphics[width = 1\textwidth]{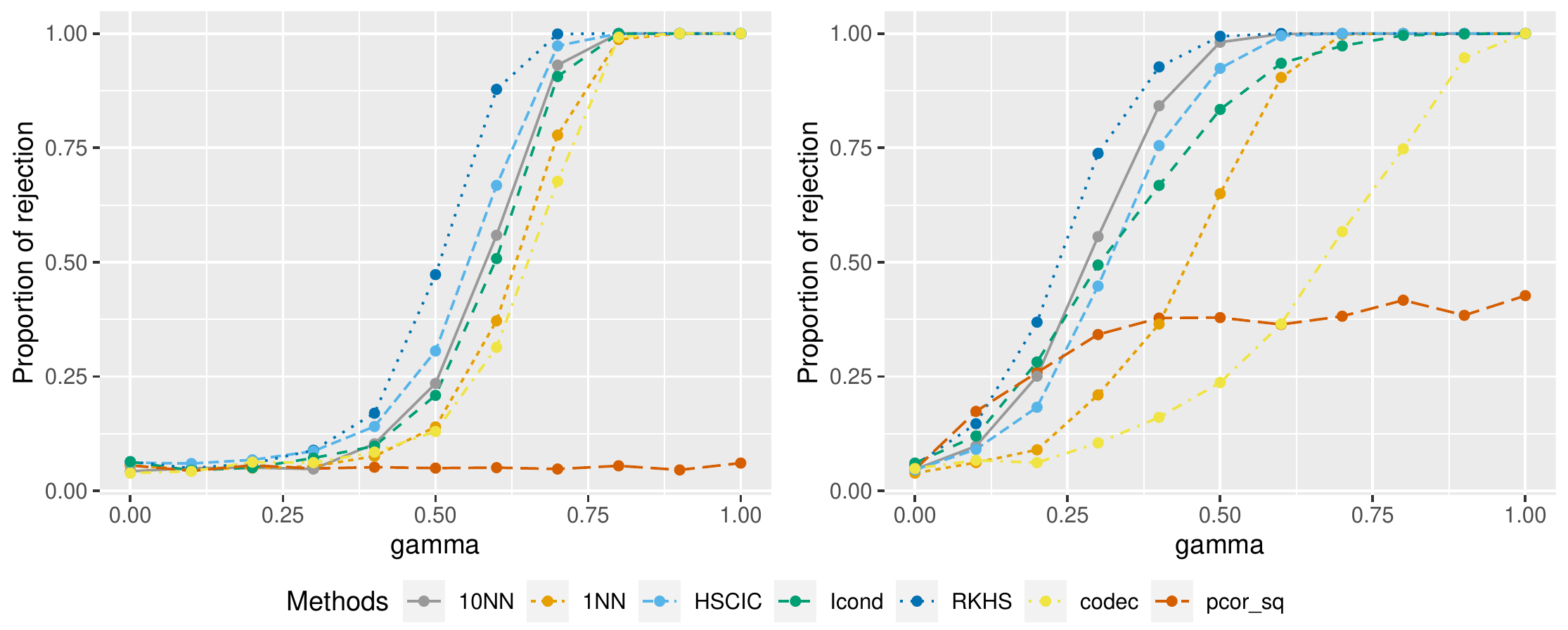}
    \caption{The proportion of rejections for different methods as $\gamma$ varies between 0 and 1 for the two examples. The left panel corresponds to the additive model and the right panel corresponds to the multiplicative model.}
    \label{fig:PowerCurve}
\end{figure}
In~\cref{fig:PowerCurve} we illustrate the performance of various test statistics for testing the null hypothesis of conditional independence in the above examples. Here we consider $T$ to be the classical partial correlation squared, $I^{COND}$~\cite[Eq.~(7)]{fukumizu2008kernel}, HS${\rm \ddot{C}}$IC \cite[Eq.~(24)]{sheng2019distance}, CODEC \cite{azadkia2019simple}, our graph-based estimators $\hat{\rho^2}(Y,Z|X)$ (with directed 1-NN, 10-NN graphs), and our RKHS-based estimator $\tilde{\rho^2}(Y,Z|X)$. To make the results comparable, we will use the same Gaussian kernel function $\exp\left(-{\|x-x'\|^2_{\R^d}}/{d}\right)$ when we are in $\R^d$ (for $d=1,2$) for different methods\footnote{
More specifically, we will use the same kernel on $\Y$ defined as $k(\cdot,\cdot) = \exp(-|\cdot - \cdot|^2)$ for $\hat{\rho^2}$, $\tilde{\rho^2}$, $I^{COND}$, HS${\rm \ddot{C}}$IC; the same kernel on $\X$ defined as $k_{\X}(\cdot,\cdot) = \exp(-|\cdot - \cdot|^2)$ for $\tilde{\rho^2}$, $I^{COND}$, HS${\rm \ddot{C}}$IC; the same kernel on $\ddot{\X}=\X\times\Z$ defined as $k_{\ddot{\X}}(\cdot,\cdot) = \exp(-\|\cdot - \cdot\|^2/2)$ for $\tilde{\rho^2}$, $I^{COND}$, HS${\rm \ddot{C}}$IC;
and kernels $k_\Z (\cdot,\cdot) = \exp(-|\cdot - \cdot|^2)$, $k_{\Y\times \X} (\cdot,\cdot) =k_{\Y\times \Z}(\cdot,\cdot) = \exp(-\|\cdot - \cdot\|^2/2) $ for $I^{COND}$.}. We set $B$ to be 100. For $\tilde{\rho^2}(Y,Z|X)$, $I^{COND}$ and HS${\rm \ddot{C}}$IC, the regularization parameter $\varepsilon$ (resp. $\lambda$) is set as $10^{-3}$.
\cref{fig:PowerCurve} shows the proportion of rejections from 1000 replications for different methods and $\gamma$ varying between 0 and 1.

It can be seen that when $H_0$ is true, i.e., $\gamma=0$, all methods have size $0.05$, as expected. When $H_0$ is violated, i.e., $\gamma>0$, our methods tend to have higher power. In both models, $\tilde{\rho^2}$ (i.e., `RKHS' in~\cref{fig:PowerCurve}) achieves the highest power among all the methods considered. $\hat{\rho^2}$ with 10-NN graph achieves the second highest power in the multiplicative model and the third highest power in the additive model.
For the additive model, HS${\rm \ddot{C}}$IC, which is defined as the Hilbert-Schmidt norm of a conditional cross-covariance operator, has the second highest power.
This is not very surprising as HS${\rm \ddot{C}}$IC also used cross-covariance operators to detect conditional dependence, similar to $\tilde{\rho^2}$.
CODEC, which also uses 1-NN graphs, has similar performance to KPC with 1-NN graph in the additive model, but inferior performance in the multiplicative model.
By increasing the number of neighbors from 1 to 10, it can be seen that $\hat{\rho^2}$ gains better performance. Classical partial correlation, however, cannot detect the nonlinear conditional associations in the examples.

\subsection{Real data examples}\label{sec:Real-Data} 
In this subsection we study the performance of our proposed methods on real data. The real data sets considered involve continuous, discrete, and non-Euclidean response variables. We also compare and contrast the performance of our methods with a number of existing and useful alternatives. In the following, unless otherwise specified, for the kernels  $k,k_\X,k_{\ddot{\X}}$, we use the Gaussian kernel with empirically chosen bandwidth, i.e., $k(x,x')=\exp\left(\frac{\|x-x'\|^2}{2s^2}\right)$, where $s$ is the median of the pairwise distances $\{\|x_i-x_j\|\}_{i<j}$. We also normalize each real-valued variable to have mean 0 and variance 1. \newline
\begin{table}
    \centering
    \caption{The variables selected by different methods for the surgical data set.}
    \label{tab:stepforward}
    \begin{tabular}{lcl} 
     \hline
     Methods   & \# variables selected &Selected variables\\
     \hline
     Stepwise forward regression& 4/8 & enzyme\_test, pindex, alc\_heavy, bcs\\
     Best subset (BIC, ${\rm PRESS}_p$)   & 4/8 &enzyme\_test, pindex, alc\_heavy, bcs\\
     KFOCI (1-NN, 2-NN, 3-NN)  & 4/8 &enzyme\_test, pindex, liver\_test, alc\_heavy\\
     VSURF \cite{genuer2015vsurf} (interpretation step) & 4/8 & enzyme\_test, liver\_test, pindex, alc\_heavy\\
     VSURF \cite{genuer2015vsurf} (prediction step) & 3/8 & enzyme\_test, liver\_test, pindex\\   
     FOCI \cite{azadkia2019simple} & 3/8& enzyme\_test, liver\_test, alc\_heavy\\
     npvarselec \cite{zambom2017package} &4/8&pindex, enzyme\_test, liver\_test, alc\_heavy\\
     MMPC \cite{tsagris2018feature} & 4/8 & pindex, enzyme\_test, liver\_test, alc\_heavy\\
    \hline
    \end{tabular}
\end{table}

\noindent\textbf{Surgical data}: The \emph{surgical} data, available in the \texttt{R} package \texttt{olsrr} \cite{hebbali2018olsrr}, consists of survival data for $n=54$ patients undergoing liver operation along with 8 covariates\footnote{
    The 8 covariates are:
bcs (blood clotting score),
pindex (prognostic index),
enzyme\_test (enzyme function test score),
liver\_test (liver function test score),
age,
gender,
indicator variable for gender,
alc\_mod
(indicator variable for history of alcohol use),
and alc\_heavy
(indicator variable for history of alcohol use).}. 
The response $Y$ is the survival time.
Investigators have found that a linear model taking $\log (Y)$ as response and 4 out of the 8 covariates as predictors describe the data well (see \cite[Section 9.4]{kutner2005applied} and also \cite[Section 10.6]{kutner2005applied}).
In particular, this linear model can be obtained by
stepwise forward regression with any ${ penter}\in[0.001,0.1]$\footnote{Recall that the variable with the smallest $p$-value less than `penter' will enter the model at every stage.},
and it is also the best submodel in terms of BIC and ${\rm PRESS}_p$ \cite[Section 9.4]{kutner2005applied}.
We also work with the transformed response $\log (Y)$.
We compare this linear model with our KFOCI and other variable selection methods such as FOCI \cite{azadkia2019simple}, \texttt{VSURF} \cite{genuer2015vsurf}, \texttt{npvarselec} \cite{zambom2017package}, and \texttt{MMPC} \cite{tsagris2018feature} (all using their default settings).
The selected variables, obtained by the different methods, are shown in Table~\ref{tab:stepforward}.
It can be seen the variables selected by KFOCI (with 1-NN, 2-NN, 3-NN graphs) are very similar to those
selected by the carefully analyzed linear regression approach\footnote{Note that `bcs' selected by the final linear model in~\cite{kutner2005applied} is replaced by `liver\_test' by KFOCI.}.
Further, KFOCI selects the same set of variables as many well-implemented variable selection algorithms such as \texttt{VSURF} \cite{genuer2015vsurf}, \texttt{npvarselec} \cite{zambom2017package}, and \texttt{MMPC} \cite{tsagris2018feature}.
\newline

\noindent\textbf{Spambase data}: This data set is available from the UCI Machine Learning Repository \cite{Dua:2019},
consisting of a response denoting whether an e-mail is spam or not, along with 57 covariates (this example has also been studied in~\cite{azadkia2019simple}).
The number of instances is $n=4601$. We first use Algorithm \ref{algo:graph} (KFOCI) to select a subset of variables.
Directed 1-NN and 10-NN graphs are considered. 
Since there is randomness in breaking ties while computing NNs, we show in the histogram in~\cref{Fig:combined_spam} the number of variables selected by `1-NN' and `10-NN' in 200 repetitions.
\begin{figure}
    \centering
    \includegraphics[width = 1\textwidth]{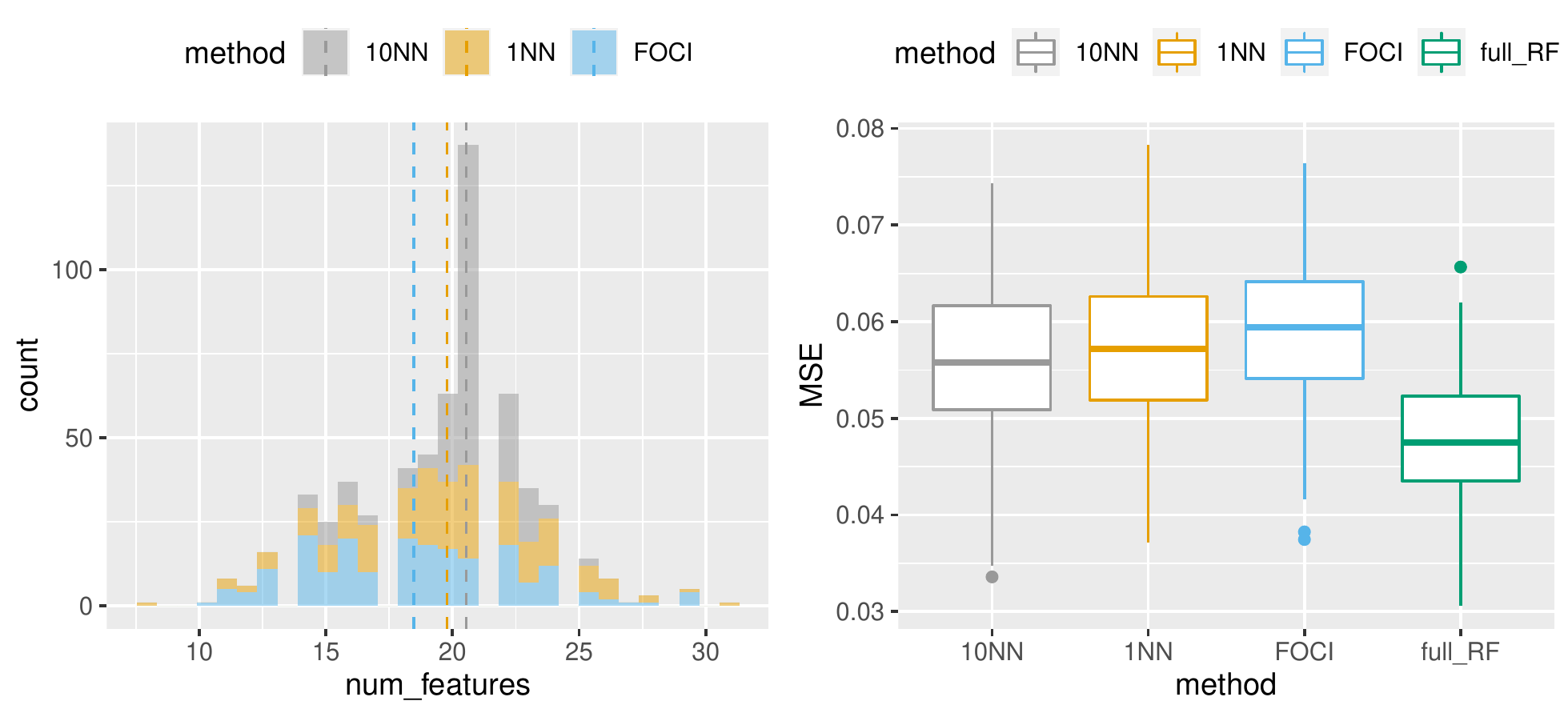}
    \caption{Left: Histogram of the number of selected variables (from 200 replications). The dashed lines show the mean for each group. Right: MSE after fitting random forest on test set for different methods (using 200 replications). For FOCI and KFOCI (`1-NN' and `10-NN'), the random forest algorithm was fit on the selected covariates. The full\_RF fits random forest with the entire set of covariates.}
    \label{Fig:combined_spam}
\end{figure}
The mean (resp. standard deviation) of the number of selected variables are 18.5, 19.8, 20.5 (resp. 4.1, 3.7, 2.1) for FOCI, `1-NN', `10-NN' respectively.
We then assign each data point with probability 0.8 (0.2) to training (test) set,
and fit a random forest model (implemented in the \texttt{R} package \texttt{randomForest} \cite{rcolorbrewer2018package} using default settings) with only the selected covariates.
The mean squared error (MSE) on the test set is reported on the right panel of Figure~\ref{Fig:combined_spam}.
Note that as there is randomness in breaking ties, splitting the training/test sets, as well as fitting the random forest, we repeat the whole procedure 200 times and present the box-plot for the MSEs.

It can be seen that in terms of MSE, `1-NN' is better than FOCI, and `10-NN' is better than `1-NN'.
The random forest fit on the whole set of covariates achieves the best performance in prediction as expected.
However, note that our methods use only about $1/3$ of the 57 covariates to achieve a prediction performance which is not much worse than that of the random forest fit with all the covariates. \newline

\noindent\textbf{Medical data}:
The medical data in \cite[Section 3.1.4]{edwards2012introduction} consists of three variables from 35 patients: Creatinine clearance $C$, digoxin clearance $D$, and urine flow $U$. The data set was also analyzed in \cite[Section 3]{fukumizu2008kernel}.
From medical knowledge, $D\indep U|C$ \cite{fukumizu2008kernel}.
$D$ should also be associated with $C$ controlling for $U$, since both creatinine $C$ and digoxin $D$ are mainly eliminated by the kidneys and can be viewed as a measure of kidney function.
Hence we expect the conditional association measure $\rho^2(D,U|C)$ to be close to 0 and $\rho^2(D,C|U)$ to be significantly different from 0.
This is affirmed by Table \ref{table:med_data}.
We used the directed 1-NN graph to compute $\hat{\rho^2}$, as the data set is quite small. Also, 
$\varepsilon$ is set to $10^{-2}$ for computing $\tilde{\rho^2}$.
As ties occur when finding the 1-NN for CODEC \cite{azadkia2019simple} and KPC (`1-NN'), which are broken uniformly at random, the results in~\cref{table:med_data} for these two methods are the means obtained from 10000 replications.
As mentioned in \cite{fukumizu2008kernel}, this is an example where the classical partial correlation fails to detect the conditional independence between $D$ and $U$ given $C$.
CODEC also fails to detect the conditional independence.
But it can be seen that KPC (`1-NN') successfully detects it, with $\hat{\rho^2}(D,U|C)$ being very close to 0. \newline

\begin{table}
    \centering
    \caption{Conditional association measures applied to the medical data. The four estimators considered here are: $\hat{\rho^2}$ (i.e., KPC (`1-NN') with directed 1-NN graph), $\tilde{\rho^2}$ (i.e., KPC (RKHS)), CODEC (the conditional dependence measure in~\cite{azadkia2019simple}), and the classical partial correlation squared. Note that we have slightly abused the notation: the $\rho^2$ of partial correlation squared and CODEC are actually different from the $\rho^2$ of KPC, but they are all between $[0,1]$, with smaller values indicating conditional independence and larger values indicating conditional association.}
    \label{table:med_data}
    \begin{tabular}{|l|c|c|} 
     \hline
     Estimands  & $\rho^2(D,U|C)$ & $\rho^2 (D,C|U)$ \\
     \hline
     Ground truth & $D\indep U|C$ & $D$ is associated with $C$ controlling $U$\\
     \hline
    Partial correlation squared & 0.23& 0.58\\
    CODEC & 0.38 &   $0.48$\\
    KPC (1-NN) & 0.04 & $0.34$\\
    KPC (RKHS) & 0.15 & 0.39\\
    \hline
    \end{tabular}
\end{table}

\noindent\textbf{Election data (histogram-valued response)}:
Consider the 2017 Korean presidential election data collected by \url{https://github.com/OhmyNews/2017-Election},
which has been analyzed in the recent paper \cite{jeon2020additive}.
The data set consists of the voting results earned by the top five candidates from 250 electoral districts in Korea.
Since the top three candidates from three major parties representing progressivism, conservatism and centrism earned most of the votes,
we will focus on the proportion of votes earned by each of these three candidates among them, i.e., $Y=(Y_1,Y_2,Y_3)$, where $Y_1,Y_2,Y_3\geq 0$ and $Y_1+Y_2+Y_3 = 1$.
The demographic information \emph{average age} ($X_1$), \emph{average years of education} ($X_2$), \emph{average housing price per square meter} ($X_3$) and \emph{average paid national health insurance premium} ($X_4$) are available for each electoral district.
Note that $Y$ can be viewed as a histogram-valued variable for which the following two characteristic kernels on $[0,\infty)^3$ are available~\cite{fukumizu2009characteristic}:
\begin{enumerate}
    \item $k_1(a,b) = \prod_{i=1}^3 (a_i+b_i+1)^{-1},$
    \item $k_2(a,b) = e^{-\sum_{i=1}^3 \sqrt{a_i+b_i}},$
\end{enumerate}
where $a=(a_1,a_2,a_3)$ and $b=(b_1,b_2,b_3)$. Since $Y\in\R^3$, we can also use a Gaussian kernel. Note that $Y$, taking values in the probability simplex, is an example of compositional data. Besides political science, such data are prevalent in many other fields such as sedimentology \cite{hijazi2009dirichlet,aitchison1986composition}, hydrochemistry \cite{otero2005relative}, economics \cite{morais2018market}, and bioinformatics \cite{xia2013microbiome,chen2016two,shi2016regression}.
Regression methods have been developed with compositional data as covariates \cite{shi2016regression,Lin2014compositional,susin2020variable}, as response \cite{aitchison1986composition,lyengar2002semi,hijazi2009dirichlet,glahn2012simplicial,tsagris2015regression,tsagris2020alpha}, or as both covariates and response \cite{chen2017complinreg}.
But there are few variable selection methods proposed for data with compositional response and Euclidean covariates, as in our case here.
Since our method tackles random variables taking values in general topological spaces, it readily yields a variable selection method for this problem.

The variables selected by KFOCI with different kernels and $K$-NN graphs are given in Table~\ref{tab:KFOCIHistselected}.
\begin{table}
    \centering
    \caption{The selected variables with different kernels and different $K$-NN graphs.}
    \label{tab:KFOCIHistselected}
    \begin{tabular}{cccccc} 
     \hline
     Kernels & 1-NN & 2-NN  & 3-NN & 4-NN & 5-NN\\
     \hline
     Gaussian &1 2 3 4 &1 2 4 3&1 2 4 3&1 2 4 & 1 2 4\\
     $k_1$ & 1 2& 1 2& 1 2 4 3& 1 2 4 3& 1 2 4\\
     $k_2$ & 1 2&1 2&2 1 4 3& 1 2 4 3& 1 2 4\\
     \hline
    \end{tabular}
\end{table}
It can be seen that in all cases, $X_1$ and $X_2$ are selected, indicating that they may be more relevant for predicting $Y$ than $X_3$, $X_4$.
This agrees with \cite[Figure 3]{jeon2020additive}, where an additive regression was fit to the election data.
It was also found that as people are more educated, from `low' to `medium' educational level, their political orientation becomes more conservative,
but interestingly it is reversed for people as they move from `medium' to `high' educational level \cite{jeon2020additive}.
This nonlinear relationship is also captured by KPC, as $X_2$ is always selected.
Note that $X_3$ and $X_4$ are both measures of wealth.
It is natural to conjecture if one of them, say $X_4$ alone, would be enough in predicting $Y$.
We can answer this by examining how large $\rho^2 (Y,X_3|X_{-3})$ is: If the estimate of $\rho^2 (Y,X_3|X_{-3})$ is close to 0 it provides evidence that $Y$ is conditionally independent of $X_3$ given $X_{-3} = \{X_1,X_2,X_4\}$.
The estimates of $\rho^2 (Y,X_i|X_{-i})$ are given in Table~\ref{tab:2estimates}.
\begin{table}
    \centering
    \caption{The estimates of $\rho^2 (Y,X_i|X_{-i})$ with different $K$-NN graphs and regularizers.} 
    \label{tab:2estimates}
    \begin{tabular}{cccccc} 
     \hline
     Estimators & 1-NN & 2-NN  & 3-NN & 4-NN & 5-NN\\
     \hline
     $\hat{\rho^2}(Y,X_1| X_{-1})$ &0.07 &0.15  &0.17 &0.14 &0.13\\
     $\hat{\rho^2}(Y,X_2|X_{-2})$ & 0.19& 0.15 & 0.09& 0.05& 0.04\\
     $\hat{\rho^2}(Y,X_3|X_{-3})$ &0.10& 0.06& 0.01& -0.02& -0.02\\
     $\hat{\rho^2}(Y,X_4|X_{-4})$ &0.07& 0.07& 0.04& 0.03& 0.01\\
     \hline
     Estimators  & $\varepsilon = 10^{-3}$  & ${\varepsilon = 10^{-4}}$ & $\varepsilon = 10^{-5}$ & $\varepsilon = 10^{-6}$&$\varepsilon = 10^{-7}$\\
     \hline
     $\tilde{\rho^2}(Y,X_1|X_{-1})$ &0.08& 0.15 &0.19 &0.23 &0.28\\
     $\tilde{\rho^2}(Y,X_2|X_{-2})$ & 0.03 & 0.07 &0.11&0.17&0.22\\
     $\tilde{\rho^2}(Y,X_3|X_{-3})$  &0.04 & 0.06 &0.10 &0.15&0.18\\
     $\tilde{\rho^2}(Y,X_4|X_{-4})$& 0.03 & 0.05 &0.08 &0.13&0.18\\
     \hline
    \end{tabular}
\end{table}
It can be seen that for $i=3,4$, the conditional association of $Y$ and $X_i$ given all other variables is weaker than that of $i=1,2$; this agrees with the intuition that $X_3$ and $X_4$ are both capturing the same latent factor --- wealth. In~\cref{tab:2estimates} we also give values of $\tilde{\rho^2}$ as we change $\varepsilon$. For a suitably chosen $\varepsilon$, say $\varepsilon=10^{-4}$, $\tilde{\rho^2}$ is similar to $\hat{\rho^2}$.

\appendix
\section{Some general discussions}\label{sec:gendis}
In this section we elaborate on some parts of the main text that were initially deferred so as not to impede the flow of the paper. 
	
\subsection{Some remarks}\label{Sec:Remarks}
\begin{remark}[About Assumption \ref{assump:separable_RKHS}]\label{rem:Separable} 
Assumption \ref{assump:separable_RKHS} is needed to ensure following: (a) The feature map $y\mapsto k(y,\cdot)$ is measurable\footnote{Note that the canonical feature map $y\mapsto k(y,\cdot)$ is measurable if $\h_\Y$ is separable (see e.g.,~\cite[Lemma 4.25]{SVM}), and 
    $\mathcal{H}_\Y$ is separable if $\Y$ is separable and $k(\cdot,\cdot)$ is continuous on $\Y\times\Y$ (see e.g.,~\cite[Theorem 7]{hein2004kernels}).}; (b) any Borel measurable function $g:\Y \to \h_\Y$ is strongly measurable\footnote{A function is strongly measurable if it is Borel measurable and has a separable range.} (when $\h_\Y$ is separable) and the Bochner integral (see e.g.,~\cite[Appendix E]{Cohn2013measure} for its formal definition) $\E [g(Y)]$ is well-defined whenever $\E \|g(Y)\|_{\h_\Y}<\infty$\footnote{Thus, if $\E\sqrt{k(Y,Y)}<\infty$, $\mu_Y :=\E k(Y,\cdot)$ is well-defined as a Bochner integral. As for any bounded linear functional $T:\h_\Y\to \R$, we have $\E \left[T(g(Y))\right]=T\left(\E[g(Y)]\right)$, we can write $\langle \mu_Y, f \rangle_{\h_\Y} = \langle \E[k(Y, \cdot)], f \rangle_{\h_\Y}   = \E [\langle k(Y,\cdot),f \rangle_{\h_\Y} ] = \E f(Y)$.};
(c) the relevant cross-covariance operators are Hilbert-Schmidt (see Section~\ref{sec:Func-Ana} for the details). Thus, for simplicity, we assume the more natural Assumption \ref{assump:separable_RKHS} rather than the conditions (a)-(c) above.
\end{remark}

\begin{remark}\label{footnote:degenerate_imply_function}
Observe that if for a.e.~$x$, $Y|X=x$ is degenerate, then $Y$ is a measurable function of $X$ a.s. To see this, let $Q(x,\cdot)$ be a regular conditional distribution of $Y$ given $X=x$. Then there exists $A\subset \X$ such that $\p(X\in A)=1$ and $Q(x,\cdot)$ ---  the transition kernel --- is degenerate for all $x\in A$.  Let $f(x) \in \Y$ be the support (which is a single element in $\Y$) of $Q(x,\cdot)$ if $x\in A$, and some fixed $y_0 \in \Y$ if $x\in A^c$. Then $f:\X \to \Y$ is a measurable function because up to a null set $f^{-1}(B) = \{x \in \X : Q(x,B)=1\}$ is a measurable set for any measurable $B \subset \Y$. Now, note that $Y = f(X)$ a.s. so $Y$ is a measurable function of $X$. 
\end{remark}

\subsection{Uncentered CME estimators}\label{sec:Un-CME}

We used the `centered' CMEs (see~\eqref{eq:CME}) to derive an expression for the estimator of $\rho^2$ since that requires less restrictive assumptions; see~\cite{klebanov2019rigorous}. But using `uncentered CMEs', i.e., using uncentered (cross)-covariance operators to construct an expression of the CME (as in~\eqref{eq:CME}) is also possible. Under appropriate assumptions, the uncentered (cross)-covariance operators yield the following uncentered CME formula (see~\cite[Theorem 5.3]{klebanov2019rigorous}):
\begin{equation}\label{eq:Uncen-CME}
\mu_{Y|X=x} = \left({^u C}_X^\dagger {^uC}_{XY}\right)^* k_\mathcal{X}(x,\cdot),
\end{equation}
where ${^uC}_{XY}$ (resp. ${^u C}_X$) is the uncentered cross-covariance (resp. covariance) operator defined as $\langle f,{^uC}_{XY}g \rangle = \mathbb{E}\left[f(X)g(Y)\right]$. The same methodology as in~\cref{sec:Est-kernel} 
    shows that one can simply replace every occurence $\tilde{K}$ by the corresponding $K$ (e.g., $\tilde{K}_X$ by $K_X$) to obtain the following simplified formula: $$\tilde{\rho^2_u}=\frac{{\rm Tr}(M^\top K_Y M)}{{\rm Tr}(N^\top K_Y N)},$$ where $M= {K}_X \left( {K}_X + \varepsilon n I\right)^{-1} - {K}_{\ddot{X}} \left( {K}_{\ddot{X}} + \varepsilon n I\right)^{-1}$ and $N=I-{K}_X \left( {K}_X + \varepsilon n I\right)^{-1}$.

We see that the centered estimator $\tilde{\rho^2}$ (in~\eqref{eq:eta-Cen-Est}) can be obtained by putting a `tilde' on all the kernel matrices (i.e., $K_X, K_Y, K_{\ddot{X}}$) in $\tilde{\rho^2_u}$, and this can be viewed as performing a `centralization' of the feature maps of the corresponding kernel, i.e., if $K = \Phi \Phi^\top$, then $\tilde{K} = HKH = H \Phi \Phi^\top H = \Phi_c \Phi_c^\top$, with $\Phi_c:=H\Phi$ being the centralized feature.

\begin{remark}[Sufficient conditions for uncentered CME]\label{rk:uncenter}
A sufficient condition for~\eqref{eq:Uncen-CME} to be valid is: For any $g\in\mathcal{H}_\mathcal{Y}$, a version of $\mathbb{E}[g(Y)|X=\cdot]\in\mathcal{H}_\mathcal{X}$. In the case when $X$ and $Y$ are independent, $\E[g(Y)|X=\cdot]$ is~a constant ($P_X$-a.e.). But it is well-known that the RKHS with the Gaussian kernel does not contain nonzero constant functions whenever $\X$ has non-empty interior; see e.g.,~\cite{minh2010some}. This fact provides an example where the CME is not guaranteed to be expressed using uncentered (cross)-covariance operators (as in~\eqref{eq:Uncen-CME}). Note that Assumption~\ref{assump:CME} always holds when $X$ and $Y$ are independent. However, in general, sufficient conditions for explicit expressions of CMEs in terms of centered and uncentered (cross)-covariance operators (as in~\eqref{eq:CME}) are usually restrictive and hard to verify (see~\cite{klebanov2019rigorous}).
\end{remark}

The following result (see Section \ref{kernel_ridge} for a proof) shows that $\tilde{\rho^2_u}$ has an interesting connection to kernel ridge regression.

\begin{proposition}\label{prop:kernel_ridge}
    Suppose $\Y = \mathbb{R}$ is equipped with the linear kernel $k (u,v) = u v$, for $u,v \in \Y$. Then,
            $$\tilde{\rho^2_u} = \frac{\|\hat{\mathbf{Y}}_x - \hat{\mathbf{Y}}_{xz}\|^2}{\|\hat{\mathbf{Y}}_x - \mathbf{Y}\|^2},$$ where $\mathbf{Y}=(Y_1,\cdots, Y_n)$, $\hat{\mathbf{Y}}_x$ is the kernel ridge regression estimator when regressing $Y$ on $X$ and $\hat{\mathbf{Y}}_{xz}$ is the kernel ridge regression estimator when regressing $Y$ on $\ddot{X}$, both with regularization parameter $n\varepsilon$.
\end{proposition}

\subsection{Invariance and continuity of $\rho^2$}\label{sec:Inv-Cont}
Informally, \emph{invariance} means that the measure should be unaffected under a ``suitable'' class of transformations and \emph{continuity} means that whenever a sequence of measures $P_n$ converges to $P$ (in an ``appropriate" sense), the sequence of values of the dependence measure for $P_n$'s should converge to that of $P$. We show, in the following result (proved in~\cref{sec:pf_inv_conti}) that $\rho^2$ satisfies these properties for a class of kernels.
\begin{proposition}\label{prop:inv_conti}
The following two properties hold for $\rho^2(Y,Z|X)$:
\begin{enumerate}
    \item \emph{Invariance.} $\rho^2(Y,Z|X)$ is invariant to any bijective transformation of $X$, and any bijective transformation of $Z$. If $\Y =\R^d$ and the kernel is of the form~\eqref{eq:h-1-2-3} then $\rho^2$ is also invariant to orthogonal transformations and translations of $Y$.
    \item \emph{Continuity.}  Let $\Y$ be a separable metric space and let $(X_n,Y_n, Z_n) \sim P_n$. Let $(Y_{n,1},Y_{n,1}')$, $(Y_{n,2},Y_{n,2}')$ be generated from the distribution of $(X_n,Y_n, Z_n)$ as described in~\eqref{eq:Y_1} and~\eqref{eq:Y_2}. If $(X_n,Y_n,Z_n, Y_{n,1},Y_{n,1}', Y_{n,2},Y_{n,2}')\overset{d}{\to} (X,Y,Z, Y_{1},Y_{1}', Y_{2},Y_{2}')$ where $(X,Y,Z) \sim P$, and $\limsup_{n\to\infty}\E [k^{1+\varepsilon}(Y_{n,i},Y_{n,i}')]<\infty$ for $i=1,2$, and $\limsup_{n\to\infty}\E [k^{1+\varepsilon}(Y_{n},Y_{n})]<\infty$ for some $\varepsilon>0$, then $\rho^2(Y_n,Z_n|X_n)\to\rho^2(Y,Z|X)$ as $n\to\infty$.  
\end{enumerate}
\end{proposition}

\subsection{Approximate computation of $\tilde{\rho^2}$ using incomplete Cholesky decomposition}\label{sec:approx_compute}

For fast computation, we can approximate each of the kernel matrices by \emph{incomplete Cholesky decomposition}. For example, we can approximate $K_X\approx L_1L_1^\top$ where $L_1\in\mathbb{R}^{n\times d_1}$ is the incomplete Cholesky decomposition of the kernel matrix\footnote{Note that if the linear kernel is used and $X_i \in \R^{d_1}$, then the decomposition $K_X=\mathbf{X} \mathbf{X}^\top$ (here $\mathbf{X} := [X_1, \ldots, X_n]^\top$ is the data matrix of the $X_i$'s) is straight-forward and exact.} $K_X$ which can be computed in $O(n d_1^2)$ time (here $d_1 \le n$) even without the need to compute and store the full gram matrix $K_X$; see e.g.,~\cite{bach2002kernel}. Let $\tilde{L}_1:=HL_1$ denote the centralized feature matrix (recall that $H :=I-\frac{1}{n}\mathbf{1}\mathbf{1}^\top$ is the centering matrix). Then $\tilde{K}_X\approx \tilde{L}_1 \tilde{L}_1^\top$. Similarly,  we can approximate each of the centralized kernel matrices $\tilde{K}_{\ddot{X}}\approx \tilde{L}_2 \tilde{L}_2^\top$, $\tilde{K}_Y\approx \tilde{L}_3 \tilde{L}_3^\top$, where $L_i\in\mathbb{R}^{n\times d_i}$ for $i=2,3$.

To compute the numerator of $\tilde{\rho^2}$ (see~\eqref{eq:Compute-Rho}), 
    by the Woodbury matrix identity, we can approximate $(\tilde{K}_X + n\varepsilon I)^{-1}$ by
    $$(\tilde{L}_1\tilde{L}_1^\top + n\varepsilon I_{n})^{-1} = \frac{1}{n\varepsilon} I_n - \frac{1}{n\varepsilon} \tilde{L}_1 (n\varepsilon I_{d_1} + \tilde{L}_1^\top \tilde{L}_1)^{-1} \tilde{L}_1^\top.$$
The same strategy applied to $(\tilde{K}_{\ddot{X}} + n\varepsilon I)^{-1}$ shows $M$ can be approximated by $\check{M} := \tilde{L}_1(n\varepsilon I_{d_1} + \tilde{L}_1^\top \tilde{L}_1)^{-1} \tilde{L}_1^\top - \tilde{L}_2(n\varepsilon I_{d_2} + \tilde{L}_2^\top \tilde{L}_2)^{-1} \tilde{L}_2^\top$. Thus
$${\rm Tr}(M^\top \tilde{K}_Y M) \approx {\rm Tr}(\check{M}^\top \tilde{L}_3 \tilde{L}_3^\top \check{M}) = \|\tilde{L}_3^\top \check{M}\|_{F}^2,$$
where $\|\cdot\|_{F}$ denotes the Frobenius norm. 

For the denominator of $\tilde{\rho^2}$, note that $N=n\varepsilon(\tilde{K}_X+n\varepsilon I)^{-1} \approx I_n -  \tilde{L}_1 (n\varepsilon I_{d_1} + \tilde{L}_1^\top \tilde{L}_1)^{-1} \tilde{L}_1^\top =: \check{N}$. Thus, ${\rm Tr}(N^\top \tilde{K}_Y N)\approx \| \tilde{L}_3^\top \check{N}\|_F^2$. Combining the calculations above, we see that the approximate version of $\tilde{\rho^2}$ can be computed in $O\left(n \max_{1\leq i \leq 3} d_i^2\right)$ time.

\subsection{Assumptions on the geometric graph functional}\label{sec:Assump-Graph}
We will use the following assumptions on the graph functional $\mathcal{G}$; these assumptions were also made in~\cite[Section 3]{deb2020kernel}. 
	
\begin{assump}\label{assump:conv_nn}
    Given the graph $\mathcal{G}_n$, let $N(1),\ldots ,N(n)$ be  independent random variables where $N(i)$ is a uniformly sampled index from among the (out-)neighbors of $X_i$ in $\mathcal{G}_n$, i.e., $\{j:(i,j)\in \emgn \}$. Assume that $\X$ is a metric space with metric $\rho_\X$, and that
	$$\rho_{\X}(X_1,X_{N(1)})\overset{p}{\to}0,\quad {\rm as}\ n\to\infty.$$
\end{assump}
\begin{assump}\label{assump:degree}
    Assume that there exists a deterministic positive sequence $r_n\geq 1$ (may or may not be bounded), such that $$\min_{1\leq i\leq n} d_i\geq r_n\qquad \mbox{a.s.}$$
    Let $\mathcal{G}_{n,i}$ denote the graph obtained from $\mathcal{G}_n$ by replacing $X_i$ with an i.i.d. random element $X_i'$. Assume that there exists a deterministic positive sequence $q_n$ (may or may not be bounded), such that
    $$\max_{1\leq i\leq n}\max\{|\mathcal{E}(\mathcal{G}_n)\setminus \mathcal{E}(\mathcal{G}_{n,i})|,|\mathcal{E}(\mathcal{G}_{n,i})\setminus \mathcal{E}(\mathcal{G}_{n})|\}\leq q_n\ \ a.s.\qquad {\rm and}\qquad \frac{q_n}{r_n} = O(1).$$
\end{assump}
\begin{assump}\label{assump:degupbd}
There exists a deterministic sequence $\{t_n\}_{n\ge 1}$ (may or may not be bounded) such that the vertex degree (including both in- and out-degrees for directed graphs) of every point $X_i$ (for $i=1,\ldots, n$) is bounded by $t_n$, and
		$\frac{t_n}{r_n}=O(1)$.
\end{assump}

\subsection{A review of some concepts from functional analysis}\label{sec:Func-Ana}
We start with a brief review of some notions from functional analysis that will be important; see~\cite{rynne2008linear,aubin2011applied} for a detailed study of these concepts. By a {\it bounded linear operator} $A$ from a Hilbert space $\h$ to itself we mean $A:\h \to \h$ is such that, for some $L >0$, $\|Av\|_\h \le L \|v\|_\h$, for all $v \in \h$. $A$ is {\it nonnegative} if $\langle u, A u\rangle_\h \ge 0$, for all $u \in \h$. There is a unique bounded operator $A^{*}:\h \to \h$, called the {\it adjoint} of $A$, such that $\langle u, A v\rangle_\h = \langle A^* u,  v\rangle_\h$, for all $u,v \in \h$. We say that $A$ is {\it self-adjoint} if $A^{*}=A$. Suppose that $\h$ is separable with orthonormal basis $\{e_i\}_{i\ge 1}$. Then the {\it trace} of a non-negative $A$ is defined as $\sum_i \langle A e_i, e_i \rangle_\h$. For a nonnegative operator $A$ if $\sum_i \langle A e_i, e_i \rangle_\h < \infty$ then $A$ is said to be a {\it trace-class} operator. A {\it compact} operator is a linear operator $L$ from a Banach space $\F$ to another Banach space $\G$, such that the image under $L$ of any bounded subset of $\F$ is a relatively compact subset (has compact closure) of $\G$. Such an operator is necessarily a bounded operator, and so is continuous.

Let $(\ran A) := \{Av: v \in \h\}$ denote the {\it range} of the operator $A$ and $\ker A := \{v \in \h: Av =0\}$ denote the {\it kernel} of $A$. Let $A^{\dagger}$ denote the {\it Moore-Penrose inverse} of $A$ (see e.g.,~\cite[Definition 2.2]{engl1996regularization}). Note that $A^{\dagger}$ is defined as the unique linear extension of $(A|_{(\ker A)^{\perp}})^{-1}$ to a (possibly unbounded) linear operator on
$(\ran A)\oplus (\ran A)^\perp$ such that ${\ker}A^\dagger = (\ran A)^\perp$.

Let $\F$ and $\G$ be separable Hilbert spaces. Let $\{f_i\}_{i \in I}$ to be an orthonormal basis for $\F$, and let $\{g_j\}_{j \in J}$ be an orthonormal basis for $\G$; here $I$ and $J$ are indexing sets being either finite or countably infinite. 
\begin{defn}[Hilbert-Schmidt operators]\label{defn:HS}
The Hilbert-Schmidt norm of a compact operator $L: \G \to \F$ is defined to be 
\begin{equation}\label{eq:HS-Norm}
    \|L\|_{\HS}^2 := \sum_{j \in J}\|Lg_j\|_\F^2.
    \end{equation}
The operator $L$ is {Hilbert-Schmidt} when this norm is finite.
The Hilbert-Schmidt operators mapping from $\G$ to $\F$ form a Hilbert space, written $\HS(\G, \F)$, with inner product 
\begin{equation}\label{eq:HS-Opt}\langle L,M\rangle_{\HS} := \sum_{j \in J} \langle L g_j,M g_j \rangle_\F,
\end{equation} 
which is independent of the orthonormal basis chosen. Here $L: \G \to \F$ and $M : \G \to \F$ are two Hilbert-Schmidt operators.  
\end{defn}

\begin{defn}[Tensor product space]\label{tensor_product_Hilbert_space}
The tensor product of two vector spaces $V,W$ over the same field $F$, denoted by $V\otimes W$, is defined as the quotient space obtained from free vector space $F(V\times W)/\sim$, where the bilinear equivalence relation $\sim$ satisfies
        \begin{enumerate}
            \item $(v,w)+(v',w)\sim (v+v',w)$ and $(v,w)+(v,w')\sim (v,w+w')$.
            \item $c(v,w)\sim(cv,w)\sim(v,cw)$.
        \end{enumerate}
    If $V$ and $W$ are both Hilbert space, an inner product can be defined by linearly extending
    $$\langle v_1\otimes w_1,v_2\otimes w_2\rangle_{V\otimes W} :=\langle v_1,v_2\rangle_V \cdot \langle w_1,w_2\rangle_W.$$
    We call the completion under this inner product as the tensor product space of the two Hilbert spaces, which is itself a Hilbert space.
    Suppose both $V$ and $W$ are separable. If $\{f_i\}$ is an orthonormal basis of $V$ and $\{e_i\}$ is an orthonormal basis of $W$, then
    $\{f_i\otimes e_j\}$ is an orthonormal basis of $V\otimes W$, and any element in $\h_\X\otimes\h_\Y$ can be written as $\sum_{i,j}a_{ij} f_i\otimes e_j$, with $\sum_{i,j}a_{ij}^2<\infty$.
    There exists an isometric isomorphism $\Phi: V\otimes W\to \HS(W,V)$ by linearly extending
    $\Phi(f\otimes g)(h) = \langle g,h\rangle_{W} f$. More specifically,
    \begin{equation}\label{isomorphism}
        \Phi(\sum_{i,j}a_{ij} f_i\otimes e_j)(h)=\sum_{i}\sum_j a_{ij}\langle e_j,h\rangle_{W} f_i.
    \end{equation}
    Since $\sum_i \left(\sum_j |a_{ij}\langle e_j,h\rangle_{W}| \right)^2\leq \sum_i \left(\sum_j a_{ij}^2\right) \left(\sum_j \langle e_j,h\rangle_{W}^2 \right)=\|h\|_{W}^2 \sum_{i,j}a_{ij}^2<\infty$, the right-hand-side above is well-defined.
    It is isometric since
$$\begin{aligned}
        \Big\| \Phi \big(\sum_{i,j}a_{ij}f_i\otimes e_j \big) \Big\|_{\HS}^2 &= \sum_k \Big \| \sum_{i}\sum_j a_{ij}\langle e_j,e_k\rangle_{W} f_i \Big \|_{V}^2\\
        &= \sum_k \Big \| \sum_{i} a_{ik} f_i \Big \|_{V}^2\\
        &=\sum_k \sum_i a_{ik}^2 = \Big \| \sum_{i,j} a_{ij} f_i\otimes e_j \Big \|_{V\otimes W}^2.
    \end{aligned}$$
    $\Phi$ is surjective because its range includes the dense set of finite-rank operators. The bounded inverse theorem \cite[Theorem 4.43]{rynne2008linear} then implies $\Phi$ is an isometric isomorphism. See \cite[Section 12]{aubin2011applied} for more detials regarding the Hilbert-Schmidt operators and Hilbert tensor products.
\end{defn}
\begin{remark}[Cross-covariance operator in tensor-product space]\label{rem:C-Cov} Here we show that~\eqref{eq:Cov} holds. By~\eqref{isomorphism}, for $a\in\h_\X\otimes\h_\Y$, $\langle f,\Phi(a)g\rangle_{\h_\X}=\langle a,f\otimes g\rangle_{\h_\X\times\h_\Y}$.
Denote the right-hand side of \eqref{eq:Cov} by $C$. Then 
$$\begin{aligned}
    \langle f,C g\rangle_{\h_\X} &= \langle C,f\otimes g\rangle_{\h_\X\otimes \h_\Y}\\
    &=\langle \E \left[ k_\X (X,\cdot)\otimes k (Y,\cdot)\right], f\otimes g\rangle_{\h_\X\otimes \h_\Y} -  \langle \E\left[k_\X (X,\cdot)\right],f\rangle_{\h_\X} \langle \E\left[k(Y,\cdot)\right],g\rangle_{\h_\Y}\\
    &=\E \left[\langle k_\X (X,\cdot)\otimes k (Y,\cdot), f\otimes g\rangle_{\h_\X\otimes \h_\Y} \right] -  \E\left[\langle k_\X (X,\cdot),f\rangle_{\h_\X}\right]\cdot \E\left[ \langle k (Y,\cdot),g\rangle_{\h_\Y}\right]\\
    &=\E\left[f(X)g(Y) \right] - \E [f(X)]\, \E [g(Y)].
\end{aligned}$$
In the third line we used the fact that if $X\in \X$ is Bochner integrable, then $\E\varphi (X)=\varphi(\E X)$ for any bounded linear functional $\varphi\in \X^*$ (see e.g.,~\cite[Appendix E]{Cohn2013measure}).
\end{remark}

\subsection{Model-X FDR control}\label{sec:ModelXFDR}
Consider the variable selection in regression problem described in~\cref{sec:application} where $X = (X_1,\ldots, X_p)$ is the predictor and $Y$ is the response variable. Our estimators of KPC can be used to develop model-free (nonparametric) variable selection methods that control \emph{false discovery rate} (FDR), under the model-X framework.
Classical FDR control methods often impose some assumption on $Y|X$, such as a Gaussian linear model \cite{barber2015controlling,javanmard2019false}.
The model-X framework \cite{candes2018panning}, on the other hand, allows for a complex nonlinear relationship between $Y$ and $X$, which is exactly the general scenario we study in~\cref{sec:application}.

Given i.i.d.~data $\{(X^{(i)},Y^{(i)})\}_{i=1}^n$ from the regression model~\eqref{eq:Reg-Mdl} let $\mathbf{X} := ((X_{j}^{(i)}))_{n \times p}$ be the $n \times p$ design matrix and let $\mathbf{Y} = (Y^{(1)}, \ldots, Y^{(n)})$. 
In the model-X knockoff framework, the distribution of $X$ is assumed to be known, and a knockoff variable $\tilde{X}$ is constructed satisfying
$[X,\tilde{X}]_{{\rm swap}(S)}\overset{d}{=}[X,\tilde{X}]$,  for all $S\subset \{1,\ldots,p\}$ and $\tilde{X}\indep Y|X$ where ${\rm swap}(S)$ is an operation that swaps $X_j$ with $\tilde{X}_j$ for every $j \in S$. For every $j \in \{1,\ldots, p\}$, a statistic $W_j$ is then constructed to measure the association between $X_j$ and the response $Y$,
such that a large value of $W_j$ provides evidence that $X_j$ is a signal variable.
In the seminal paper \cite{candes2018panning}, $W_j$ was set as the \emph{Lasso coefficient-difference} after fitting a lasso regression of $Y$ on $[X,\tilde{X}]$.
However, this Lasso-based statistic is designed specifically for a linear model and may not be powerful for the general regression model~\eqref{eq:Reg-Mdl}.
Using our statistic $\hat{\rho^2}$, we can set
\begin{equation}\label{eq:W_j}
W_j([\mathbf{X},\tilde{\mathbf{X}}],\mathbf{Y}) :=\hat{\rho^2} (Y,X_j|[X_{-j},\tilde{X}]) - \hat{\rho^2} (Y,\tilde{X}_j|[X,\tilde{X}_{-j}]).
\end{equation}
To maintain finite sample FDR control, $W_j$ should satisfy the \emph{flip-sign property}, i.e., for any $S\subset \{1,\ldots,p\}$,
$$W_j([\mathbf{X},\tilde{\mathbf{X}}]_{{\rm swap}(S)},\mathbf{Y})=\left\{\begin{aligned}
    &W_j([\mathbf{X},\tilde{\mathbf{X}}]_{{\rm swap}(S)},\mathbf{Y}),\quad &j\notin S,\\
    &-W_j([\mathbf{X},\tilde{\mathbf{X}}]_{{\rm swap}(S)},\mathbf{Y}),\quad &j\in S.\\
\end{aligned}\right.$$

The \emph{flip-sign property} is easily satisfied if the $K$-NN graph uses distance metrics treating $X_j$ and $\tilde{X}_j$ in the same way\footnote{More specifically: (i) The distance matrix $D_{[\mathbf{X},\tilde{\mathbf{X}}]}$ (i.e., $(D_{[\mathbf{X},\tilde{\mathbf{X}}]})_{ij}$ is the distance between the $i$-th row and the $j$-th row of $[\mathbf{X},\tilde{\mathbf{X}}]$) does not change if the $j$-th columns of $\mathbf{X}$ and $\tilde{\mathbf{X}}$ are swapped; (ii) with $\mathbf{X}_{-j}$ and $\tilde{\mathbf{X}}_{-j}$ fixed, $D_{[\mathbf{X},\tilde{\mathbf{X}}_{-j}]}$ as a function of $\mathbf{X}_j $ is the same as $D_{[\mathbf{X}_{-j},\tilde{\mathbf{X}}]}$ as a function of $\tilde{\mathbf{X}}_j$; (iii) for $k\neq j$, $D_{[\mathbf{X},\tilde{\mathbf{X}}_{-j}]}$ and $D_{[\mathbf{X}_{-j},\tilde{\mathbf{X}}]}$ do not change if the two columns $\mathbf{X}_k$ and $\tilde{\mathbf{X}}_k$ are swapped. Note that the Euclidean distance satisfies these properties.}.
Note that $\hat{\rho^2}$ in~\eqref{eq:W_j} can also be replaced by $\tilde{\rho^2}$ (from~\cref{sec:Est-kernel}). As long as the kernels treat $X_j$ and $\tilde{X}_j$ in the same way\footnote{More specifically: (i) The kernel matrix $K_{[\mathbf{X},\tilde{\mathbf{X}}]}$ does not change if the $j$-th columns of $\mathbf{X}$ and $\tilde{\mathbf{X}}$ are swapped; (ii) with $\mathbf{X}_{-j}$ and $\tilde{\mathbf{X}}_{-j}$ fixed, $K_{[\mathbf{X},\tilde{\mathbf{X}}_{-j}]}$ as a function of $\mathbf{X}_j $ is the same as $K_{[\mathbf{X}_{-j},\tilde{\mathbf{X}}]}$ as a function of $\tilde{\mathbf{X}}_j$; (iii) for $k\neq j$, $K_{[\mathbf{X},\tilde{\mathbf{X}}_{-j}]}$ and $K_{[\mathbf{X}_{-j},\tilde{\mathbf{X}}]}$ do not change if the two columns $\mathbf{X}_k$ and $\tilde{\mathbf{X}}_k$ are swapped. Note that the Gaussian kernel satisfies (i), (iii) and also (ii) if the bandwidths are the same for $[\mathbf{X},\tilde{\mathbf{X}}_{-j}]$ and $[\mathbf{X}_{-j},\tilde{\mathbf{X}}]$.} (e.g., when using the Gaussian kernel), the FDR control will also be valid.
In such case, if we select $$\hat{S}=\left\{j: W_j\geq \tau:=\min \Big\{t>0:\frac{\#\{j:W_j\leq -t\}}{\#\{j:W_j\geq t\}}\leq q \Big\} \right\},$$ then a modified FDR is controlled, i.e., $\E\left[\frac{|\{j\in\hat{S}\cap \h_0\}|}{|\hat{S}|+1/q} \right]\leq q$ (for $q \in (0,1)$); see~\cite{candes2018panning}. Here $\h_0$ is the set of null variables that are conditionally independent of the response given other predictor variables. If we select
$$\hat{S}_+=\left\{j: W_j\geq \tau_+ :=\min \Big\{t>0:\frac{1+ \#\{j:W_j\leq -t\}}{\#\{j:W_j\geq t\}}\leq q \Big\}\right\},$$ then the FDR is controlled, i.e.,
$\E\left[\frac{|\{j\in\hat{S}\cap \h_0\}|}{\max\{|\hat{S}|, 1\}} \right]\leq q$; see~\cite{candes2018panning}.
Note that the requirement $Y\in\R$ is also not required here.

\section{Proofs}\label{sec:Proofs}

\subsection{Proof of Lemma~\ref{lem:Well-Def}}\label{sec:Well-Def}
First observe that by Assumption \ref{assump:characteristic_kernel}, $\E \left[\|k(Y_1,\cdot)\|_{\h_\Y}^2\right] =\E \left[k(Y,Y)\right]<\infty$.
By the Cauchy-Schwarz's inequality, $\E [\E [k(Y_1,Y_1')|X]]$ is upper bounded by 
$$\E \Big [\E [\|k(Y_1,\cdot)\|_{\h_\Y}\cdot  \|k(Y_1',\cdot)\|_{\h_Y}\ |X] \Big]\leq \E \left[\E \Big[\frac{\|k(Y_1,\cdot)\|_{\h_\Y}^2 + \|k(Y_1',\cdot)\|_{\h_\Y}^2}{2} \big|X\Big]\right]<\infty.$$
Similarly, $\E\left[\E[k(Y_2,Y_2')|X,Z]\right]<\infty$. Hence the numerator of $\rho^2(Y,Z|X)$ in \eqref{eta_formula} is finite.
The denominator is also finite by Assumption \ref{assump:characteristic_kernel}.

Next we show that $\E[{\rm MMD}^2({\delta_Y},P_{Y|X})]\neq 0$.
If $\E[{\rm MMD}^2({\delta_Y},P_{Y|X})]=0$, then equivalently $\E \|k(Y,\cdot)-\E [k(Y,\cdot)|X]\|^2_{\h_\Y}=0$. So, conditional on $X=x$, $k(Y,\cdot)=\E[k(Y,\cdot)|X=x]$ for $P_X$-a.e.~$x$. Since $k$ is characteristic, $y\mapsto k(y,\cdot)$ is injective, which implies that $Y|X=x$ is degenerate for $P_X$-a.e.~$x$. This is a contradiction to Assumption \ref{assump:nondegenrate}.
\qed

\subsection{Proof of Lemma~\ref{lem:eta-2}}\label{sec:eta-2}
Let us first try to explain the notation in definition~\eqref{eq:eta}. From the definition of the MMD, $\rho^2$ in~\eqref{eq:eta} can be re-expressed as: $\rho^2 = \frac{\E \left[ \|\mu_{P_{Y|XZ}} - \mu_{P_{Y|X}} \|_{\h_\Y}^2\right]}{\E \left[\|\mu_{\delta_Y} - \mu_{P_{Y|X}} \|_{\h_\Y}^2\right]}$ where the mean embeddings above have the following expressions:
\begin{eqnarray*}
	\mu_{\delta_Y}(\cdot) & = &  k(\cdot,Y), \\
	\mu_{P_{Y|X}} (\cdot) & = &  \E_{Y_1 \sim Y|X} [k(\cdot,Y_1)] = \E[k(\cdot,Y)|X], \\
	 \mu_{P_{Y|XZ}} (\cdot) & = &  \E_{Y_2 \sim Y|XZ} [k(\cdot,Y_2)] = \E[k(\cdot,Y)|X,Z],
\end{eqnarray*}
where the expectations should be understood as Bochner integrals (see e.g.,~\cite{Diestel1974,Dinculeanu2011}). Using the notation in the statement of the lemma, we have
\begin{eqnarray*}
\E \left[\|\mu_{\delta_Y} - \mu_{P_{Y|X}} \|_{\h_\Y}^2\right] & = & \E \left[\|\mu_{\delta_Y} \|_{\h_\Y}^2\right] + \E\left[\| \mu_{P_{Y|X}} \|_{\h_\Y}^2\right]  - 2 \E \left[\langle\mu_{\delta_Y},  \mu_{P_{Y|X}} \rangle_{\h_\Y}\right] \\
& = & \E[k(Y,Y)]+\E_X  \left[\E[k(Y_1,Y_1')|X]\right]-2\E_{Y,X} [\E[k(Y,Y_1)|X] ] \\
& = & \E[k(Y,Y)] - \E \left[\E[k(Y_1,Y_1')|X]\right],
\end{eqnarray*}
where $Y_1, Y_1' |X=x \stackrel{iid}{\sim} P_{Y|x}$, $Y_1,Y$ are conditionally independent given $X$, and the second equality follows from the observations: $\| \mu_{P_{Y|X}} \|_{\h_\Y}^2  = \E[k(Y_1,Y_1')|X], \; \mbox{and}\; \langle\mu_{\delta_Y},  \mu_{P_{Y|X}} \rangle_{\h_\Y} = \E[k(Y,Y_1)|X].$
Similarly, we can show that ${\rm MMD}^2(P_{Y|XZ},P_{Y|X})$ equals
\begin{eqnarray*}
& & \E\left[\|\E [k(Y,\cdot)|X,Z] - \E [k(Y,\cdot)|X]  \|^2_{\h_\Y} \right] \\
& = & \E\left[\E[k(Y_2,Y_2')|X,Z]\right]+ \E\left[\E[k(Y_1,Y_1')|X]\right] - 2\E\left[\E_{Y_1\sim Y|X,Y_2\sim Y|X,Z} k(Y_1,Y_2)\right]\\
& = & \E\left[\E[k(Y_2,Y_2')|X,Z]\right] - \E\left[\E[k(Y_1,Y_1')|X]\right].
\end{eqnarray*}
This proves the result.
\qed


\subsection{Proof of Theorem~\ref{thm:Eta}}\label{sec:Eta}
The proof is divided into three parts.
    
\noindent\textbf{Step 1.} We will first show that $\rho^2(Y,Z|X)\in[0,1]$. Observe that $\rho^2 \geq 0$ is clear. To show that $\rho^2 \le 1$, we will use the following result --- a version of Jensen's inequality (see~\cite[Theorems 3.6 and 3.8]{perlman1974jensen} for its proof).
\begin{lemma}[Jensen's inequality]\label{lem:Jensens} Let $\W$ be a real Banach space, $W$ be a Bochner integrable random variable taking value in $\W$, and $g:\W \to \mathbb{R}$ be a lower-semicontinuous convex function such that $g(W)$ is integrable. Then 
    $$g(\E W)\leq \E g(W).$$
    If $g$ is strictly convex\footnote{By strictly convexity we mean: $g(\lambda x +(1-\lambda) y) < \lambda g(x) + (1-\lambda) g(y),\ \forall x\neq y,\lambda\in(0,1)$.}
    and $\p(W=\E W)<1$, then $g(\E W) < \E g(W)$.
\end{lemma}

Now, observe that,
    $$\begin{aligned}
        \E[{\rm MMD}^2({\delta_Y},P_{Y|X})] &= \E \|k(Y,\cdot)-\E [k(Y,\cdot)|X]\|_{\h_\Y}^2\\
        &=\E\left[\E[\|k(Y,\cdot)-\E(k(Y,\cdot)|X)\|^2_{\h_\Y}|X,Z] \right]\\
        &\geq \E\left[\|\E[k(Y,\cdot)-\E(k(Y,\cdot)|X)|X,Z  ]\|^2_{\h_\Y} \right]\quad (\mathrm{Jensen's \; inequality})\\
        &=\E\left[\|\E(k(Y,\cdot)|X,Z)-\E(k(Y,\cdot)|X)  \|^2_{\h_\Y} \right]\\
        &=\E[{\rm MMD}^2(P_{Y|XZ},P_{Y|X})].
    \end{aligned}$$
    where we have applied the above Jensen's inequality to the function\footnote{We have the following lemma: \begin{lemma}
If $\mathcal{H}$ be a real Hilbert space, then $f:\h \to \R$ defined as $f: x \mapsto \|x\|_\h^2$ is strictly convex.
\end{lemma}
\begin{proof} Let $x\neq y \in \h$ and $\lambda \in (0,1)$. Then, $f(\lambda x +(1-\lambda) y) < \lambda f(x) + (1-\lambda) f(y)$ if and only if 
    $$\begin{aligned}
            &\quad \|\lambda x +(1-\lambda) y\|^2_\h< \lambda \|x\|^2_\h + (1-\lambda) \|y\|^2_\h\\
        &\Leftrightarrow 2\lambda(1-\lambda)\langle x,y\rangle_\h < \lambda(1-\lambda)(\|x\|^2_\h+\|y\|_\h^2)\\
        &\Leftrightarrow 0<\|x-y\|^2_\h,
    \end{aligned}$$
    which holds as $x\neq y$. This proves the claim.
\end{proof}}   $g: f \mapsto \|f\|_{\h_\Y}^2$ and $W := k(Y,\cdot)-\E [k(Y,\cdot)|X]$.
Hence $\rho^2\leq 1$. 

\noindent\textbf{Step 2.} Next we show that $\rho^2 = 1$ if and only if $Y$ is a measurable function of $Z$ and $X$.

If $\rho^2 = 1$, then for a.e.~$x,z$, the above Jensen's inequality attains equality, which means that (Lemma~\ref{lem:Jensens}):
    $$\p\left( k(Y,\cdot) = \E_{Y|X=x,Z=z}k(Y,\cdot) \big|X=x,Z=z\right)=1.$$
Hence given $X=x,Z=z$, $Y$ is degenerate. Thus, $Y$ is a measurable function of $X,Z$ (Remark \ref{footnote:degenerate_imply_function}).

Conversely, if $Y=f(X,Z)$ for some measurable function $f$, then
    $$\begin{aligned}
        \E[{\rm MMD}^2(P_{Y|XZ},P_{Y|X})] &=\E[{\rm MMD}^2(\delta_{f(X,Z)},P_{Y|X})] =\E[{\rm MMD}^2(\delta_{Y},P_{Y|X})],
    \end{aligned}$$
 and so $\rho^2(Y,Z|X)=1$.

\noindent\textbf{Step 3.} Now we have to show that $\rho^2(Y,Z|X) = 0$ if and only if $P_{Y|XZ}=P_{Y|X}$ a.s.

If $P_{Y|XZ}=P_{Y|X}$ a.s., then $\E[{\rm MMD}^2(P_{Y|XZ},P_{Y|X})]  = 0$, and thus $\rho^2(Y,Z|X)=0$.

Conversely, if $\rho^2(Y,Z|X) = 0$, since $k$ is characteristic, $P_{Y|XZ}=P_{Y|X}$ a.s. \qed


\subsection{Proof of~\cref{prop:class_parcor}}\label{pf:1}
For notational clarity, write $r \equiv \rho_{YZ\cdot X}$. By assumption, there exist $\mu_1,\mu_2,\sigma_1,\sigma_2$ such that:
$$(Y,Z)^\top |X \sim N\left(\begin{pmatrix}
    \mu_1\\
    \mu_2
\end{pmatrix},\begin{pmatrix}
    \sigma_1^2 & r\sigma_1\sigma_2\\
    r\sigma_1\sigma_2 & \sigma_2^2\\
\end{pmatrix} \right).$$
Further, $Y|X,Z\sim N\left(\mu_1 + r\frac{\sigma_1}{\sigma_2}(Z-\mu_2),(1-r^2)\sigma_1^2\right)$.
Letting $\xi \sim N(0,2\sigma_1^2)$,  $\rho^2(Y,Z|X)$ can be simplified to
$$\rho^2(Y,Z|X)=\frac{\E [h_3 (\sqrt{1-r^2} |\xi|)] - \E [h_3 (|\xi|)]}{h_3(0) - \E [h_3(|\xi|)]},$$
which is strictly increasing in $r^2$ (as $h_3$ is strictly decreasing), if $\sigma_1^2={\rm Var}(Y|X)$ is fixed. 

Note that $\rho^2(Y,Z|X)=0$ if and only if $r=0$. Also, $\rho^2(Y,Z|X)=1$ if and only if $r^2=1$.

For the linear kernel, $h_3(u) = -u^2/2$ (here $u \ge 0$), and therefore $\rho^2(Y,Z|X)=r^2$.
\qed

\subsection{Proof of~\cref{prop:monotone_lambda}}\label{pf:2}
\begin{lemma}\label{lem:sto_order}
    Suppose $Y\overset{d}{=}-Y$. Let $\xi$ be an independent noise which is symmetric and unimodal about 0. Then for fixed $c\geq 0$, $\p\left(|\lambda Y + \xi |\leq c \right)$ is a nonincreasing function of $\lambda$, for $\lambda \in [0,\infty)$.
\end{lemma}
\begin{proof}
    It suffices to show that for any $c \geq 0$, (a) $\p(|\xi|\leq c)\geq  \p\left(|Y + \xi|\leq c \right)$, and (b) for $\lambda>1$, $\p\left(|\lambda Y + \xi |\leq c \right)\leq \p\left(|Y + \xi|\leq c \right)$. Note that
    $$\begin{aligned}
        \p\left(|Y + \xi|\leq c \right) =\p(Y=0)\p(\xi \in [-c,c]) + 2\int_{(0,\infty)}\p(\xi \in[y-c,y+c]) dP_Y(y),
    \end{aligned}$$
    where we have used the fact $\p(\xi\in[-y-c,-y+c]) = \p (\xi \in [y-c,y+c]) $, and that $Y$ has a symmetric distribution.
Note that $h(y) :=\p (\xi \in [y-c,y+c])$ is a nonincreasing function on $[0,\infty)$. Hence,
    $$\begin{aligned}
        \p\left(|Y + \xi |\leq c \right) &= \p(Y=0)\p(\xi \in [-c,c]) + 2\E[h(Y)1_{Y>0}]\\
        &\geq \p(Y=0)\p(\xi \in [-c,c]) + 2\E[h(\lambda Y)1_{Y>0}]=\p\left(|\lambda Y + \xi|\leq c \right).
    \end{aligned}$$
    Also, $\p\left(|Y + \xi |\leq c \right) \leq \p(Y=0)\p(\xi \in [-c,c]) + 2\E[h(0)1_{Y>0}]=\p(|\xi|\leq c)$. These complete the proof of the lemma.
\end{proof}

Now let us prove~\cref{prop:monotone_lambda}. Recall the notation $\xi \sim \epsilon - \epsilon'$. Under the assumption of Proposition~\ref{prop:monotone_lambda},
$$\rho^2(Y,Z|X)=\frac{\E [h_3 (| \xi |)] - \E\left[\E [h_3\left(|\lambda [f(X,Z_1) - f(X,Z_2)] + \xi| \right)|X]\right] }{h_3(0) - \E\left[\E [h_3\left(|\lambda [f(X,Z_1) - f(X,Z_2)] + \xi| \right)|X]\right]},$$
By Lemma \ref{lem:sto_order} applied to $Y = f(X,Z_1) - f(X,Z_2)$, we know that conditional on $X$, $|\lambda_1 [f(X,Z_1) - f(X,Z_2)] + \xi|$ is stochastically less than $|\lambda_2 [f(X,Z_1) - f(X,Z_2)] + \xi|$ whenever $\lambda_1\leq \lambda_2$.
Since $h_3$ is a decreasing function, $\E\left[\E [h_3\left(|\lambda [f(X,Z_1) - f(X,Z_2)] + \xi| \right)|X]\right]$ is decreasing in $\lambda$, and hence 
$\rho^2(Y,Z|X)$ is increasing in $\lambda$.\qed

\subsection{Proof of Proposition \ref{prop:inv_conti}}\label{sec:pf_inv_conti}
From the form of $\rho^2(Y,Z|X)$ in~\cref{lem:eta-2} we see that $\rho^2(Y,Z|X)$ does not change by bijectively transforming $X$ and $Z$.

With the kernel given in \eqref{eq:h-1-2-3} and using the notation in~\cref{lem:eta-2}, $$\rho^2(Y,Z|X) = \frac{\E  h_3(\|Y_2-Y_2'\|)  - \E  h_3(\|Y_1-Y_1'\|)}{h_3(0) - \E h_3(\|Y_1-Y_1'\|)}.$$
Hence replacing $Y_i,Y_i'$, for $i=1,2$, by $OY_i + b,OY_i'+b$, where $O$ is an orthogonal matrix and $b$ is a vector, does not change $\rho^2(Y,Z|X)$. This shows the desired  invariance.

Next we show the continuity result. By Skorokhod's representation theorem, we can assume the convergence in distribution is actually a.s.~convergence.
The boundedness of $1+\varepsilon$ moments in our assumption implies uniform integrability for large $n$. Together with the a.s.~convergence, we have $L_1$ convergence, so $\rho^2(Y_n,Z_n|X_n)\to\rho^2(Y,Z|X)$ follows.
\qed

\subsection{Proof of~\cref{lem:AC19}}\label{sec:Cha-Stat}
We first show the general case, where $k(y_1,y_2) :=\int 1_{y_1\geq t} 1_{y_2\geq t} d P_{Y}(t)$, $y_1,y_2  \in \R$.
We use the equivalent formulation \eqref{eta_formula} to calculate $\rho^2$. For $Y_2, Y_2'$ as in~\eqref{eq:Y_2},
\begin{equation}\label{eq:connect_chatterjee}
    \begin{aligned}
        \E\left[\E[k(Y_2,Y_2')|X,Z] \right] &= \E\left[\E\left[\int 1_{Y_2\geq t}1_{Y_2'\geq t}d P_Y(t) \Big|X,Z\right] \right]\\
        &=\E\left[\int \p\left(Y\geq t|X,Z\right)^2d P_Y(t) \right].
    \end{aligned}
\end{equation}
Likewise, $\E\left[\E[k(Y_1,Y_1')|X] \right] = \E\left[\int \p\left(Y\geq t|X\right)^2d P_Y(t) \right]$. For $t \in \R$, let~$g_t(X,Z) :=\p\left(Y\geq t|X,Z\right)$. Then, we have,
$$\begin{aligned}
    \E\left[\E[k(Y_2,Y_2')|X,Z]\right] - \E\left[\E[k(Y_1,Y_1')|X]\right]&=\int \E\left[g_t(X,Z)^2 - \left(\E[g_t(X,Z)|X]\right)^2 \right]d P_Y(t)\\
    &=\int \E\left[ {\rm Var}(g_t(X,Z)|X) \right]d P_Y(t)\\
    &=\int \E\left[ {\rm Var}(\p(Y\geq t|X,Z)|X) \right]d P_Y(t),
\end{aligned}$$
which is exactly the numerator of $T(Y, Z| X)$. The denominator of $\rho^2$ in~\eqref{eta_formula} is
$$\begin{aligned}
    \E[k(Y,Y)]-\E[\E[k(Y_1,Y_1')|X]] &=\E\left[\int 1_{Y\geq t}^2 dP_Y(t) \right] - \E\left[\int \p\left(Y\geq t|X\right)^2d P_Y(t) \right]\\
    &=\int \E\left[1_{Y\geq t}^2 - \p\left(Y\geq t|X\right)^2 \right] d P_Y(t)\\
    &=\int \E\left[{\rm Var}(1_{Y\geq t}|X) \right]d P_Y(t),
\end{aligned}$$
which coincides with the denominator of $T(Y,Z|X)$.

Now, suppose $Y$ is continuous, and $k(y,y') :=\frac{1}{2}\left(|y|+|y'|-|y-y'|\right)=\min\{y,y'\}=\int_{0}^1 1_{y\geq t}1_{y'\geq t}dt$ for all $y,y'\in[0,1]$.
The right-inverse $F_Y^{-1}(t):=\inf \{y : F_Y(y)\geq t\}$ is non-decreasing and satisfies
$F_Y^{-1}(F_Y(Y))\overset{}{=}Y$ a.s. Let $Y'$ be an independent copy of $Y$. Then,
$$\begin{aligned}
    \E\left[\E[k(F_Y(Y_2),F_Y(Y_2'))|X,Z] \right] &= \E\left[\E\left[\int_0^1 1_{F_Y(Y_2)\geq t}1_{F_Y(Y_2')\geq t}dt \Big|X,Z\right] \right]\\
    &=\E\left[\int_0^1 \p\left(F_Y(Y)\geq t|X,Z\right)^2d t \right]\\
    &=\E \left[\p\left(F_Y(Y)\geq F_Y(Y')|X,Z\right)^2 \right]\\
    &=\E \left[\p\left(Y\geq Y'|X,Z\right)^2 \right]\\
    &=\E\left[\int \p\left(Y\geq t|X,Z\right)^2d P_Y(t) \right].
\end{aligned}$$
The rest of the proof is the same as that after \eqref{eq:connect_chatterjee}.
\qed

\subsection{Proof of \cref{prop:concen_moment}}\label{sec:pf_concen_moment}
By \cite[Proposition 3.1]{deb2020kernel}, both $\sqrt{n}\left(\frac{1}{n}\sum_i \sum_{(i,j)\in\mathcal{E}(\mathcal{G}_n^{\ddot{X}})} \frac{k (Y_i,Y_j)}{d_i^{\ddot{X}}} - \E k(Y_1,Y_{\ddot{N}(1)})  \right) $ and $ \sqrt{n}\left(\frac{1}{n}\sum_i \sum_{(i,j)\in\mathcal{E}(\mathcal{G}_n^{{X}})} \frac{k (Y_i,Y_j)}{d_i^{{X}}} - \E k(Y_1,Y_{{N}(1)})  \right)$ are $O_p(1)$.
Since we also have $\sqrt{n}\left(\frac{1}{n}\sum_{i=1}^n k(Y_i,Y_i) - \E k(Y_1,Y_1) \right) = O_p(1)$, the following holds:
$$\sqrt{n}\left(\hat{\rho}^2 (Y,Z|X)-\frac{\E k(Y_1,Y_{\ddot{N}(1)}) - \E k(Y_1,Y_{{N}(1)}) }{\E k(Y_1,Y_1) - \E k(Y_1,Y_{{N}(1)})} \right) = O_p(1).$$
\qed

\subsection{Proof of Theorem \ref{thm:conv_rate}}\label{sec:pf_conv_rate}
Let 
\begin{eqnarray*}
	Q_n & := & \frac{1}{n}\sum_{i=1}^n \sum_{(i,j)\in\mathcal{E}(\mathcal{G}_n^{\ddot{X}})} \frac{k (Y_i,Y_j)}{d_i^{\ddot{X}}} - \frac{1}{n}\sum_{i=1}^n \sum_{(i,j)\in\mathcal{E}(\mathcal{G}_n^{{X}})}\frac{k (Y_i,Y_j)}{d_i^{X}}, \\ S_n & := & \frac{1}{n}\sum_{i=1}^n k (Y_i,Y_i)- \frac{1}{n}\sum_{i=1}^n \sum_{(i,j)\in\mathcal{E}(\mathcal{G}_n^{{X}})}\frac{k (Y_i,Y_j)}{d_i^{X}},
\end{eqnarray*} 
and their population limits be 
$$Q:=\mathbb{E}\left[\mathbb{E}[k(Y_2,Y_2')|X,Z]\right] - \mathbb{E}\left[\mathbb{E}[k(Y_1,Y_1')|X]\right] ,
 \quad S:=\mathbb{E}[k(Y,Y)]-\mathbb{E}[\mathbb{E}[k(Y_1,Y_1')|X]].$$
Then $\hat{\rho^2}(Y,Z|X)=\frac{Q_n}{S_n}$, and $\rho^2(Y,Z|X)=\frac{Q}{S}$.
From the proof of \cite[Theorem 5.1]{deb2020kernel} and \cite[Corollary 5.1]{deb2020kernel}\footnote{Note that with the notation in \cite{deb2020kernel}, by Assumption \ref{assump:intrin_dim} on the intrinsic dimensionality of $X$, we have $N(\mu_X,B(x^*,t),\varepsilon,0)\leq C_1(t/\varepsilon)^d$, and {$t_n/K_n$} being bounded is implied by Assumption \ref{assump:knndegbound}
(in \cite{deb2020kernel}, $t_n$ is defined as the upper bound of vertex degree (including both in- and out-degrees for directed graphs) in the $K$-NN graph). Hence the same argument in \cite[Section 5.1]{deb2020kernel} works through.}, $\frac{1}{n}\sum_i \sum_{(i,j)\in\mathcal{E}(\mathcal{G}_n^{\ddot{X}})} \frac{k (Y_i,Y_j)}{d_i^{\ddot{X}}}$ and $\frac{1}{n}\sum_i \sum_{(i,j)\in\mathcal{E}(\mathcal{G}_n^{{X}})}\frac{k (Y_i,Y_j)}{d_i^{X}}$
are within $O_p\left(\sqrt{\nu_n}\right)$ distance of their theoretical limits.
Since $\frac{1}{n}\sum_{i=1}^n k(Y_i,Y_i)$ is also within $O_p(n^{-1/2})\lesssim O_p\left(\sqrt{\nu_n}\right)$ distance from its theoretical limit,
$Q_n$ (resp. $S_n$) is within $O_p\left(\sqrt{\nu_n}\right)$ distance to $Q$ (resp. $S$).
Hence
$$\left|\hat{\rho^2}(Y,Z|X) - \rho^2(Y,Z|X)\right| = \left|\frac{S\cdot Q_n - Q\cdot S_n}{S \cdot S_n}\right| = \left|\frac{S \cdot O_p\left(\sqrt{\nu_n} \right) - Q \cdot O_p\left(\sqrt{\nu_n} \right)}{S \cdot S_n} \right| = O_p\left(\sqrt{\nu_n} \right).$$
The last equality follows as $S> 0$.
\qed

\subsection{Proof of~\cref{lem:eta}}\label{sec:eta}
We use the same notation as in~\cite{sheng2019distance}. Define $S_X: \mathcal{H}_\mathcal{X}\to\mathbb{R}^n$ such that $S_X: f\mapsto (f(X_1),\cdots,f(X_n))^\top$.
Then, for $\alpha = (\alpha_1, \ldots, \alpha_n) \in \R^n$, $\langle S_X f,\alpha \rangle_{\mathbb{R}^n}=\sum_i \alpha_i f(X_i)=\langle f,S_X^* \alpha\rangle_{\mathcal{H}_\mathcal{X}}$ where
we have $S_X^*: \mathbb{R}^n\to \mathcal{H}_\mathcal{X}$ given by $S_X^*: \alpha\mapsto \sum_i \alpha_ik_{\mathcal{X}}(X_i,\cdot)$.
Similarly define $S_Y, S_Y^*,S_Z$ and $S_Z^*$. For $f\in\mathcal{H}_\mathcal{X}$,
$$\begin{aligned}
    \hat{C}_{YX}f &= \frac{1}{n}\sum_{i=1}^n k (Y_i,\cdot)f(X_i) - \left(\frac{1}{n}\sum_{i=1}^n k (Y_i,\cdot) \right) \left(\frac{1}{n}\sum_{i=1}^n f(X_i) \right) \\
    &=\frac{1}{n} S_Y^* S_X f - \frac{1}{n^2}S_Y^* \mathbf{1}\mathbf{1}^\top S_X f \\
    &=\frac{1}{n} S_Y^*(I-\frac{1}{n}\mathbf{1}\mathbf{1}^\top)S_Xf \;=\; \frac{1}{n}S_Y^* H S_X f.
\end{aligned}$$
Hence $\hat{C}_{YX}=\frac{1}{n}S_Y^* H S_X$.
Similarly, we have $$\hat{C}_{Y\ddot{X}}=\frac{1}{n}S_Y^* H S_{\ddot{X}}, \qquad \hat{C}_X=\frac{1}{n}S_X^* H S_X, \qquad \hat{C}_{\ddot{X}}=\frac{1}{n}S_{\ddot{X}}^* H S_{\ddot{X}}.$$
Note that for all $\alpha \in \R^n$, $S_X S_X^* \alpha = \sum_{i=1}^n \alpha_i S_X k_{\mathcal{X}}(X_i,\cdot)=K_X\alpha$, where $K_X$ is the kernel matrix with $(K_X)_{ij}=k_\mathcal{X}(X_i,X_j)$.
Hence $S_XS_X^* = K_X$. Letting $e_i$ denote the $i$-th unit vector in $\R^n$, for $i=1,\ldots, n$, we have 
$$\begin{aligned}
    \hat{C}_{YX}(\hat{C}_{{X}} + \varepsilon I)^{-1} (k_{{\mathcal{X}}}(X_i,\cdot)-\hat{\mu}_X) &= \hat{C}_{YX}(\hat{C}_{{X}}+\varepsilon I)^{-1} S_X^* (e_i-\frac{1}{n}\mathbf{1})\\
    &=\frac{1}{n}S_Y^* H S_X \left(\frac{1}{n}S_X^* H S_X + \varepsilon I\right)^{-1} S_X^* He_i\\
    &=\frac{1}{n}S_Y^* H S_X S_X^* \left(\frac{1}{n} H S_XS_X^* + \varepsilon I\right)^{-1}  He_i
\end{aligned}$$    
where we have used the fact that, for operators $A$ and $B$, $(BA+\varepsilon I)^{-1} B = B(AB+\varepsilon I)^{-1}$, which holds by direct verification. Now, using $K_X = S_X S_X^*$, we can show that the right side of the above display equals
$$\begin{aligned}
S_Y^* H K_X \left( H K_X + n\varepsilon I\right)^{-1} H e_i
    &=S_Y^* H K_X \left( H^2 K_X + n\varepsilon I\right)^{-1} H e_i\\
    &=S_Y^* H K_X H\left( H K_X H + n\varepsilon I\right)^{-1} e_i\\
    &=S_Y^* \tilde{K}_X \left( \tilde{K}_X + n\varepsilon I\right)^{-1} e_i,
\end{aligned}$$
where $\tilde{K}_X := HK_X H$ is the centered kernel matrix. Similarly, we have 
$$\hat{C}_{Y\ddot{X}}(\hat{C}_{\ddot{X}} + \varepsilon I)^{-1} \left( k_{\ddot{\mathcal{X}}}(\ddot{X}_i,\cdot) - \hat{\mu}_{\ddot{X}}\right) = S_Y^* \tilde{K}_{\ddot{X}} \left( \tilde{K}_{\ddot{X}} + n\varepsilon I\right)^{-1} e_i.$$
Recalling that $M= \tilde{K}_X \big( \tilde{K}_X + n\varepsilon I\big)^{-1} - \tilde{K}_{\ddot{X}} \big( \tilde{K}_{\ddot{X}} + n\varepsilon I\big)^{-1}$, the numerator of $\tilde{\rho^2}$ reduces to $\sum_{i=1}^n \|S_Y^* M e_i\|^2_{\mathcal{H}_\mathcal{Y}}$. Note that, for $\alpha \in \R^n$, $\|S_Y^* \alpha\|^2_{\mathcal{H}_\mathcal{Y}} = \langle \sum_i \alpha_i k (Y_i,\cdot),\sum_j \alpha_j k (Y_j,\cdot)\rangle_{\mathcal{H}_\mathcal{Y}}=\sum_{i,j=1}^n \alpha_i\alpha_j k(Y_i,Y_j)=\alpha^\top K_Y\alpha$.
Thus, 
$$\sum_{i=1}^n \|S_Y^* M e_i\|^2_{\mathcal{H}_\mathcal{Y}} = \sum_{i=1}^n e_i^\top M^\top K_Y Me_i = {\rm Tr}(M^\top K_Y M) = {\rm Tr}(M^\top \tilde{K}_Y M),$$ where we have used the fact that $H^2 = H$.
Now the denominator of $\tilde{\rho^2}$ can be simplified as
$$ \sum_{i=1}^n \|S_Y^* H e_i - S_Y^* \tilde{K}_X \left( \tilde{K}_X + n\varepsilon I\right)^{-1} e_i\|^2_{\mathcal{H}_\mathcal{Y}}$$
where we have used the fact that $k (Y_i,\cdot)- \hat{\mu}_Y = S_Y^* H e_i$.
Letting $N_0 :=H-\tilde{K}_X \left( \tilde{K}_X + n\varepsilon I\right)^{-1}$, and recalling that $N=I-\tilde{K}_X \left( \tilde{K}_X + n\varepsilon I\right)^{-1}$, the above display can be expressed as
$$
 \sum_{i=1}^n \|S_Y^* N_0 e_i\|^2_{\mathcal{H}_\mathcal{Y}} =\sum_{i=1}^n e_i^\top N_0^\top K_Y N_0 e_i = {\rm Tr}(N_0^\top K_Y N_0) = {\rm Tr}(N^\top \tilde{K}_Y N).
$$
This proves the desired result. \qed

\subsection{Proof of Proposition~\ref{prop:kernel_ridge}}\label{kernel_ridge}
    Recall that uncentered estimator $$\tilde{\rho^2_u} =\frac{{\rm Tr}(M^\top K_Y M)}{{\rm Tr}(N^\top K_Y N)},$$
    where $M= {K}_X \left( {K}_X + n\varepsilon I\right)^{-1} - K_{\ddot{X}} \left({K}_{\ddot{X}} + n\varepsilon I\right)^{-1}$ and $N=I-{K}_X \left({K}_X + n\varepsilon I\right)^{-1}$.
    By assumption we have $K_Y = \mathbf{Y} \mathbf{Y}^\top$. As $M,N$ are symmetric matrices,
    \begin{equation}\label{eq:ker_ridge}
        \tilde{\rho^2_u} =\frac{{\rm Tr}(M^\top \mathbf{Y} \mathbf{Y}^\top M)}{{\rm Tr}(N^\top \mathbf{Y} \mathbf{Y}^\top N)} = \frac{\|M\mathbf{Y}\|^2 }{\|N \mathbf{Y}\|^2 } = \frac{\big\|K_X (K_X + n\varepsilon I)^{-1} \mathbf{Y} -  K_{\ddot{X}}(K_{\ddot{X}} + n\varepsilon I)^{-1} \mathbf{Y}\big\|^2 }{\|\mathbf{Y} - K_X(K_X + n\varepsilon I)^{-1} \mathbf{Y}\|^2 }.
    \end{equation}
Kernel ridge regression yields $\hat{\mathbf{Y}} = K(K+\lambda I)^{-1}\mathbf{Y}$ for a generic kernel matrix and regularization parameter $\lambda$ (see e.g.,~\cite[Section 7.3.4]{kung_2014}). Hence the above display reduces to $$\tilde{\rho^2_u} = \frac{\|\hat{\mathbf{Y}}_x - \hat{\mathbf{Y}}_{xz}\|^2}{\|\hat{\mathbf{Y}}_x - \mathbf{Y}\|^2},$$ with the regularization parameter $\rho = n\varepsilon$.

\subsection{Proof of Proposition~\ref{prop:reduce_classi}}\label{sec:pf_reduce_classi}
The same argument as \eqref{eq:ker_ridge} shows that the centered estimator $\tilde{\rho^2}$ can be expressed as 
\begin{eqnarray*}
\frac{\big\|\tilde{K}_X (\tilde{K}_X + n\varepsilon I)^{-1} H\mathbf{Y} -  \tilde{K}_{\ddot{X}}(\tilde{K}_{\ddot{X}} + n\varepsilon I)^{-1} H\mathbf{Y}\big\|^2}{\|H\mathbf{Y} - \tilde{K}_X(\tilde{K}_X + n\varepsilon I)^{-1} H\mathbf{Y}\|^2} = \frac{\Big\|\big[(\tilde{K}_X + n\varepsilon)^{-1} - (\tilde{K}_{\ddot{X}}+n\varepsilon I)^{-1} \big]H\mathbf{Y}\Big\|^2}{\| (\tilde{K}_X+n\varepsilon I)^{-1}H\mathbf{Y}\|^2}.
\end{eqnarray*}
With linear kernels, $K_X = \mathbf{X}\mathbf{X}^\top$ and $K_{\ddot{X}} = \ddot{\mathbf{X}}\ddot{\mathbf{X}}^\top$.
    Let $\mathbf{Y}_c=H\mathbf{Y}$, $\mathbf{X}_c=H\mathbf{X}$, $\ddot{\mathbf{X}}_c=H\ddot{\mathbf{X}}$ be the centered versions
    of $\mathbf{Y}$, $\mathbf{X}$, $\ddot{\mathbf{X}}$ by subtracting the mean from all columns.
    The empirical classical partial correlation first fits two linear regressions of $\mathbf{Y}$ on $\mathbf{X}$ and $\mathbf{Z}$ on $\mathbf{X}$ (an intercept term is added so the design matrix is $[\mathbf{1}_n\ \mathbf{X}]$), and then outputs the correlation of the two resulting residuals $r_Y,r_Z$.
    Therefore, the partial correlation remains unchanged when we replace $\mathbf{Y},\mathbf{Z}$ by $\mathbf{Y}_c,\mathbf{Z}_c$.
    Without loss of generality, suppose $\mathbf{Y}=\mathbf{Y}_c$, $\mathbf{Z}=\mathbf{Z}_c$ have been centered.
    By the matrix identity $$(\mathbf{X}_c\mathbf{X}_c^\top +n\varepsilon I)^{-1} = \frac{1}{n\varepsilon} I - \frac{1}{n\varepsilon} \mathbf{X}_c(n\varepsilon I+ \mathbf{X}_c^\top \mathbf{X}_c)^{-1}\mathbf{X}_c^\top,$$
    we have 
\begin{eqnarray*}
\hat{\rho^2}(Y,Z|X) & = & \frac{\left\|\left[(\mathbf{X}_c\mathbf{X}_c^\top + n\varepsilon)^{-1} - (\ddot{\mathbf{X}}_c\ddot{\mathbf{X}}_c^\top +n\varepsilon I)^{-1} \right]\mathbf{Y}_c\right\|^2}{\| (\mathbf{X}_c\mathbf{X}_c^\top+n\varepsilon I)^{-1}\mathbf{Y}_c\|^2} \\ & =& \frac{\left\|\left[\mathbf{X}_c (n\varepsilon I+ \mathbf{X}_c^\top \mathbf{X}_c)^{-1}\mathbf{X}_c^\top - \ddot{\mathbf{X}}_c (n\varepsilon I + \ddot{\mathbf{X}}_c^\top \ddot{\mathbf{X}}_c)^{-1}\ddot{\mathbf{X}}_c^\top  \right]\mathbf{Y}_c\right\|^2}{\| (I - \mathbf{X}_c (n\varepsilon I + \mathbf{X}_c^\top \mathbf{X}_c)^{-1}\mathbf{X}_c^\top)\mathbf{Y}_c \|^2} \\& \overset{\varepsilon\to 0}{\longrightarrow} & \frac{\|{\rm Proj}_{\mathbf{X}_c}\mathbf{Y}_c - {\rm Proj}_{\ddot{\mathbf{X}}_c}\mathbf{Y}_c \|^2}{\|\mathbf{Y}_c - {\rm Proj}_{\mathbf{X}_c}\mathbf{Y}_c\|^2},
\end{eqnarray*}
where ${\rm Proj}_{\mathbf{X}_c}\mathbf{Y}_c$ is the projection of $\mathbf{Y}_c$ onto the column space of $\mathbf{X}_c$.
    Note that $r_{Z}$ is in the column space of $\ddot{\mathbf{X}}_c$ and is orthogonal to the column space of $\mathbf{X}_c$, and $\mathbf{Y}_c - {\rm Proj}_{\mathbf{X}_c}\mathbf{Y}_c = r_{Y}$.
    Hence ${\rm Cor}(r_Y,r_Z)={\rm Cor}(\mathbf{Y}_c - {\rm Proj}_{\mathbf{X}_c}\mathbf{Y}_c,r_Z) = {\rm Cor}(\mathbf{Y}_c - {\rm Proj}_{\mathbf{X}_c}\mathbf{Y}_c,{\rm Proj}_{\ddot{\mathbf{X}}_c}\mathbf{Y}_c -{\rm Proj}_{\mathbf{X}_c}\mathbf{Y}_c)$.
    Note that $\left(\mathbf{Y}_c - {\rm Proj}_{\mathbf{X}_c}\mathbf{Y}_c\right) - \left({\rm Proj}_{\ddot{\mathbf{X}}_c}\mathbf{Y}_c -{\rm Proj}_{\mathbf{X}_c}\mathbf{Y}_c \right)$
    is orthogonal to ${\rm Proj}_{\ddot{\mathbf{X}}_c}\mathbf{Y}_c -{\rm Proj}_{\mathbf{X}_c}\mathbf{Y}_c$.
    Hence ${\rm Cor}(r_Y,r_Z)^2 = {\rm Cor}(\mathbf{Y}_c - {\rm Proj}_{\mathbf{X}_c}\mathbf{Y}_c,{\rm Proj}_{\ddot{\mathbf{X}}_c}\mathbf{Y}_c -{\rm Proj}_{\mathbf{X}_c}\mathbf{Y}_c)^2 = \frac{\|{\rm Proj}_{\mathbf{X}_c}\mathbf{Y}_c - {\rm Proj}_{\ddot{\mathbf{X}}_c}\mathbf{Y}_c \|^2}{\|\mathbf{Y}_c - {\rm Proj}_{\mathbf{X}_c}\mathbf{Y}_c\|^2}$.
\qed

\subsection{Proofs of Theorems~\ref{thm:CME_result} and~\ref{thm:kernel_consistency}}\label{pf:kernel_consistency}
We start with some preliminaries. The following well-known result (see e.g.,~\cite[Lemma 5]{fukumizu2007statistical}) shows that the cross-covariance operator $C_{YX}$ and its sample version are close in the Hilbert-Schmidt norm and are consequently close in the operator norm, since $\|\cdot\|_{\rm op}\leq \|\cdot \|_{\HS}$. 
\begin{lemma}\label{consistency_emp_cross_cov}
Suppose that $\E[k_\X(X,X)] <\infty$ and $\E[k (Y,Y)] <\infty$, and $\h_\X,\h_\Y$ are both separable. Then $\|\hat{C}_{YX}-C_{YX}\|_{\HS} = O_p(n^{-1/2}).$
\end{lemma}

Recall the expressions for $\rho^2$ and $\tilde{\rho^2}$ in~\eqref{eq:eta-Cen} and~\eqref{eq:eta-Cen-Est} respectively. 
    We divide the proof into several steps. Recall that the operator norm of a linear map $A: \V \to \W$ (for two given normed vector spaces $\V$ and $\W$) is  defined as
$$\|A\|_{\op} :=\inf\{c\geq 0:\|Av\|\leq c\|v\|{\mbox{ for all }}v\in \V\}.$$ For notational simplicity, $\|\cdot\|$ will, by default, denote either the norm in an RKHS or the operator norm. 

\noindent\textbf{Step 1.} 
    $\quad \E\left[\|[(C_X^\dagger C_{XY})^* - C_{YX}(C_X+\varepsilon I)^{-1}](k_\X(X,\cdot) - \mu_X)\|^2_{\h_\Y}\right]\to 0$ as $\varepsilon\to 0^+$.

    Denote by $A :=C_X^\dagger C_{XY}$, which is a bounded linear operator by Lemma \ref{center_CME}. Further,
    $\ran C_{XY} \subset \ran C_X$ and $C_X C_X^\dagger C_X=C_X$ imply $C_X A = C_{XY}$. Hence 
\begin{eqnarray*}
       & & \E\left[\big\|[(C_X^\dagger C_{XY})^* - C_{YX}(C_X+\varepsilon I)^{-1}](k_\X(X,\cdot) - \mu_X)\big\|^2\right] \\
        &= &\E\left[\|[A^* - A^* C_X(C_X+\varepsilon I)^{-1}](k_\X(X,\cdot) - \mu_X)\|^2\right]\\
        &\leq & \|A^*\|\cdot \E\left[\|[I-C_X(C_X+\varepsilon I)^{-1}](k_\X(X,\cdot) - \mu_X)\|^2 \right]
    \end{eqnarray*}
where in the last display we have used the definition of the operator norm. Let $\{e_j\}_{j\ge 1}$ be an eigenbasis of $C_X$ with corresponding eigenvalues $\{\lambda_j\}_{j\ge 1}$. Since $C_X$ is trace-class (see Lemma \ref{trace_class}), $\sum_j \lambda_j <\infty$. Then,  $I-C_X(C_X+\varepsilon I)^{-1}$ has eigenbasis $\{e_j\}_{j\ge 1}$ with corresponding eigenvalues $\{\varepsilon( \varepsilon + \lambda_j)^{-1}\}_{j\ge 1}$. Thus,
    $$\begin{aligned}
        & \E\left[\big\|[I-C_X(C_X+\varepsilon I)^{-1}](k_\X(X,\cdot) - \mu_X)\big\|^2 \right] \\
        &=\E\left[\Big\|\sum_j \frac{\varepsilon}{\varepsilon+\lambda_j}\langle k_\X(X,\cdot) - \mu_X, e_j\rangle e_j \Big\|^2 \right] \;\, =\;\, \E\left[\sum_j \frac{\varepsilon^2 \langle k_\X(X,\cdot) - \mu_X, e_j\rangle^2}{(\varepsilon + \lambda_j)^2} \right]\\
        & = \sum_j \frac{\varepsilon^2 {\rm Var}(e_j(X))}{(\varepsilon + \lambda_j)^2} \;\; =\;\; \sum_j \frac{\varepsilon^2 \lambda_j}{(\varepsilon + \lambda_j)^2}.
    \end{aligned}$$
It is easily seen that the above quantity converges to $0$ as $\varepsilon\to 0^+$.\newline

    \noindent\textbf{Step 2.} $\frac{1}{n}\sum_{i=1}^n \big\|[(C_X^\dagger C_{XY})^* - C_{YX}(C_X+\varepsilon I)^{-1}](k_\X(X_i,\cdot) - \mu_X) \big\|^2 \overset{p}{\to} 0$ as $\varepsilon\to 0^+$.

    This is a direct consequence of Step 1 and Markov's inequality. \newline

    \noindent\textbf{Step 3.} $\frac{1}{n}\sum_{i=1}^n \big\|[\hat{C}_{YX}(\hat{C}_X+\varepsilon_n I)^{-1}-C_{YX}(C_X+\varepsilon_n I)^{-1}](k_\X(X_i,\cdot) - \hat{\mu}_X) \big\|^2\overset{p}{\to} 0$.

    This is the only step where we use $\varepsilon_n n^{1/2}\to \infty$. We will use the following simple result (proved in Section~\ref{pf:same_limit}). 
\begin{lemma}
    \label{same_limit}
Suppose that $\{U_i^{(n)}: 1\leq i\leq n\}_{n\geq 1}$ and $\{V_i^{(n)}: 1\leq i\leq n\}_{n\geq 1}$ are random elements taking values in a Hilbert space with norm $\|\cdot\|$. If $\frac{1}{n}\sum_{i=1}^n \|U_i^{(n)}-V_i^{(n)}\|^2\overset{p}{\to}0$ and $\frac{1}{n}\sum_{i=1}^n \|U_i^{(n)}\|^2 \overset{p}{\to} U$, then $\frac{1}{n}\sum_{i=1}^n \|V_i^{(n)}\|^2 \overset{p}{\to} U$.
\end{lemma}
For $i=1,\ldots, n$, let
\begin{eqnarray*}
    V_i^{(n)} & := & \big[\hat{C}_{YX}(\hat{C}_X+\varepsilon_n I)^{-1}-C_{YX}(C_X+\varepsilon_n I)^{-1}\big](k_\X(X_i,\cdot) - \hat{\mu}_X) \\
    & = & \Big[(\hat{C}_{YX} - C_{YX})(\hat{C}_X+\varepsilon_n I)^{-1} \\ 
    &&\qquad \quad + \; C_{YX}\big((\hat{C}_X+\varepsilon_n I)^{-1} - ({C}_X+\varepsilon I)^{-1} \big)\Big](k_\X(X_i,\cdot) - \hat{\mu}_X).
\end{eqnarray*}
Letting $U_i^{(n)} := (\hat{C}_{YX} - C_{YX})(\hat{C}_X+\varepsilon_n I)^{-1}(k_\X(X_i,\cdot) - \hat{\mu}_X)$, for $i=1,\ldots, n$, note that
    $$\begin{aligned}
        \frac{1}{n} \sum_{i=1}^n \|U_i^{(n)}\|^2 & =  \frac{1}{n}\sum_{i=1}^n \big\|(\hat{C}_{YX} - C_{YX})(\hat{C}_X+\varepsilon_n I)^{-1}(k_\X(X_i,\cdot) - \hat{\mu}_X)\big\|^2\\
        &\leq \frac{1}{n}\sum_{i=1}^n\|\hat{C}_{YX} - C_{YX}\|^2\cdot {\|(\hat{C}_X+\varepsilon_n I)^{-1}\|^2}\cdot \|k_\X(X_i,\cdot) - \hat{\mu}_X\|^2\\
        &\leq O_p\left(\frac{1}{n}\right)\cdot {\frac{1}{\varepsilon_n^2}}\cdot \frac{1}{n} \sum_{i=1}^n \|k_\X(X_i,\cdot) - \hat{\mu}_X\|^2 = O_p\left(\frac{1}{n\varepsilon_n^2}\right) = o_p(1),
    \end{aligned}$$
where we have used Lemma~\ref{consistency_emp_cross_cov}, $\|(\hat{C}_X + \varepsilon_n)^{-1}\|\leq \frac{1}{\varepsilon_n}$ (since if $\hat{C}_X$ has eigenvalues $\hat{\lambda}_i$, then $(\hat{C}_X + \varepsilon_n)^{-1}$ has eigenvalues $\frac{1}{\hat{\lambda}_i+\varepsilon_n}\leq \frac{1}{\varepsilon_n}$), and that $n \varepsilon_n^2 \to \infty$.

Thus, in view of Lemma~\ref{same_limit}, we only need to show that $$\frac{1}{n}\sum_{i=1}^n  \|V_i^{(n)}-U_i^{(n)}\|^2 = \frac{1}{n}\sum_{i=1}^n \Big\|C_{YX}\big((\hat{C}_X+\varepsilon_n I)^{-1} - ({C}_X+\varepsilon_n I)^{-1} \big)(k_\X(X_i,\cdot) - \hat{\mu}_X) \Big\|^2 \overset{p}{\to} 0.$$
    Using the equality $B^{-1}-C^{-1} = C^{-1}(C-B)B^{-1}$, and $A= C_X^\dagger C_{XY}$, $C_{YX}= A^* C_X$ as before,
    $$\begin{aligned}
        &\qquad \Big\|C_{YX}\big((\hat{C}_X+\varepsilon_n I)^{-1} - ({C}_X+\varepsilon_n I)^{-1} \big) \Big\|\\
        &= \big\|C_{YX} ({C}_X+\varepsilon_n I)^{-1} (C_X - \hat{C}_X) (\hat{C}_X+\varepsilon_n I)^{-1}\big\|\\
        &\leq \big\|A^* C_X ({C}_X+\varepsilon_n I)^{-1}\|\cdot \|C_X - \hat{C}_X\|\cdot \|(\hat{C}_X+\varepsilon_n I)^{-1}\big\|\\
        &\leq \|A^*\|\cdot \|C_X ({C}_X+\varepsilon_n I)^{-1}\| \cdot O_p(n^{-1/2}) \cdot \frac{1}{\varepsilon_n} = O_p\left(\frac{1}{\varepsilon_n \sqrt{n}} \right)=o_p(1).
    \end{aligned}$$
    This implies that $$\frac{1}{n}\sum_{i=1}^n  \|V_i^{(n)}-U_i^{(n)}\|^2 \le \Big\|C_{YX}\big((\hat{C}_X+\varepsilon_n I)^{-1} - ({C}_X+\varepsilon_n I)^{-1}\big)\Big \| \cdot  \frac{1}{n}\sum_{i=1}^n \|k_\X(X_i,\cdot) - \hat{\mu}_X \|^2 \overset{p}{\to} 0.$$
Now this step follows from Lemma~\ref{same_limit}. \newline

    \noindent\textbf{Step 4.} $\frac{1}{n}\sum_{i=1}^n \big\|C_{YX}(C_X+\varepsilon_n I)^{-1}(k_\X(X_i,\cdot) - \hat{\mu}_X) - C_{YX}(C_X+\varepsilon_n I)^{-1}(k_\X(X_i,\cdot) - {\mu}_X) \big\|^2\overset{p}{\to} 0$.\\

    This follows as
    $$\begin{aligned}
        &\quad \frac{1}{n}\sum_{i=1}^n \|C_{YX}(C_X+\varepsilon_n I)^{-1}(k_\X(X_i,\cdot) - \hat{\mu}_X) - C_{YX}(C_X+\varepsilon_n I)^{-1}(k_\X(X_i,\cdot) - {\mu}_X) \|^2\\
        &=\frac{1}{n}\sum_{i=1}^n \|C_{YX}(C_X+\varepsilon_n I)^{-1}(\hat{\mu}_X - {\mu}_X) \|^2 \;\,\leq \;\, \|A^*\|\cdot \|C_X(C_X+\varepsilon_n I)^{-1}\|\cdot \frac{1}{n}\sum_{i=1}^n \|\hat{\mu}_X - \mu_X\|^2\\
        &\leq \|A^*\|\cdot 1 \cdot \|\hat{\mu}_X-\mu_X\|^2\overset{a.s.}{\to}0,
    \end{aligned}$$
    by the strong law of large numbers~\cite{hoffmann1976law}.\\

    Now let us show Theorem~\ref{thm:CME_result}. By successive applications of Lemma \ref{same_limit},
    $$\begin{aligned}
        &\qquad \lim_{n\to \infty} \frac{1}{n} \sum_{i=1}^n \|\mu_{Y|X_i} - \hat \mu_{Y|X_i}\|_{\h_\Y}^2 \\
        &=\lim_{n\to \infty} \frac{1}{n} \sum_{i=1}^n \|\mu_Y + (C_X^\dagger C_{XY})^*\left(k(X_i,\cdot)-\mu_X \right) - \hat{\mu}_Y - \hat{C}_{YX}(\hat{C}_X + \varepsilon_n I)^{-1}\left( k(X_i,\cdot)-\hat{\mu}_X\right)\|_{\h_\Y}^2\\
        &=\lim_{n\to \infty} \frac{1}{n} \sum_{i=1}^n \|(C_X^\dagger C_{XY})^*\left(k(X_i,\cdot)-\mu_X \right) - \hat{C}_{YX}(\hat{C}_X + \varepsilon_n I)^{-1}\left( k(X_i,\cdot)-\hat{\mu}_X\right)\|_{\h_\Y}^2\\
        &\overset{\rm Step\ 3}{=}\lim_{n\to \infty} \frac{1}{n} \sum_{i=1}^n \|(C_X^\dagger C_{XY})^*\left(k(X_i,\cdot)-\mu_X \right) - {C}_{YX}({C}_X + \varepsilon_n I)^{-1}\left( k(X_i,\cdot)-\hat{\mu}_X\right)\|_{\h_\Y}^2\\
        &\overset{\rm Step\ 4}{=}\lim_{n\to \infty} \frac{1}{n} \sum_{i=1}^n \|(C_X^\dagger C_{XY})^*\left(k(X_i,\cdot)-\mu_X \right) - {C}_{YX}({C}_X + \varepsilon_n I)^{-1}\left( k(X_i,\cdot)-{\mu}_X\right)\|_{\h_\Y}^2\\
        &\overset{\rm Step\ 2}{=}0,
    \end{aligned}$$
    where the limit is in probability, and in the second equality we have used $\frac{1}{n}\sum_{i=1}^n \|\mu_Y -\hat{\mu}_Y\|^2_{\h_\Y} = \|\mu_Y -\hat{\mu}_Y\|^2_{\h_\Y} \overset{p}{\to}0$.\\

    Then we can use Theorem~\ref{thm:CME_result} and Lemma \ref{same_limit} to show Theorem \ref{thm:kernel_consistency}. By the expression of $\tilde{\rho^2}$ in \eqref{eq:eta-Cen-Est}, we see that the numerator
    $$\begin{aligned}
        \lim_{n\to\infty }\frac{1}{n}\sum_{i=1}^n \|\hat{\mu}_{Y|\ddot{X}_i}  - \hat{\mu}_{Y|{X_i}}\|^2_{\mathcal{H}_\Y}=\lim_{n\to\infty } \frac{1}{n}\sum_{i=1}^n \|{\mu}_{Y|\ddot{X}_i}  - {\mu}_{Y|{X_i}}\|^2_{\mathcal{H}_\Y} = \E[\|\mu_{Y|\ddot{X}}-\mu_{Y|X}\|^2_{\h_\Y}],
    \end{aligned}$$
    and the denominator
    $$\lim_{n\to\infty} \frac{1}{n} \sum_{i=1}^n \|k (Y_i,\cdot)- \hat{\mu}_{Y|{X_i}}\|^2_{\mathcal{H}_\Y} = \lim_{n\to\infty} \frac{1}{n} \sum_{i=1}^n \|k (Y_i,\cdot)- {\mu}_{Y|{X_i}}\|^2_{\mathcal{H}_\Y} = \E[\|k(Y,\cdot)-\mu_{Y|X}\|_{\h_\Y}^2].$$
    Both limits are in probability, and they are exactly the numerator and denominator of $\rho^2(Y,Z|X)$ respectively.
\qed

\subsection{Proof of Theorem~\ref{thm:var_select}}\label{pf:var_select}
The proof is similar to that of \cite[Theorem 6.1]{azadkia2019simple}. Let $j_1,\ldots,j_p$ be the complete ordering of all variables produced by the algorithm without imposing the stopping rule. Let 
$S_k := \{j_1,\ldots,j_k\}$ for $1\leq k\leq p$, $S_0:=\emptyset$, $S_k:= S_p$ for $k>p$, and recall that $\kappa = \lfloor \frac{M}{\delta} + 1 \rfloor$.
Let $\varepsilon_1,\varepsilon_2>0$ be small such that $((1-\varepsilon_2)\delta - 2\varepsilon_1)\lfloor \frac{M}{\delta} + 1 \rfloor > M$ and $(1-\varepsilon_2)\delta + 2\varepsilon_1 <\delta$.
Note that $\varepsilon_1,\varepsilon_2$ only depend on $\delta$ and $M$.
Define:
\begin{enumerate}
    \item Event $E_0$: $S_{\kappa}$ is sufficient.
    \item Event $E$: $|T_n(S_k) - T(S_k)|\leq \varepsilon_1$ for $1\leq k\leq \kappa $.
\end{enumerate}
We will show that $E$ implies the selected subset is sufficient and the probability of $E$ is large.
\begin{lemma}
    \label{suff_lemma_1}
    Suppose $E$ happens, and for some $1\leq k\leq \kappa$
    \begin{equation}
        \label{suff_lem}
        T_n(S_k) - T_n(S_{k-1}) \leq (1-\varepsilon_2)\delta.
    \end{equation}
    Then $S_{k-1}$ is sufficient.
\end{lemma}

\begin{proof}
    If $k>p$, then there is nothing to prove. Suppose $k\leq p$.
    Since the algorithm chooses $j_k\in \{1,\ldots p\}\backslash S_{k-1}$ to maximize $T_n(S_{k-1}\cup\{j\})$, we have for all $j\in \{1,\ldots p\}\backslash S_{k-1}$:
    $$\begin{aligned}
        T(S_{k-1}\cup\{j\}) - T(S_{k-1}) & \leq T_n(S_{k-1}\cup\{j\}) - T_n(S_{k-1}) + 2\varepsilon_1\\
            &\leq T_n(S_k) - T_n(S_{k-1}) + 2\varepsilon_1\\
               &\leq (1-\varepsilon_2)\delta + 2\varepsilon_1 < \delta.
    \end{aligned}$$
    By the definition of $\delta$, $S_{k-1}$ is sufficient.
\end{proof}

\begin{lemma}
    \label{two_E}
    Event $E$ implies $E_0$.
\end{lemma}
\begin{proof}
    Suppose $E$ holds. If \eqref{suff_lem} holds for some $1\leq k\leq \kappa$, then by Lemma \ref{suff_lemma_1},
    $S_{k-1}$ is sufficient and $E_0$ holds.
    Suppose \eqref{suff_lem} is violated for all $1\leq k\leq \kappa $. Then 
    $$T(S_k) - T(S_{k-1}) \geq T_n(S_k) - T_n(S_{k-1}) - 2\varepsilon_1 > (1-\varepsilon_2)\delta -2\varepsilon_1.$$
    Hence,
    $$\begin{aligned}
        T(S_\kappa) &=\sum_{k=1}^{\kappa} (T(S_k) - T(S_{k-1})) + T(S_0) \geq \kappa \cdot \left((1-\varepsilon_2)\delta -2\varepsilon_1 \right)+ 0 >M,
    \end{aligned}$$
    by the constuction of $\varepsilon_1,\varepsilon_2$. This yields a contradiction since $T(S_\kappa)$ cannot be greater than $M$, the bound of the kernel.
\end{proof}

\begin{lemma}
    Event $E$ implies that $\hat{S}$ is sufficient.
\end{lemma}
\begin{proof}
    If the algorithm stopped after $\kappa = \lfloor \frac{M}{\delta} + 1 \rfloor$, then $S_\kappa \subset \hat{S}$.
    By Lemma \ref{two_E}, $E$ happens implies $E_0$ happens, i.e., $S_\kappa$ being sufficient.
    Hence $\hat{S}$ is also sufficient.

    If the algorithm stopped at $k< \kappa$, by the stopping rule
    $$T_n(S_{k+1}) < T_n (S_{k}).$$
    Hence \eqref{suff_lem} holds, and by Lemma \ref{suff_lemma_1}, $S_{k}$ is sufficient.
\end{proof}

\begin{lemma}
    There exist $L_1,L_2$ depending only on $\alpha,\beta_1,\beta_2,\gamma,\{C_i\}_{i=1}^6,d,M,\delta$ such that
    $$\p(E)\geq 1-L_1 p^{\kappa} e^{-L_2 n}.$$
\end{lemma}

\begin{proof}
    By \cite[Equation (C.38)]{deb2020kernel} with the notation therein and \cite[Section 5.1]{deb2020kernel}, there exists $\xi_1,\xi_2,\xi_3> 0$ depending on $\alpha,\beta_1,\beta_2,\gamma,\{C_i\}_{i=1}^6,d,M$ such that for any $S$ of size $\leq \kappa$
    $$|\E T_n(S) - T(S)|\leq \xi_1\left(\varepsilon_n^{\beta_2} + \sqrt{\nu_{1,n}} \right)\leq \xi_1\frac{(\log n)^{\xi_2}}{n^{\xi_3}}.$$
    By \cite[Proposition 3.1]{deb2020kernel}, there exists $C^*$ depending on $C_2$ and $M$ such that for any $S$ of size $\leq \kappa$
    $$\p\left( |T_n(S) - \E T(S)|\geq t \right)\leq 2\exp\left(-C^* nt^2 \right).$$
    Hence,
    $$\p \left(|T_n(S) -T(S)| \geq \xi_1\frac{(\log n)^{\xi_2}}{n^{\xi_3}}+ t \right) \leq 2 e^{-C^* nt^2}.$$
    By a union bound:
    $$\p\left(\bigcup_{|S|\leq \kappa} \Big\{|T_n(S)-T(S)| \geq \xi_1\frac{(\log n)^{\xi_2}}{n^{\xi_3}}+t\Big\}\right)\leq 2 p^\kappa e^{-C^*nt^2}.$$
    Let $t=\frac{\varepsilon_1}{2}$. For large $n$, say $n\geq n_0$, $\xi_1\frac{(\log n)^{\xi_2}}{n^{\xi_3}}<\frac{\varepsilon_1}{2}$ so
    $$\p(E)\geq 1-2 p^\kappa e^{-L_2 n}.$$
    We can adjust the constant $2$ (again depending only on $\alpha,\beta_1,\beta_2,\gamma,\{C_i\}_{i=1}^6,d,M,\delta$) so that the above inequality also holds for small $n$.
\end{proof}

Combining the previous two lemmas we obtain the proof of Theorem \ref{thm:var_select}.
\qed

\section{Some Auxiliary Lemmas}\label{sec:techlem}

\begin{lemma}
    \label{multi_normal}
    Suppose $(Y,X,Z)$ is jointly Gaussian. Then, when using the linear kernel, $\rho^2(Y,Z|X)=0$ implies $Y\indep Z|X$.
\end{lemma}
\begin{proof}
    For multivariate Gaussian distribution, the conditional distribution $(Y,Z)|X$ is Gaussian with conditional variance and covariance not depending on the value of $X$.
    Suppose $Y=(Y^{(1)},\cdots, Y^{(d)})$ and
    $$(Y^{(i)},Z)^\top |X=x \sim N\left(\begin{pmatrix}
        \mu_i(x)\\
        \mu_z(x)
    \end{pmatrix},\begin{pmatrix}
        \sigma_i^2 & \beta_i^\top\\
        \beta_i & \Sigma_z
    \end{pmatrix} \right)$$
    where $\mu_i(\cdot)$ and $\mu_z(\cdot)$ are linear functions. By \eqref{eq:LinKernel},
    $$\begin{aligned}
        \rho^2(Y,Z|X)&=\frac{\sum_{i=1}^d \E({\rm Var}[\E(Y^{(i)}|X,Z)|X])}{\E\left[\sum_{i=1}^d{\rm Var}(Y^{(i)}|X)\right]}\\
        &=\frac{\sum_{i=1}^d \E({\rm Var}[\mu_i(X)+\beta_i^\top \Sigma_z^{-1}(Z-\mu_z(X))|X])}{\E[\sum_{i=1}^d\sigma_i^2]}\\
        &=\frac{\sum_{i=1}^d \beta_i^\top \Sigma_z^{-1}\Sigma_z \Sigma_z^{-1}\beta_i}{\sum_{i=1}^d\sigma_i^2}\\
        &=\frac{\sum_{i=1}^d \beta_i^\top  \Sigma_z^{-1}\beta_i}{\sum_{i=1}^d\sigma_i^2}.\\
    \end{aligned}$$
Here $\Sigma_{z}^{-1}$ is regarded as the generalized inverse or Moore-Penrose inverse of $\Sigma_z$.
    We can suppose $Z$ is non-degenerate, otherwise $Y\indep Z|X$ is trivial.
    Without loss of generality, we may further suppose $\Sigma_z$ to be invertible, since otherwise we can perform a bijective linear transformation to reduce the dimensionality of $Z$.
    Now $\Sigma_z^{-1}$ is positive-definite and $\rho^2(Y,Z|X)=0$ implies $\beta_i=0$ for all $i$, which further indicates $Y\indep Z|X$ by the property of multivariate Gaussian distribution.
\end{proof}

\begin{lemma}
    \label{trace_class}
    Suppose $\h_\X$ is separable and $\E [k_\X (X,X)]<\infty$. Then $C_X$ is trace-class.
\end{lemma}
\begin{proof}
    Let $\{e_i\}_{i\ge 1}$ be an orthonormal basis of $\h_\X$. Then,
    $$\begin{aligned}
        \sum_i \langle C_X e_i, e_i \rangle_{\h_\X} &= \sum_i {\rm Var}(e_i(X))\\
        &\leq \sum_i \E [e_i(X)^2]\\
        &=\sum_i \E\left[\langle e_i,k_\X (X,\cdot)\rangle_{\h_\X}^2 \right]\\
        &=\E [\|k_\X (X,\cdot)\|^2_{\h_\X}] <\infty.
    \end{aligned}
    $$
    The last equality follows from Parseval's identity.
\end{proof}
\begin{lemma}\label{finite_support}
    Suppose all other conditions in Lemma \ref{center_CME} hold except for Assumption \ref{assump:CME}. If the support of $X$ is finite and $k_\X$ is characteristic, then Assumption \ref{assump:CME} holds.
\end{lemma}
\begin{proof}
    Suppose that $X$ is supported on $\{x_i\}_{i=1}^m$ and $X$ has probability mass function $\{p(x_i)\}_{i=1}^m$. Then $C_{X}f$ has the following explicit expression:
    $$C_{X}f=\sum_{i=1}^m p(x_i)f(x_i)k_\X(x_i,\cdot) -  \E [f(X)] \left[\sum_{i=1}^m p(x_i)k_\X(x_i,\cdot)\right].$$
    Hence ${\rm ran}\ C_{X}\subset {\rm span}\{k_\X(x_i,\cdot)\}_{i=1}^m$ is finite-dimensional.
    By the decomposition $C_{XY}=C_X^{1/2}V_{XY}C_Y^{1/2}$ \cite{baker1973joint}, we have ${\rm ran}\ C_{XY}\subset {\rm ran}\ C_X^{1/2}={\rm ran}\ C_X$.
    By \cite[Theorem 4.1]{klebanov2019rigorous}, ${\rm ran}\ C_{XY}\subset {\rm ran}\ C_X$ is equivalent to Assumption C in that paper.
    Together with $k_\X$ being characteristic, \cite[Theorem 4.3]{klebanov2019rigorous} implies that the centered CME formula~\eqref{eq:CME} holds.
\end{proof}

\subsection{Proof of Lemma~\ref{same_limit}}\label{pf:same_limit}
    Since 
    $$\begin{aligned}
        \|V_i^{(n)}\|^2 &\leq \|U_i^{(n)}\|^2 + 2\|U_i^{(n)}\|\cdot \|U_i^{(n)}-V_i^{(n)}\| + \|U_i^{(n)}-V_i^{(n)}\|^2,\\
        \|V_i^{(n)}\|^2 &\geq \|U_i^{(n)}\|^2 - 2\|U_i^{(n)}\|\cdot \|U_i^{(n)}-V_i^{(n)}\| + \|U_i^{(n)}-V_i^{(n)}\|^2,\\
    \end{aligned}$$
    it suffices to show that $\frac{1}{n}\sum_{i=1}^n 2\|U_i^{(n)}\|\cdot\|U_i^{(n)}-V_i^{(n)}\| \overset{p}{\to} 0$.
    Let $\delta>0$.
    $$\frac{1}{n}\sum_{i=1}^n 2\|U_i^{(n)}\|\cdot\|U_i^{(n)}-V_i^{(n)}\| \leq \frac{1}{n}\sum_{i=1}^n \left[\delta \|U_i^{(n)}\|^2 + \frac{1}{\delta}\|U_i^{(n)}-V_i^{(n)}\|^2 \right]\overset{p}{\to}\delta U.$$
    Since $\delta>0$ is arbitrary, this concludes the proof. \qed


\end{document}